\documentclass[twoside,11pt,final]{entics} 
\usepackage{enticsmacro}
\usepackage{graphicx}
\usepackage[all]{xy}
\usepackage{rotate}
\usepackage{multirow}
\usepackage{multicol}
\usepackage{comment}
\usepackage{bussproofs}
\usepackage{amssymb}
\usepackage{latexsym}
\usepackage{mathrsfs}
\usepackage{graphics}
\usepackage{comment}
\usepackage{fancyhdr}
\usepackage{comment}
\usepackage{tikz-cd}
\usepackage{proof}
\usepackage{color}
\usepackage[title]{appendix}
\usepackage{cmll}  
\usepackage{newtxmath}
\usepackage[inline]{enumitem}
\usetikzlibrary{positioning}
\usetikzlibrary{matrix}
\usetikzlibrary{calc}
\usetikzlibrary{backgrounds}
\usepackage{stmaryrd}

\newenvironment{bprooftree}
  {\leavevmode\hbox\bgroup}
  {\DisplayProof\egroup}

\makeatletter
\def\namedlabel#1#2{\begingroup
    #2%
    \def\@currentlabel{#2}%
    \phantomsection\label{#1}\endgroup
}
\makeatother

\newcommand{\MALL}{\sf{MALL}}
\newcommand{\LL}{\sf{LL}}

\newcommand{\Pts}{\sf{PtS}}
\newcommand{\Bes}{\sf{BeS}}
\newcommand{\At}{\sf{At}}

\newcommand{\red}[1]{\textcolor{red}{#1}}

\newcommand\ie{\hbox{\textit{i.e.}}}
\newcommand\eg{\hbox{\textit{e.g.}}}

\newcommand\bang{\mathop{!}}
\newcommand\quest{\mathord{?}}
\newcommand\limp{\multimap}
\newcommand\tensor\otimes

\def\conf{MFPS 2025} 	
\volume{5}			
\def\volu{5}			
\def\papno{4}			
\def\mydoi{16888}
\def\lastname{Barroso-Nascimento, Piotrovskaya, Pimentel} 
\def\runtitle{Base-extension Semantics for MALL} 
\def\copynames{V. Barroso-Nascimento, E. Piotrovskaya, }    
\begin{document}
\begin{frontmatter}

\title{A Proof-Theoretic Approach to\\ the Semantics of Classical Linear Logic}
				
 \thanks[ALL]{
 Piotrovskaya is supported by the Engineering and Physical Sciences Research Council grants EP/T517793/1 and EP/W524335/1.
 Barroso-Nascimento and Pimentel are supported by the Leverhulme grant RPG-2024-196. 
 Pimentel has received funding from the European Union's Horizon 2020 research and innovation programme under the Marie Sk\l odowska-Curie grant agreement Number 101007627.}   

   \author{Victor Barroso-Nascimento\thanksref{a}\thanksref{1email}}
   \author{Ekaterina Piotrovskaya\thanksref{a}\thanksref{2email}}
   \author{Elaine Pimentel\thanksref{a}\thanksref{3email}}
   \address[a]{Computer Science Department\\ University College London\\				
    London, UK}  
     \thanks[1email]{Email:  \href{mailto:v.nascimento@ucl.ac.uk} {\texttt{\normalshape
        victorluisbn@gmail.com}}}
   \thanks[2email]{Email: \href{mailto:kate.piotrovskaya.21@ucl.ac.uk} {\texttt{\normalshape
        kate.piotrovskaya.21@ucl.ac.uk}}} 
  \thanks[3email]{Email:  \href{mailto:e.pimentel@ucl.ac.uk} {\texttt{\normalshape
        e.pimentel@ucl.ac.uk}}}
\begin{abstract} 
Linear logic ($\LL$) is a resource-aware, abstract logic programming language that refines both classical and intuitionistic logic. 
Linear logic semantics is typically presented in one of two ways: by associating each formula with the set of all contexts that can be used to prove it (\eg\ phase semantics) or by assigning meaning directly to proofs (\eg\ coherence spaces). 

This work proposes a different perspective on assigning meaning to proofs by adopting a proof-theoretic perspective. 
More specifically, we employ base-extension semantics ($\Bes$) to characterise proofs through the notion of base support.
Recent developments have shown that $\Bes$ is powerful enough to capture proof-theoretic notions in structurally rich logics such as intuitionistic linear logic. In this paper, we extend this framework to the classical case, presenting a proof-theoretic approach to the semantics of the multiplicative-additive fragment of linear logic ($\MALL$).
\end{abstract}
\end{frontmatter}


\section{Introduction}\label{sec:intro}
In model-theoretic semantics~\cite{sep-model-theory}, when giving meaning to a sentence $p$, one generally assigns an {\em interpretation} to determine whether it is true or false. This process may involve adding missing information as, \eg, in Kripke systems for modal logics: Since modalities ``qualify'' the notion of truth, mathematical structures {\em support} the validity of $\Box A$ by checking the validity of $A$ in such structures instead. 

If a particular interpretation $\mathcal{M}$ results in $p$ expressing a true statement, we say that $\mathcal{M}$ is a {\em model} of $p$, or equivalently, that $\mathcal{M}$ {\em satisfies} $p$, which can be symbolically denoted as $\Vdash_\mathcal{M}  p$.  However, it is important to note that asserting ``$p$ is true in $\mathcal{M}$'' is simply a reformulation of the claim that $p$, when understood according to $\mathcal{M}$, is true. In this sense, model-theoretic truth relies on ordinary truth and can always be restated in terms of it.

Proof-theoretic semantics~\cite{pts-91,sep-proof-theoretic-semantics} ($\Pts$), on the other hand, provides an alternative perspective for the meaning of logical operators compared to the viewpoint offered by model-theoretic semantics. In $\Pts$, the concept of {\em truth} is substituted with that of {\em proof}, emphasizing the fundamental nature of proofs as a means through which we gain demonstrative knowledge, particularly in mathematical contexts. This makes $\Pts$ a more adequate approach for comprehending reasoning since it ensures that the meaning of logical operators, such as connectives in logics, is defined based on their usage in inferences.

Base-extension semantics~\cite{Sandqvist2015IL,DBLP:journals/synthese/Schroeder-Heister06} ($\Bes$) is a strand of $\Pts$ where proof-theoretic validity is defined relative to a given collection of inference rules regarding basic formulas of the language. More specifically, in $\Bes$ the characterisation of consequence is given by an inductively defined semantic judgment whose base case is given by provability in an {\em atomic system} (or a {\em base}). 

A base is a collection of rules involving only atomic formulas. The nature of such collection/formulas change depending on the logic considered. For example, in~\cite{Sandqvist2015IL}, atoms are intuitionistic atomic propositions and rules have natural deduction style, \eg

\begin{center}
\begin{bprooftree}
\AxiomC{}
\noLine
\UnaryInfC{}
\noLine
\UnaryInfC{}
\noLine
\UnaryInfC{}
\noLine
\UnaryInfC{}
\noLine
\UnaryInfC{}
\noLine
\UnaryInfC{}
\noLine
\UnaryInfC{$l$}
\UnaryInfC{$r$}
\end{bprooftree}
\qquad
\begin{bprooftree}
\AxiomC{}
\noLine
\UnaryInfC{}
\noLine
\UnaryInfC{}
\noLine
\UnaryInfC{}
\noLine
\UnaryInfC{}
\noLine
\UnaryInfC{}
\noLine
\UnaryInfC{}
\noLine
\UnaryInfC{$r$}
\UnaryInfC{$p$}
\end{bprooftree}
\qquad 
\begin{bprooftree}
\AxiomC{$[l]$}
\noLine
\UnaryInfC{.}
\noLine
\UnaryInfC{.}
\noLine
\UnaryInfC{.}
\noLine
\UnaryInfC{$p$}
\UnaryInfC{$u$}
\end{bprooftree}
\end{center}
One could view these as rules that assign interpretations to atomic sentences, much like how models operate in model-theoretic semantics. 
For example, if $l,r,p, u$ represent the sentences ``{\em We are in London}'', ``{\em It rains all the time}'', ``{\em We must be prudent}'' 
and ``{\em We carry an umbrella}''  respectively, and $\mathcal{B}$ is a base containing the rules above, one can infer that ``{\em We carry an umbrella}'' is  supported by $\mathcal{B}$, denoted by $\Vdash_\mathcal{B}  u$. As usual in semantics, starting from validity-as-deduction in the atomic case, the interpretation of more complex sentences is built compositionally from the meanings of its components, with logical connectives guiding the construction.

As expected, different logics permit various approaches, each imposing its own requirements. For instance, substructural logics -- often described as non-classical systems that omit one or more structural rules of classical logic -- typically require the use of multisets rather than sets of formulas~\cite{DBLP:journals/corr/abs-2402-01982,DBLP:conf/tableaux/GheorghiuGP23}. In contrast, classical systems often require a more refined definition of bases~\cite{DBLP:journals/igpl/Makinson14,Sandqvist}.

In this paper, we explore how $\Bes$ applies to both substructural and classical settings. In the following, we outline the main challenges in developing proof-theoretic semantics for classical linear logic. 

\medskip 

\noindent
{\bf The question of falsity.} In model-theoretic semantics, falsity ($\bot$) is often defined as ``never valid''. For example, in Kripke semantics, this is expressed as
\vspace{-0.7em}
\[
\not\Vdash_\mathcal{M} \bot
\vspace{-0.5em}
\]
This, however, raises the philosophical question of what constitutes the syntactic counterpart to semantic refutability~\cite{DBLP:journals/axioms/Goranko19}. Dummet avoids this problem by treating falsity as the conjunction of all basic sentences~\cite{dummett1991logical}, which is stated in~\cite{Sandqvist2015IL} as
\vspace{-0.7em}
\[
\Vdash_\mathcal{B} \bot \mbox{ iff } \Vdash_\mathcal{B} p \mbox{ for all $p$ atomic.}
\vspace{-0.5em}
\]
Alternatively,  in~\cite{DBLP:journals/corr/abs-2306-03656} the semantic of the logical constant $\bot$ was not defined, but instead $\bot$ was allowed to be manipulated by the atomic rules of the base -- hence being considered as a ``fixed atomic formula''. In this work, we adopt the same approach, which not only circumvents the aforementioned discussion but also enables an elegant presentation of $\Bes$ for the classical substructural case, as discussed next.

\medskip 

\noindent
{\bf Dealing with classical notions of validity.} 
The proof-theoretic essence of $\Bes$, where validity is built on the concept of proofs, presents a challenge: How can classical systems be described within this framework? This question is particularly relevant in the context of natural deduction systems, where inference rules inherently exhibit a constructive nature. In fact, the most common approaches to handling classical proofs often lead to non-harmonic systems~\cite{prawitz1965}, whereas harmonic conservative extensions of intuitionistic natural deduction systems tend to simulate the multiple-conclusion behavior of classical sequent systems~\cite{DBLP:conf/birthday/GabbayG05,DBLP:books/sp/24/PereiraP24,restall}. 

In~\cite{Sandqvist}, Sandqvist proposed an inferential semantic justification for first-order classical logic, thus avoiding reliance on a notion of bivalent truth.
However, as pointed out in~\cite{DBLP:journals/jphil/PiechaSS15} and further discussed in~\cite{DBLP:journals/igpl/Makinson14}, Sandqvist's system lacks robustness in its choice of primitive connectives. Moreover, the proposed solution remains somewhat unsatisfactory from an inferentialist perspective, as it heavily depends on the duality of connectives to describe the entire logical system. In~\cite{DBLP:journals/corr/abs-2306-03656}, we tackle this problem with a different approach: Classical proofs are defined by taking into account an idea advanced by David Hilbert to justify non-constructive proof methods, where the concept of consistency is conceptually prior to that of truth, and in order to prove the truth of a proposition in a given context it suffices to prove its consistency \cite{ConsistencyHillbert61554c58-c869-34f0-b322-2cff263d9ae0,Hilbert1900,Hilbert1979-HILMPL}.

In this paper, we make a great use of allowing $\bot$ in a base, and show an interesting connection between the semantic characterisation of proofs and model-theoretic truth conditions: just like classical models can sometimes be obtained by restricting intuitionistic models, classical proof conditions can be obtained through a very small, uniform restriction on intuitionistic proof conditions. We show that the restriction works even in cases as complex as that of linear logic.

\medskip 

\noindent
{\bf Tackling substructurality.} While classical logic emphasizes {\em truth} and intuitionistic logic emphasizes {\em proofs}, linear logic~\cite{DBLP:journals/tcs/Girard87} ($\LL$) introduces a focus on resources, where ``{\em $\phi$  implies $\psi$}'' is interpreted as ``{\em consume $\phi$ to produce $\psi$}''. This has a substructural nature, since formulas (\ie\ resources) cannot be freely copied or erased anymore. $\LL$ can be also seen as an abstract logic programming language~\cite{DBLP:conf/oopsla/AndreoliP90}, since  it is  sound and non-deterministic complete with respect to the logical interpretation of programs and has a proof-search strategy attached to it~\cite{DBLP:journals/logcom/Andreoli92,DBLP:journals/apal/MillerNPS91}. 

Linear logic semantics is typically presented in two ways~\cite{sep-logic-linear}: by associating each formula with the set of all contexts that can be used to prove it (\eg\ phase semantics~\cite{DBLP:journals/iandc/FagesRS01,DBLP:journals/tcs/Girard87}) or by assigning meaning directly to proofs (\eg\ coherence spaces~\cite{DBLP:journals/tcs/Girard87} and relational semantics~\cite{DBLP:journals/mscs/Falco03,DBLP:journals/tcs/Ehrhard12})\footnote{Other possible approaches are, \eg, Kripke-style semantics~\cite{DBLP:journals/jsyml/AllweinD93,DBLP:journals/iandc/HodasM94}, categorical models~\cite{DBLP:journals/tcs/Abramsky93,DBLP:journals/mscs/Ehrhard93} and game semantics~\cite{DBLP:journals/jsyml/AbramskyJ94,DBLP:conf/lics/AbramskyM99}.}.

Here, we adopt a {\em different perspective} on assigning meaning to proofs~\cite{DBLP:journals/jphil/Ayhan21,DBLP:journals/jphil/dAragona22,DBLP:journals/synthese/DicherP21,pym2025categorical,Stafford2023-STAFAT-7}, by developing a $\Bes$ for the multiplicative-additive fragment of (classical) linear logic ($\MALL$). The central idea is to apply a uniform restriction on the intuitionistic proof conditions: rather than requiring the derivation of an arbitrary atomic proposition $p$, we now consistently demand the construction of a proof of $\bot$, treated as a fixed atom.

For example, in ~\cite{DBLP:conf/tableaux/GheorghiuGP23}, the $\Bes$ semantic clause for the multiplicative conjunction $\otimes$ is stated as:
\begin{itemize}
 \item[$(\otimes)$] $\Vdash^{\Gamma_{\At}}_{\mathcal{B}} \phi \otimes \psi$ iff, for all $\mathcal{C} \supseteq \mathcal{B}$, $p$ atomic and $\Delta_{\At}$, if $\phi,\psi \Vdash^{\Delta_{\At}}_{\mathcal{C}} p$ then $\Vdash^{\Gamma_{\At}, \Delta_{\At}}_{\mathcal{C}} p$;
\end{itemize}
Applying the restriction, it will have the following form (highlighting the use of $\bot$ in red):
\begin{itemize}
 \item[$(\otimes)$] $\Vdash^{\Gamma_{\At}}_{\mathcal{B}} \phi \otimes \psi$ iff, for all $\mathcal{C} \supseteq \mathcal{B}$ and $\Delta_{\At}$, if $\phi,\psi \Vdash^{\Delta_{\At}}_{\mathcal{C}} \red{\bot}$ then $\Vdash^{\Gamma_{\At}, \Delta_{\At}}_{\mathcal{C}} \red{\bot}$;
\end{itemize}
The restriction is as simple as it is illuminating, clarifying the semantic import of atomic quantification as well as its relation to structural operations. 

\medskip

In the following sections, we explore these notions in depth and establish that a natural deduction system for $\MALL$ is sound and complete with respect to our proposed semantics, providing the first $\Bes$ for classical substructural systems.

\section{Multiplicative Additive Linear Logic}\label{sec:MALL}
Classical linear logic~\cite{DBLP:journals/tcs/Girard87} ($\LL$) is a resource-sensitive logic, meaning that formulas are consumed when used in proofs unless explicitly marked with the exponentials $\bang$ and $\quest$. Formulas marked with these exponentials behave \emph{classically}, \ie, they can be contracted (duplicated) and weakened (erased) during proofs.

The propositional connectives of $\LL$ include the additive conjunction $\with$ and disjunction $\oplus$, as well as their multiplicative counterparts, tensor $\otimes$ and par $\parr$, along with their respective units. While the linear implication $\limp$ can be expressed in the classical setting using $\parr$ and negation, we make it explicit here due to the inferentialist perspective adopted in this work.

\par\nobreak\bgroup%
\begin{tikzpicture}[node distance=1ex]
  \node [matrix of math nodes] (gr) {
     \phi, \psi  & \Coloneq &
    \node(a){p}; & \mid & \node(imp){\phi \multimap  \psi}; & \mid & \node(tens){\phi \otimes  \psi}; & \mid & \mathsf{1} & \mid &
    \node(plus){\phi \oplus  \psi}; & \mid & \mathsf{0} &
    \mid & \node(bang){\mathsf{!} \phi}; \\
    & \node[right] {} ; & 
     \node(a'){ }; &  & 
     \node(imp'){ }; & \mid &
    \node(par){\phi \parr  \psi}; & \mid & \node(bot){\bot}; & \mid &
    \node(with){\phi \with  \psi}; & \mid & \node(top){\top}; &
     \mid & \node(qm){\mathsf{?} \phi}; \\
  } ;
  \node at ($(bang.north east)!.5!(qm.south east)+(1.8,0)$) {
    \refstepcounter{equation}
    \label{eq:gram}
  } ;
  \begin{scope}[on background layer]
    \fill[rounded corners,color=green!5!white]
       ($(a.north west) - (0.1,0)$) rectangle ($(a'.south east)+(.25,-.5)$) ;
    \node[align=flush left] at ($(a'.south west)!.5!(a'.south east)+(.15,-.2)$){
      \tiny\scshape \hspace{-1em} atoms
    } ;
      \fill[rounded corners,color=yellow!5!white]
      (imp.north west) rectangle ($(imp'.south east)+(.6,-.5)$) ;
    \node[align=flush left] at ($(imp'.south west)!.5!(imp'.south east)+(.15,-.2)$){
      \tiny\scshape \hspace{-0.85em} implication
    } ;
    \fill[rounded corners,color=blue!5!white]
       (tens.north west) rectangle ($(bot.south east)-(0,.5)$) ;
    \node at ($(par.south west)!.5!(bot.south east)-(0,.2)$) {
      \tiny\scshape multiplicatives
    } ;
    \fill [rounded corners,color=red!5!white]
       (plus.north west) rectangle ($(top.south east)-(0,.5)$) ;
    \node at ($(with.south west)!.5!(top.south east)-(0,.2)$) {
      \tiny\scshape additives
    } ;
    \fill [rounded corners,color=cyan!10!white]
       (bang.north west) rectangle ($(qm.south east)+(.2,-.5)$) ;
    \node[align=flush left] at ($(qm.south west)!.5!(qm.south east)+(.15,-.2)$) {
      \tiny\scshape exp.
    } ;
  \end{scope}
\end{tikzpicture}
\egroup\par\nobreak\noindent
\setcounter{equation}{0}

We will concentrate on the multiplicative-additive fragment of $\LL$, called $\MALL$, and adopt the following notation: we fix a countably infinite set of propositional atoms and call it {\At}; lowercase Latin letters ($p, q$) denote atoms; capital Greek letters with the subscript `$\At$' ($\Gamma_{\At}, \Delta_{\At}$) denote finite multisets of atoms; lowercase Greek letters ($\phi, \psi$) denote formulas; capital Greek letters without the subscript `$\At$' ($\Gamma, \Delta$) denote multisets of formulas; commas between multisets denote multiset union; and $\neg \phi$ is to be read as $\phi \multimap \bot$ for any formula $\phi$.

A  {\em sequent} is a pair $\Gamma\vdash \phi$ in which $\Gamma$ is a multiset of formulas and $\phi$ is a formula in $\MALL$.
The natural deduction inference rules of $\MALL$ in the sequent style presentation are depicted below.
    \begin{center}
    \hspace{-1.5em}
        \noindent\begin{minipage}{0.2\textwidth}\scriptsize
        \begin{prooftree}\label{eq:axiom}
          \def\ScoreOverhang{0.5pt}
              \AxiomC{}
              \RightLabel{Ax}
              \UnaryInfC{$\phi\vdash\phi$}
        \end{prooftree}
        \begin{prooftree}\label{eq:tensor-intro}
        \def\ScoreOverhang{0.5pt}
            \AxiomC{$\Gamma \vdash \phi$}
            \AxiomC{$\Delta \vdash \psi$}
            \RightLabel{$\otimes$I}
            \BinaryInfC{$\Gamma,\Delta \vdash \phi\otimes\psi$}
        \end{prooftree}
        \begin{prooftree}\label{eq:mtop-intro}
        \def\ScoreOverhang{0.5pt}
            \AxiomC{}
            \RightLabel{$1$I}
            \UnaryInfC{$\vdash 1$}
        \end{prooftree}
        \begin{prooftree}\label{eq:and-intro}
        \def\ScoreOverhang{0.5pt}
            \AxiomC{$\Gamma \vdash \phi$}
            \AxiomC{$\Gamma \vdash \psi$}
            \RightLabel{$\with$I}
            \BinaryInfC{$\Gamma \vdash \phi \with \psi$}
        \end{prooftree}
        \end{minipage}
        \noindent\begin{minipage}{0.25\textwidth}\scriptsize
        \begin{prooftree}\label{eq:Raa}
        \def\ScoreOverhang{0.5pt}
            \AxiomC{$\Gamma, \neg \phi \vdash \bot$}
            \RightLabel{Raa}
            \UnaryInfC{$\Gamma \vdash \phi$}
        \end{prooftree}
        \begin{prooftree}\label{eq:tensor-elim}
        \def\ScoreOverhang{0.5pt}
            \AxiomC{$\Gamma \vdash \phi\otimes\psi$}
            \AxiomC{$\Delta, \phi, \psi \vdash \chi$}
            \RightLabel{$\otimes$E}
            \BinaryInfC{$\Gamma, \Delta \vdash \chi$}
        \end{prooftree}
        \begin{prooftree}\label{eq:mtop-elim}
        \def\ScoreOverhang{0.5pt}
            \AxiomC{$\Gamma \vdash \phi$}
            \AxiomC{$\Delta \vdash 1$}
            \RightLabel{$1$E}
            \BinaryInfC{$\Gamma,\Delta \vdash \phi$}
        \end{prooftree}
        \begin{prooftree}\label{eq:and-elim}
        \def\ScoreOverhang{0.5pt}
            \AxiomC{$\Gamma \vdash \phi_1 \with \phi_2$}
            \RightLabel{$\with\text{E}_i$}
            \UnaryInfC{$\Gamma \vdash \phi_i$}
        \end{prooftree}
        \end{minipage}
        \noindent\begin{minipage}{0.2\textwidth}\scriptsize
        \begin{prooftree}\label{eq:atop-intro}
        \def\ScoreOverhang{0.5pt}
            \AxiomC{}
            \RightLabel{$\top$I}
            \UnaryInfC{$\Gamma \vdash \top$}
        \end{prooftree}
        \begin{prooftree}\label{eq:implication-intro}
          \def\ScoreOverhang{0.5pt}
              \AxiomC{$\Gamma,\phi \vdash \psi$}
              \RightLabel{$\multimap$I}
              \UnaryInfC{$\Gamma \vdash \phi \multimap \psi$}
        \end{prooftree}
        \begin{prooftree}\label{eq:par-intro}
        \def\ScoreOverhang{0.5pt}
            \AxiomC{$\Gamma, \neg \phi, \neg \psi \vdash \bot$}
            \RightLabel{$\parr\text{I}$}
            \UnaryInfC{$\Gamma \vdash \phi \parr \psi$}
        \end{prooftree}
        \begin{prooftree}\label{eq:or-intro}
        \def\ScoreOverhang{0.5pt}
            \AxiomC{$\Gamma \vdash \phi_i$}
            \RightLabel{$\oplus\text{I}_i$}
            \UnaryInfC{$\Gamma \vdash \phi_1 \oplus \phi_2$}
        \end{prooftree}
        \end{minipage}
        \noindent\begin{minipage}{0.3\textwidth}\scriptsize
        \begin{prooftree}\label{eq:abot-elim}
        \def\ScoreOverhang{0.5pt}
            \AxiomC{$\Gamma \vdash 0$}
            \RightLabel{$0$E}
            \UnaryInfC{$\Gamma,\Delta \vdash \phi$}
        \end{prooftree}
        \begin{prooftree}\label{eq:implication-elim}
          \def\ScoreOverhang{0.5pt}
              \AxiomC{$\Gamma \vdash \phi \multimap \psi$}
              \AxiomC{$\Delta \vdash \phi$}
              \RightLabel{$\multimap$E}
              \BinaryInfC{$\Gamma,\Delta \vdash \psi$}
        \end{prooftree}
        \begin{prooftree}
        \def\ScoreOverhang{0.5pt}
            \AxiomC{$\Gamma \vdash \phi \parr \psi$}
            \AxiomC{$\Delta ,\phi \vdash \bot$}
            \AxiomC{$\Theta ,\psi\vdash \bot$} 
            \RightLabel{$\parr$E\label{eq:par-elim}}
            \TrinaryInfC{$\Gamma, \Delta, \Theta \vdash \bot$}
        \end{prooftree}
        \begin{prooftree}
        \def\ScoreOverhang{0.5pt}
            \AxiomC{$\Gamma \vdash \phi \oplus \psi$}
            \AxiomC{$\Delta,\phi \vdash \chi$}
            \AxiomC{$\Delta,\psi \vdash \chi$}
            \RightLabel{$\oplus$E\label{eq:or-elim}}
            \TrinaryInfC{$\Gamma, \Delta \vdash \chi$}
        \end{prooftree}
        \end{minipage}
    \end{center}

Given a proof system $\mathcal{P}$, a {\em $\mathcal{P}$-derivation} is a finite rooted tree with nodes labeled by sequents, axioms at the top nodes, and where each node is connected with the (immediate) successor nodes (if any) according to the inference rules above. A sequent $\Gamma\vdash \phi$ is {\em derivable} in $\mathcal{P}$, notation $\Gamma\vdash_{\mathcal{P}} \phi$, if and only if there is a derivation of $\Gamma\vdash \phi$ in  $\mathcal{P}$. 

\begin{example} For any $\MALL$-formula $\phi$,
$\neg\neg\phi\vdash_{\MALL}\phi$. Consider, \eg, the derivation
\[
\infer[\small \mbox{Raa}]{\neg\neg\phi\vdash\phi}
{\infer[\small \limp \mbox{E}]{\neg\neg\phi,\neg\phi\vdash\bot}
{\infer[\small \mbox{Ax}]{\neg\neg\phi\vdash\neg\phi\limp\bot}{}
&
\infer[\small \mbox{Ax}]{\neg\phi\vdash\neg\phi}{}}}
\]
Moreover, {\em any derivation} of this sequent has at least one instance of the rule Raa. In particular, $\neg\neg\phi\vdash\phi$ {\em is not provable} in intuitionistic $\MALL$~\cite{DBLP:journals/tcs/Girard87}. 

The following substitution rule is admissible in $\MALL$~\cite{DBLP:journals/igpl/MartinsM04,DBLP:journals/aml/Mints98,DBLP:journals/aml/Negri02}:
\begin{prooftree}
    \AxiomC{$\Gamma \vdash \phi$}
    \AxiomC{$\Delta, \phi \vdash \psi$}
    \RightLabel{\scriptsize{Subs}}
    \BinaryInfC{$\Gamma, \Delta \vdash \psi$}
\end{prooftree}
This rule
represents the traditional natural deduction operation of {\em composition}.
\end{example}

\section{Base-extension Semantics}\label{sec:BeS}
$\Bes$ is founded on an inductively defined judgment called {\em support}, which mirrors the syntactic structure of formulas. The inductive definition begins with a base case: the support of atomic propositions is determined by derivability in a given {\em base} -- a specified collection of inference rules that govern atomic propositions. Sandqvist~\cite{Sandqvist2015IL} introduced a sound and complete formulation of $\Bes$ for Intuitionistic Propositional Logic (IPL).

In this work, we adopt Sandqvist's~\cite{Sandqvist2015IL} terminology, adapting it to the linear logic setting as presented in~\cite{DBLP:journals/corr/abs-2402-01982,DBLP:conf/tableaux/GheorghiuGP23}. Additionally, we refine the framework to represent atomic rules in a tree-based, sequent-style format, aligning it more closely with standard proof-theoretic presentations.

\subsection{Atomic Derivability}

The $\Bes$ begins by defining derivability in a {\em base}. We use, as does Sandqvist, systems containing rules over basic sentences for the semantical analysis. Unlike Sandqvist, we use rules that are more in line with sequent calculus definitions, and also allow the logical constant $\bot$ to be manipulated by rules -- we will abuse the notation and write $\At$ for the set of atomic formulas together with $\bot$.

\begin{definition}[Base]
 An {\em atomic system} (a.k.a. a {\em base}) $\mathcal{B}$ is a set of atomic rules of the form
\begin{prooftree}
\AxiomC{$\Gamma_{\At}^1 \vdash p^{1}$}
\AxiomC{$\ldots$}
\AxiomC{$\Gamma_{\At}^n \vdash p^{n}$}
\TrinaryInfC{$\Delta_{\At} \vdash q$}
\end{prooftree}

\noindent
which is closed under rules of the following shape for all $p$, $r$, $\Gamma_{\At}$ and $\Pi_{\At}$:

\begin{prooftree}
\AxiomC{}
\noLine
\UnaryInfC{}
\RightLabel{\scriptsize{Ax}}
\UnaryInfC{$p \vdash p$}
\qquad
\DisplayProof
\qquad
\AxiomC{$\Gamma_{\At} \vdash p$}
\AxiomC{$\Pi_{\At}, p \vdash r$}
\RightLabel{\scriptsize{Subs}}
\BinaryInfC{$\Gamma_{\At}, \Pi_{\At} \vdash r$}
\end{prooftree}
\end{definition}

\begin{definition}[Extensions]
An atomic system $\mathcal{C}$ is an {\em extension} of an atomic system $\mathcal{B}$ (written $\mathcal{C}\supseteq\mathcal{B}$), if $\mathcal{C}$ results from adding a (possibly empty) set of atomic rules to $\mathcal{B}$.
\end{definition}

\begin{definition}[Deducibility] For every base $\mathcal{B}$, the relation $\vdash_{\mathcal{B}}$ is defined as follows:

\begin{enumerate}
    \item If $\Gamma_{\At} \vdash p$ is the conclusion of an axiomatic rule in $\mathcal{B}$, then $\Gamma_{\At } \vdash_\mathcal{B} p$ holds;
    \item Assume $\mathcal{B}$ contains a non-axiomatic rule with the following shape:

\begin{prooftree}
\AxiomC{$\Gamma_{\At}^1 \vdash p^{1}$}
\AxiomC{$\ldots$}
\AxiomC{$\Gamma_{\At}^n \vdash p^{n}$}
\TrinaryInfC{$\Delta_{\At} \vdash q$}
\end{prooftree}
Then, if $\Gamma_{\At}^{i} \vdash_{\mathcal{B}} p^{i}$ holds for all $1 \leq i \leq n$, $\Delta_{\At} \vdash_{\mathcal{B}} q$ also holds.

\end{enumerate}

\end{definition}
The deducibility relation $\vdash_\mathcal{B}$ coincides with the usual notion in the system of natural deduction consisting of just the rules in $\mathcal{B}$, that is, $p^1, \dots , p^{n} \vdash_\mathcal{B} q$ iff there exists a deduction with the rules of $\mathcal{B}$ whose conclusion is $\{p^{1}, \dots, p^{n}\} \vdash q$.
\begin{example}
Let $l=$ {\em We are in London}, $r=$ {\em It rains all the time}, $p=$ {\em We must be prudent},  
$u=$ {\em We carry an umbrella}   and $\mathcal{B}$ is a base containing the following rules
\[
\infer{l \vdash r}{}\qquad\infer{r \vdash p}{}\qquad\infer{\vdash u}{{l\vdash p}}
\]
together with all instances of Ax and Subs. Then the following is a deduction showing $\vdash_{\mathcal{B}}u$:

\begin{prooftree}
    \AxiomC{}
    \UnaryInfC{$l \vdash r$}
        \AxiomC{}
    \UnaryInfC{$r \vdash p$}
         \RightLabel{\scriptsize{Subs}}
    \BinaryInfC{$l \vdash p$}
    \UnaryInfC{$\vdash u$}
\end{prooftree}
\end{example}

Note that in our atomic rules, both the multiset $\Gamma_{\At}$ and the atom $p$ in a sequent $\Gamma_{\At} \vdash p$ are fixed. This means that, given a derivation concluding $\Delta_{\At} \vdash p$, we cannot apply a rule whose premise is $\Gamma_{\At} \vdash p$ unless $\Gamma_{\At} = \Delta_{\At}$. This stands in contrast to the approach in~\cite{DBLP:journals/corr/abs-2402-01982}, where the contexts in atomic rules are left unconstrained. By fixing the context, we obtain a clean, tree-style representation that aligns naturally with the structure of multiplicative and additive rules, as illustrated below:

\begin{prooftree}
\AxiomC{$\Gamma_{\At} \vdash p$}
\AxiomC{$\Delta_{\At} \vdash q$}
\BinaryInfC{$\Gamma_{\At}, \Delta_{\At} \vdash r$}
\qquad
\DisplayProof
\qquad
\AxiomC{$\Gamma_{\At} \vdash p$}
\AxiomC{$\Gamma_{\At} \vdash q$}
\BinaryInfC{$\Gamma_{\At} \vdash r$}
\end{prooftree}
The price to pay is twofold:
(i) Since the multisets appearing in the conclusion are not necessarily related to the multisets appearing on premises, contexts can no longer be taken to simply fulfill the role of tracking dependencies of a derivation, making our rules closer to pure sequent calculus than to (sequent-style) natural deduction; (ii) explicit inclusion of the structural rules Ax and Subs becomes necessary for bases to behave properly. 

These are easy trades, since (i) tackles the criticism that $\Bes$ is too ``natural deduction driven''~\cite{DBLP:journals/synthese/DicherP21}, while (ii) only makes explicit what the {\em App} rule in~\cite{DBLP:journals/corr/abs-2402-01982,DBLP:conf/tableaux/GheorghiuGP23,Sandqvist2015IL} hides on the treatment of arbitrary contexts. 

We also use a special notation for the \textit{structural base}:

\begin{definition}[Structural base]\label{def:structural}
The {\em structural base} $\mathcal{S}$ is the base that only contains instances of axiom and substitution, \ie\ the smallest of all bases.
\end{definition}

Clearly, we have $\mathcal{B} \supseteq \mathcal{S}$ for all bases $\mathcal{B}$, a property that is useful for many definitions.

\subsection{Semantics}
We will now define  
the support relation,
which is reducible to derivability in $\mathcal{B}$ and its extensions, hence obtaining a semantics defined exclusively in terms of proofs and proof conditions.

\begin{definition}[Support] \label{def:support}
    The \emph{support relation}, denoted as $\Vdash^{\Gamma_{\At}}_{\mathcal{B}}$, is defined as follows, where all multisets of formulas are assumed to be finite:

    \begin{itemize}[itemsep=0.5em, font=\normalfont]
        \item[\namedlabel{eq:supp-at}{(At)}] $\Vdash^{\Gamma_{\At}}_{\mathcal{B}} p$ iff, for all $\mathcal{C} \supseteq \mathcal{B}$ and $\Delta_{\At}$, if $p, \Delta_{\At} \vdash_{\mathcal{C}} \bot$ then $\Gamma_{\At}, \Delta_{\At} \vdash_{\mathcal{C}} \bot$, for $p \in \At$;
     
        \item[\namedlabel{eq:supp-tensor}{$(\otimes)$}] $\Vdash^{\Gamma_{\At}}_{\mathcal{B}} \phi \otimes \psi$ iff, for all $\mathcal{C} \supseteq \mathcal{B}$ and $\Delta_{\At}$, if $\phi,\psi \Vdash^{\Delta_{\At}}_{\mathcal{C}} \bot$ then $\Vdash^{\Gamma_{\At}, \Delta_{\At}}_{\mathcal{C}} \bot$;
    
        \item[\namedlabel{eq:supp-imply}{$(\multimap)$}]  $\Vdash^{\Gamma_{\At}}_{\mathcal{B}} \phi \multimap \psi$ iff, for all $\mathcal{C} \supseteq \mathcal{B}$ and $\Delta_{\At}, \Theta_{\At}$, if $\Vdash^{\Delta_{\At}}_{\mathcal{C}} \phi$ and $\psi \Vdash^{\Theta_{\At}}_{\mathcal{C}} \bot$ then $\Vdash^{\Gamma_{\At}, \Delta_{\At}, \Theta_{\At}}_{\mathcal{C}} \bot$;
    
        \item[\namedlabel{eq:supp-1}{(1)}] $\Vdash^{\Gamma_{\At}}_{\mathcal{B}} 1$ iff, for all $\mathcal{C} \supseteq \mathcal{B}$ and $\Delta_{\At}$, if $\Vdash^{\Delta_{\At}}_{\mathcal{C}} \bot$ then $\Vdash^{\Gamma_{\At}, \Delta_{\At}}_{\mathcal{C}} \bot$; 
    
        \item[\namedlabel{eq:supp-parr}{$(\parr)$}] $\Vdash^{\Gamma_{\At}}_{\mathcal{B}} \phi \parr \psi$ iff, for all $\mathcal{C} \supseteq \mathcal{B}$ and $\Delta_{\At}, \Theta_{\At}$, if $\phi \Vdash^{\Delta_{\At}}_{\mathcal{C}} \bot$ and $\psi \Vdash^{\Theta_{\At}}_{\mathcal{C}} \bot$ then $\Vdash^{\Gamma_{\At}, \Delta_{\At}, \Theta_{\At}}_{\mathcal{C}} \bot$;

        \item[\namedlabel{eq:supp-and}{$(\with)$}] $\Vdash^{\Gamma_{\At}}_{\mathcal{B}} \phi \with \psi$ iff $\Vdash^{\Gamma_{\At}}_{\mathcal{B}} \phi$ and $\Vdash^{\Gamma_{\At}}_{\mathcal{B}} \psi$;
        
        \item[\namedlabel{eq:supp-plus}{$(\oplus)$}] $\Vdash^{\Gamma_{\At}}_{\mathcal{B}} \phi \oplus \psi$ iff, for all $\mathcal{C} \supseteq \mathcal{B}$ and $\Delta_{\At}$, if $\phi \Vdash^{\Delta_{\At}}_{\mathcal{C}} \bot$ and $\psi \Vdash^{\Delta_{\At}}_{\mathcal{C}} \bot$ then $\Vdash^{\Gamma_{\At}, \Delta_{\At}}_{\mathcal{C}} \bot$;

        \item[\namedlabel{eq:supp-top}{$(\top)$}] $\Vdash^{\Gamma_{\At}}_{\mathcal{B}} \top$ for all $\mathcal{B}$ and $\Gamma_{\At}$;

        \item[\namedlabel{eq:supp-0}{(0)}] $\Vdash^{\Gamma_{\At}}_{\mathcal{B}} 0$ iff $\Vdash^{\Gamma_{\At}, \Delta_{\At}}_{\mathcal{B}} \bot$ for all $\Delta_{\At}$;

        \item[\namedlabel{eq:supp-inf}{(Inf)}] $\Gamma \Vdash^{\Theta_{\At}}_{\mathcal{B}} \phi$ iff, for all $\mathcal{C} \supseteq \mathcal{B}$ and all $\Delta^{i}_{\At}$, if $\Gamma = \{\psi^{1}, \dots , \psi^{n}\}$ and $\Vdash^{\Delta^{i}_{\At}}_{\mathcal{C}} \psi^{i}$ for $1 \leq i \leq n$, then $\Vdash^{\Delta^{1}_{\At}, \dots, \Delta^{n}_{\At}, \Theta_{\At}}_{\mathcal{C}} \phi$. 
    \end{itemize}
\end{definition}

\begin{definition}[Validity] \label{def:validity}
    An inference from $\Gamma$ to $\phi$ is \emph{valid}, written as $\Gamma \Vdash \phi$, if $\Gamma \Vdash^{\varnothing}_{\mathcal{B}} \phi$ for all $\mathcal{B}$.
\end{definition}

We read $\Gamma \Vdash^{\Delta_{\At}}_{\mathcal{B}} \phi$ as ``the base $\mathcal{B}$ supports an inference from $\Gamma$ to $\phi$ relative to the multiset $\Delta_{\At}$'' and we write $\Gamma \Vdash_{\mathcal{B}} \phi$ to denote $\Gamma \Vdash^{\varnothing}_{\mathcal{B}} \phi$ for any $\mathcal{B}, \Gamma, \phi$.

We would like to place special emphasis on the clause (At). In~\cite{DBLP:conf/tableaux/GheorghiuGP23}, it is formulated as follows: 
\begin{itemize}[itemsep=0.5em, font=\normalfont] \label{AtAlt1}
    \item[\namedlabel{eq:supp-atAlternative1}{(At')}] $\Vdash^{\Gamma_{\At}}_{\mathcal{B}} p$ iff $\Gamma_{\At} \vdash_{\mathcal{B}} p$;
\end{itemize}
That is {\em atomic support} is reduced to {\em atomic derivability}. Together with the restriction that $\bot$ is excluded from the bases' rules in~\cite{DBLP:conf/tableaux/GheorghiuGP23}, this induces an {\em intuitionistic} flavour in the definition of support for atomic propositions.

This raises the question: How can we capture the {\em classical} notion of an atomic clause within a substructural framework? Our method proceeds in two steps. First, we apply the ``elimination approach'' to rules, interpreting the clause~\ref{eq:supp-atAlternative1} as stating that anything derivable from $p$ must also be derivable from whatever supports $p$:
\begin{itemize}[itemsep=0.5em, font=\normalfont] \label{AtAlt2}
   \item[\namedlabel{eq:supp-atAlternative2}{(At'')}] $\Vdash^{\Gamma_{\At}}_{\mathcal{B}} p$ iff for all $\mathcal{C} \supseteq \mathcal{B}$, $q$ atomic and $\Delta_{\At}$, if $\Delta_{\At},p \vdash_{\mathcal{C}} q$, then $\Delta_{\At},\Gamma_{\At} \vdash_{\mathcal{C}} q$;
\end{itemize}

The second step, which will be used in all the clauses, is to substitute the atomic occurrences in Sandqvist's clauses by $\bot$, obtaining~\ref{eq:supp-at}.

The first step is completely unproblematic since clauses ~\ref{eq:supp-atAlternative1} and ~\ref{eq:supp-atAlternative2} are equivalent as shown next. The second step, in turn, captures the the new perspective to $\Bes$ to classical systems.
\begin{proposition} \label{prop:equivalencenewandoldatdef}
    $\Gamma_{\At} \vdash_{\mathcal{B}} p$ iff for all $\mathcal{C} \supseteq \mathcal{B}$, $q$ atomic and $\Delta_{\At}$, if $\Delta_{\At},p \vdash_{\mathcal{C}} q$, then $\Delta_{\At},\Gamma_{\At} \vdash_{\mathcal{C}} q$.
\end{proposition}

\begin{proof}
    Assume $\Gamma_{\At} \vdash_{\mathcal{B}} p$. Further assume for an arbitrary $\mathcal{C} \supseteq \mathcal{B}$, $q$ and $\Delta_{\At}$, that $\Delta_{\At},p \vdash_{\mathcal{C}} q$. Since deductions are preserved under base extensions, we have $\Gamma \vdash_{\mathcal{C}} p$; an application of Subs yields $\Gamma_{\At}, \Delta_{\At} \vdash_{\mathcal{C}} q$. For the converse, assume that, for all $\mathcal{C} \supseteq \mathcal{B}$, $q$ atomic and $\Delta_{\At}$, if $\Delta_{\At},p \vdash_{\mathcal{C}} q$, then $\Delta_{\At},\Gamma_{\At} \vdash_{\mathcal{C}} q$. Take $\mathcal{C = \mathcal{B}}$, $\Delta_{\At} = \varnothing$ and notice that Ax yields $p \vdash_{\mathcal{B}} p$, so by our assumption we conclude $\Gamma_{\At} \vdash_{\mathcal{B}} p$.
\end{proof}

The next result stresses the fact that support does not directly correspond to derivability.
\begin{lemma} \label{lemma:derivabilityimpliessupport}
    If $\Gamma_{\At} \vdash_{\mathcal{B}} p$ then $\Vdash^{\Gamma_{\At}}_{\mathcal{B}} p$. The converse is false in general.
\end{lemma}

\begin{proof}
    Assume  $\Gamma_{\At} \vdash_{\mathcal{B}} p$. Further assume, for an arbitrary $\mathcal{C} \supseteq \mathcal{B}$ and arbitrary $\Delta_{\At}$, that $p, \Delta_{\At} \vdash_{\mathcal{C}} \bot$. Since deductions are preserved under base extensions, we also have that $\Gamma_{\At} \vdash_{\mathcal{C}} p$. By composing the two deductions, we obtain $\Gamma_{\At}, \Delta_{\At} \vdash_{\mathcal{C}} \bot$. Since $\mathcal{C} \supseteq \mathcal{B}$ such that $p, \Delta_{\At} \vdash_{\mathcal{C}} \bot$ for arbitrary $\Delta_{\At}$, and $\Gamma_{\At}, \Delta_{\At} \vdash_{\mathcal{C}} \bot$, by~\ref{eq:supp-at}, $\Vdash^{\Gamma_{\At}}_{\mathcal{B}} p$.

    On the other hand, let $\Gamma_{\At} = \varnothing$ and $\mathcal{B}$ be the base containing only rules with the following shape for every atomic multiset $\Theta_{\At}$:

    \begin{center}
        \begin{bprooftree}
            \AxiomC{$\Theta_{\At}, p \vdash \bot$}
            \UnaryInfC{$\Theta_{\At} \vdash \bot$}
        \end{bprooftree}  
    \end{center}

The atom $p$ appears only on the right-hand side of instances of Ax and Subs in $\mathcal{B}$, so any deduction concluding $\vdash p$ would necessarily end with an application of one of these rules. However, both require a premise whose consequent is $p$. Since the appropriate instance of Ax has $p \vdash p$ as its conclusion, it cannot serve as the final rule in such a deduction. We therefore conclude that no deduction with conclusion $\vdash p$ exists in $\mathcal{B}$, and thus $\nvdash_{\mathcal{B}} p$.

Now assume, for arbitrary $\mathcal{C} \supseteq \mathcal{B}$ and arbitrary $\Delta_{\At}$, that $p, \Delta_{\At} \vdash_{\mathcal{C}} \bot$. An application of the rule with $\Theta_{\At} = \Delta_{\At}$ yields $\Delta_{\At} \vdash_{\mathcal{C}} \bot$. Since $\mathcal{C} \supseteq \mathcal{B}$ such that $p, \Delta_{\At} \vdash_{\mathcal{C}} \bot$ for arbitrary $\Delta_{\At}$, and $\Delta_{\At} \vdash_{\mathcal{C}} \bot$, by~\ref{eq:supp-at}, $\Vdash_{\mathcal{B}} p$. Hence $\Vdash_{\mathcal{B}} p$, even though $ \nvdash_{\mathcal{B}} p$.
\end{proof}

Interestingly enough, the result holds in the following special case:

\begin{lemma} \label{lemma:bottomisspecial}
    $\Vdash^{\Gamma_{\At}}_{\mathcal{B}} \bot$ if and only if $\Gamma_{\At} \vdash_{\mathcal{B}} \bot$.
\end{lemma}

\begin{proof}
    ($\Rightarrow$): Assume $\Vdash^{\Gamma_{\At}}_{\mathcal{B}} \bot$. Notice that $\bot \vdash_{\mathcal{B}} \bot$; since $\bot$ is an atom, by~\ref{eq:supp-at} we conclude $\Gamma_{\At} \vdash_{\mathcal{B}} \bot$. 
    
    \noindent($\Leftarrow$): Immediate by Lemma~\ref{lemma:derivabilityimpliessupport} with $p = \bot$.
\end{proof}

The following standard result states that the support relation is monotone  w.r.t. bases. 
\begin{lemma}[Monotonicity] \label{lemma:monotonicity}
    If $\Gamma \Vdash^{\Delta_{\At}}_{\mathcal{B}} \phi$ and $\mathcal{C} \supseteq \mathcal{B}$, then $\Gamma \Vdash^{\Delta_{\At}}_{\mathcal{C}} \phi$.
\end{lemma}
\begin{proof}
    Assume $\Gamma \Vdash^{\Delta_{\At}}_{\mathcal{B}} \phi$. Let $\Gamma = \{\psi_1,\dots,\psi_n\}$ and let $\mathcal{C}$ be arbitrary. By~\ref{eq:supp-inf}, one can assume, for an arbitrary $\mathcal{D} \supseteq \mathcal{C}$ and arbitrary multisets $\Theta^{i}_{\At}$, that $\Vdash^{\Theta^{i}_{\At}}_{\mathcal{D}} \psi_i$ for all $\psi_i \in \Gamma$ $(1 \leq i \leq n)$, obtain $\Vdash^{\Theta^{1}_{\At},\dots,\Theta^{n}_{\At}, \Delta_{\At}}_{\mathcal{D}} \phi$ and, hence, conclude $\Gamma \Vdash^{\Delta_{\At}}_{\mathcal{C}} \phi$ by the same definition. It thus remains to prove the case for when $\Gamma = \varnothing$, \ie~that $\Vdash^{\Delta_{\At}}_{\mathcal{B}} \phi$ implies $\Vdash^{\Delta_{\At}}_{\mathcal{C}} \phi$. This is done by induction. Note that in the case of $(\with)$, we require an induction hypothesis stating that \[\Vdash^{\Delta_{\At}}_{\mathcal{B}} \chi \text{ implies } \Vdash^{\Delta_{\At}}_{\mathcal{C}} \chi\] holds true for any proper subformula $\chi$ of $\phi$. It is easy to see that such claim is reasonable given the behaviour of $(\with)$, namely, splitting a formula into its subformulae, combined with the treatment of other connectives below.

    \begin{description}[itemsep=0.5em]
        \item[$\phi = p:$] Assume $\Vdash^{\Delta_{\At}}_{\mathcal{B}} p$ and let $\mathcal{C} \supseteq \mathcal{B}$ be arbitrary. Further assume that, for an arbitrary $\mathcal{D} \supseteq \mathcal{C}$ and arbitrary $\Sigma_{\At}$, $p, \Sigma_{\At} \vdash_{\mathcal{D}} \bot$. Then since $\Vdash^{\Delta_{\At}}_{\mathcal{B}} p$ and $p, \Sigma_{\At} \vdash_{\mathcal{D}} \bot$ by~\ref{eq:supp-at}, $\Delta_{\At}, \Sigma_{\At} \vdash_{\mathcal{D}} \bot$. Since $\mathcal{D} \supseteq \mathcal{C}$ such that $p, \Sigma_{\At} \vdash_{\mathcal{D}} \bot$ for arbitrary $\Sigma_{\At}$, we obtain $\Vdash^{\Delta_{\At}}_{\mathcal{C}} p$ by~\ref{eq:supp-at}. 

        \item[$\phi = \alpha \otimes \beta:$] Assume $\Vdash^{\Delta_{\At}}_{\mathcal{B}} \alpha \otimes \beta$ and let $\mathcal{C} \supseteq \mathcal{B}$ be arbitrary. Further assume that, for an arbitrary $\mathcal{D} \supseteq \mathcal{C}$ and arbitrary $\Sigma_{\At}$, $\alpha, \beta \Vdash^{\Sigma_{\At}}_{\mathcal{D}} \bot$. Then, by~\ref{eq:supp-tensor}, we obtain $\Vdash^{\Delta_{\At}, \Sigma_{\At}}_{\mathcal{D}} \bot$. Since $\mathcal{D} \supseteq \mathcal{C}$ such that $\alpha, \beta \Vdash^{\Sigma_{\At}}_{\mathcal{D}} \bot$ for arbitrary $\Sigma_{\At}$, we obtain $\Vdash^{\Delta_{\At}}_{\mathcal{C}} \alpha \otimes \beta$ by~\ref{eq:supp-tensor}.

        \item[$\phi = \alpha \multimap \beta:$] Assume $\Vdash^{\Delta_{\At}}_{\mathcal{B}} \alpha \multimap \beta$ and let $\mathcal{C} \supseteq \mathcal{B}$ be arbitrary. Further assume that, for an arbitrary $\mathcal{D} \supseteq \mathcal{C}$ and arbitrary $\Sigma_{\At}, \Pi_{\At}$, $\Vdash^{\Sigma_{\At}}_{\mathcal{D}} \alpha$ and $\beta \Vdash^{\Pi_{\At}}_{\mathcal{D}} \bot$. Then, by~\ref{eq:supp-imply}, we obtain $\Vdash^{\Delta_{\At}, \Sigma_{\At}, \Pi_{\At}}_{\mathcal{D}} \bot$. Since $\mathcal{D} \supseteq \mathcal{C}$ such that $\Vdash^{\Sigma_{\At}}_{\mathcal{D}} \alpha$ and $\beta \Vdash^{\Pi_{\At}}_{\mathcal{D}} \bot$ for arbitrary $\Sigma_{\At}, \Pi_{\At}$, we obtain $\Vdash^{\Delta_{\At}}_{\mathcal{C}} \alpha \multimap \beta$ by~\ref{eq:supp-imply}.

        \item[$\phi = 1:$] Assume $\Vdash^{\Delta_{\At}}_{\mathcal{B}} 1$ and let $\mathcal{C} \supseteq \mathcal{B}$ be arbitrary. Further assume that, for an arbitrary $\mathcal{D} \supseteq \mathcal{C}$ and arbitrary $\Sigma_{\At}$, $\Vdash^{\Sigma_{\At}}_{\mathcal{D}} \bot$. Then, by~\ref{eq:supp-1}, we obtain $\Vdash^{\Delta_{\At}, \Sigma_{\At}}_{\mathcal{D}} \bot$. Since $\mathcal{D} \supseteq \mathcal{C}$ such that $\Vdash^{\Sigma_{\At}}_{\mathcal{D}} \bot$ for arbitrary $\Sigma_{\At}$, we obtain $\Vdash^{\Delta_{\At}}_{\mathcal{C}} 1$ by~\ref{eq:supp-1}.

        \item[$\phi = \alpha \parr \beta:$] Assume $\Vdash^{\Delta_{\At}}_{\mathcal{B}} \alpha \parr \beta$ and let $\mathcal{C} \supseteq \mathcal{B}$ be arbitrary.  Further assume that, for an arbitrary $\mathcal{D} \supseteq \mathcal{C}$ and arbitrary $\Sigma_{\At}$, $\Pi_{\At}$, $\alpha \Vdash^{\Sigma_{\At}}_{\mathcal{D}} \bot$ and $\beta \Vdash^{\Pi_{\At}}_{\mathcal{D}} \bot$. Then, by~\ref{eq:supp-parr}, we obtain $\Vdash^{\Delta_{\At}, \Sigma_{\At}, \Pi_{\At}}_{\mathcal{D}} \bot$. Since $\mathcal{D} \supseteq \mathcal{C}$ such that $\alpha \Vdash^{\Sigma_{\At}}_{\mathcal{D}} \bot$ and $\beta \Vdash^{\Pi_{\At}}_{\mathcal{D}} \bot$ for arbitrary $\Sigma_{\At}$, $\Pi_{\At}$, we obtain $\Vdash^{\Delta_{\At}}_{\mathcal{C}} \alpha \parr \beta$ by~\ref{eq:supp-parr}.

        \item[$\phi = \alpha \with \beta:$] Assume $\Vdash^{\Delta_{\At}}_{\mathcal{B}} \alpha \with \beta$ and let $\mathcal{C} \supseteq \mathcal{B}$ be arbitrary. Then, by~\ref{eq:supp-and}, $\Vdash^{\Delta_{\At}}_{\mathcal{B}} \alpha$ and $\Vdash^{\Delta_{\At}}_{\mathcal{B}} \beta$. By induction hypothesis then, $\Vdash^{\Delta_{\At}}_{\mathcal{C}} \alpha$ and $\Vdash^{\Delta_{\At}}_{\mathcal{C}} \beta$, hence, by~\ref{eq:supp-and} again, $\Vdash^{\Delta_{\At}}_{\mathcal{C}} \alpha \with \beta$.

        \item[$\phi = \alpha \oplus \beta:$] Assume $\Vdash^{\Delta_{\At}}_{\mathcal{B}} \alpha \oplus \beta$ and let $\mathcal{C} \supseteq \mathcal{B}$ be arbitrary.  Further assume that, for an arbitrary $\mathcal{D} \supseteq \mathcal{C}$ and arbitrary $\Sigma_{\At}$, $\alpha \Vdash^{\Sigma_{\At}}_{\mathcal{D}} \bot$ and $\beta \Vdash^{\Sigma_{\At}}_{\mathcal{D}} \bot$. Then, by~\ref{eq:supp-plus}, we obtain $\Vdash^{\Delta_{\At}, \Sigma_{\At}}_{\mathcal{D}} \bot$. Since $\mathcal{D} \supseteq \mathcal{C}$ such that $\alpha \Vdash^{\Sigma_{\At}}_{\mathcal{D}} \bot$ and $\beta \Vdash^{\Sigma_{\At}}_{\mathcal{D}} \bot$ for arbitrary $\Sigma_{\At}$, we obtain $\Vdash^{\Delta_{\At}}_{\mathcal{C}} \alpha \oplus \beta$ by~\ref{eq:supp-plus}.

        \item[$\phi = \top:$] Assume $\Vdash^{\Delta_{\At}}_{\mathcal{B}} \top$ and let $\mathcal{C} \supseteq \mathcal{B}$ be arbitrary. By~\ref{eq:supp-top}, $\Vdash^{\Delta_{\At}}_{\mathcal{C}} \top$.

        \item[$\phi = 0:$] Assume $\Vdash^{\Delta_{\At}}_{\mathcal{B}} 0$ and let $\mathcal{C} \supseteq \mathcal{B}$ be arbitrary. Then, by~\ref{eq:supp-0}, we obtain $\Vdash^{\Delta_{\At}, \Sigma_{\At}}_{\mathcal{B}} \bot$ for arbitrary $\Sigma_{\At}$. By Lemma~\ref{lemma:bottomisspecial} then, $\Delta_{\At}, \Sigma_{\At} \vdash_{\mathcal{B}} \bot$. Since deductions are preserved under base extensions, also $\Delta_{\At}, \Sigma_{\At} \vdash_{\mathcal{C}} \bot$, so, by Lemma~\ref{lemma:bottomisspecial} again, $\Vdash^{\Delta_{\At}, \Sigma_{\At}}_{\mathcal{C}} \bot$, hence $\Vdash^{\Delta_{\At}}_{\mathcal{C}} 0$ by~\ref{eq:supp-0}.
    \end{description}
\end{proof}

We can now reformulate the notion of valid inference using Definition~\ref{def:structural}.

\begin{lemma}[Validity]\label{lemma:validity}
    $\Gamma \Vdash \phi$ if and only if $\Gamma \Vdash_{\mathcal{S}} \phi$.
\end{lemma}

\begin{proof}
    ($\Rightarrow$): Since $\Gamma \Vdash \phi$ holds for all bases by Definition~\ref{def:validity}, it in particular holds for $\mathcal{S}$, \ie~$\Gamma \Vdash_{\mathcal{S}} \phi$.
    
    \noindent($\Leftarrow$): Assume $\Gamma \Vdash_{\mathcal{S}} \phi$ and consider any base $\mathcal{B} \supseteq \mathcal{S}$. By Lemma~\ref{lemma:monotonicity}, $\Gamma \Vdash_{\mathcal{B}} \phi$. Since $\mathcal{B}$ is arbitrary, $\Gamma \Vdash \phi$.
\end{proof}

Given an atomic multiset $\Gamma_{\At}$, saying that it supports a formula (possibly with non-empty context $\Delta_{\At}$) should be equivalent to saying that the formula is supported when the same multiset appears as part of the context. In other words, the multiset can be lifted into the context and vice versa. We formalise next this idea, beginning with the case of $\bot$ and then extending it to an arbitrary formula $\phi$ (Lemma~\ref{lemma:interchangeablesets}). 
In the latter case, we omit $\Delta_{\At}$, as the simplified formulation suffices for the proofs that follow.

\begin{lemma} \label{lemma:floatingatom}
    $\Gamma_{\At} \Vdash^{\Delta_{\At}}_{\mathcal{B}} \bot$ if and only if $\Vdash^{\Gamma_{\At}, \Delta_{\At}}_{\mathcal{B}} \bot$.
\end{lemma}

\begin{proof}
    ($\Rightarrow$): Assume $\Gamma_{\At} \Vdash^{\Delta_{\At}}_{\mathcal{B}} \bot$. Let $\Gamma_{\At} = \{p^1,\dots,p^n\}$. We know that $p \vdash_{\mathcal{B}} p$ holds for arbitrary $p$, as well as $\Vdash^{\{p\}}_{\mathcal{B}} p$ by Lemma~\ref{lemma:derivabilityimpliessupport}. Then, $\forall p^i \in \Gamma_{\At}$, it is the case that $\Vdash^{\{p^i\}}_{\mathcal{B}} p^i$. Thus, by~\ref{eq:supp-inf}, from $\Gamma_{\At} \Vdash^{\Delta_{\At}}_{\mathcal{B}} \bot$ and $\Vdash^{\{p^i\}}_{\mathcal{B}} p^i$, $\forall p^i \in \Gamma_{\At}$, we obtain $\Vdash^{\Delta_{\At}, p^i,\dots, p^n}_{\mathcal{B}} \bot$, \ie~$\Vdash^{\Delta_{\At}, \Gamma_{\At}}_{\mathcal{B}} \bot$. 
    
    \noindent($\Leftarrow$): Assume $\Vdash^{\Gamma_{\At}, \Delta_{\At}}_{\mathcal{B}} \bot$. Let $\Gamma_{\At} = \{p^1,\dots,p^n\}$. Further assume that, for an arbitrary $\mathcal{C} \supseteq \mathcal{B}$ and arbitrary multisets $\Theta_{\At}^i$, $\forall p^i \in \Gamma_{\At}$, $\Vdash^{\Theta^i_{\At}}_{\mathcal{C}} p^i$. By Lemma~\ref{lemma:bottomisspecial}, $\Gamma_{\At}, \Delta_{\At} \vdash_{\mathcal{B}} \bot$ and, by monotonicity, $\Gamma_{\At}, \Delta_{\At} \vdash_{\mathcal{C}} \bot$. Then, by~\ref{eq:supp-at}, from $\Vdash^{\Theta^1_{\At}}_{\mathcal{C}} p^1$ and $\Gamma_{\At}, \Delta_{\At} \vdash_{\mathcal{C}} \bot$ we obtain $\Gamma_{\At} \setminus \{p^1\}, \Delta_{\At}, \Theta^1_{\At} \vdash_{\mathcal{C}} \bot$. Repeat for $\Vdash^{\Theta^2_{\At}}_{\mathcal{C}} p^2,\dots,\Vdash^{\Theta^n_{\At}}_{\mathcal{C}} p^n$ to obtain $\Delta_{\At}, \Theta^1_{\At},\dots,\Theta^n_{\At} \vdash_{\mathcal{C}} \bot$. Hence, by Lemma~\ref{lemma:bottomisspecial}, $\Vdash^{\Delta_{\At}, \Theta^1_{\At},\dots,\Theta^n_{\At}}_{\mathcal{C}} \bot$, and since we had chosen arbitrary $\mathcal{C} \supseteq \mathcal{B}$ such that $\forall p^i \in \Gamma_{\At}$, $\Vdash^{\Theta^i_{\At}}_{\mathcal{C}} p^i$ for arbitrary multisets $\Theta^i_{\At}$, by~\ref{eq:supp-inf}, $p^1,\dots,p^n \Vdash^{\Delta_{\At}}_{\mathcal{B}} \bot$, \ie~$\Gamma_{\At} \Vdash^{\Delta_{\At}}_{\mathcal{B}} \bot$.
\end{proof}

\begin{lemma} \label{lemma:interchangeablesets}
    $\Gamma_{\At} \Vdash_{\mathcal{B}} \phi$ if and only if $\Vdash^{\Gamma_{\At}}_{\mathcal{B}} \phi$.
\end{lemma}

\begin{proof}
     ($\Rightarrow$): Assume $\Gamma_{\At} \Vdash_{\mathcal{B}} \phi$. We know that $p \vdash_{\mathcal{B}} p$ holds for arbitrary $p$, as well as $\Vdash^{\{p\}}_{\mathcal{B}} p$ by Lemma~\ref{lemma:derivabilityimpliessupport}. Then, $\forall p_i \in \Gamma_{\At}$ $(1 \leq i \leq n)$, it is the case that $\Vdash^{\{p_i\}}_{\mathcal{B}} p_i$. Thus, by~\ref{eq:supp-inf}, from $\Gamma_{\At} \Vdash_{\mathcal{B}} \phi$ and $\Vdash^{\{p_i\}}_{\mathcal{B}} p_i$, $\forall p_i \in \Gamma_{\At}$, we obtain $\Vdash^{p_i,\dots, p_n}_{\mathcal{B}} \phi$, \ie~$\Vdash^{\Gamma_{\At}}_{\mathcal{B}} \phi$. \\

    \noindent($\Leftarrow$): Assume $\Vdash^{\Gamma_{\At}}_{\mathcal{B}} \phi$. Further assume that, for an arbitrary $\mathcal{C} \supseteq \mathcal{B}$, $\forall p_i \in \Gamma_{\At}$ $(1 \leq i \leq n)$ and arbitrary multisets $\Delta^{i}_{\At}$, $\Vdash^{\Delta^{i}_{\At}}_{\mathcal{C}} p_i$. What follows is the proof by induction. Note that in the case of $(\with)$, we require an induction hypothesis stating that \[\Vdash^{\Gamma_{\At}}_{\mathcal{B}} \psi \text{ implies } \Gamma_{\At} \Vdash_{\mathcal{B}} \psi\] holds true for any proper subformula $\psi$ of $\phi$, as in the proof of Lemma~\ref{lemma:monotonicity}.

    \begin{description}[itemsep=0.5em]
        \item[$\phi = p:$] we have assumed $\Vdash^{\Gamma_{\At}}_{\mathcal{B}} p$. Further assume that for an arbitrary $\mathcal{D} \supseteq \mathcal{C}$ and arbitrary $\Theta_{\At}$, $p, \Theta_{\At} \vdash_{\mathcal{D}} \bot$. From these, by~\ref{eq:supp-at}, we obtain $\Gamma_{\At}, \Theta_{\At} \vdash_{\mathcal{D}} \bot$, and hence $\Gamma_{\At} \Vdash^{\Theta_{\At}}_{\mathcal{D}} \bot$ by Lemmas~\ref{lemma:bottomisspecial},\ref{lemma:floatingatom}. Since $\forall p_i \in \Gamma_{\At}$, $\Vdash^{\Delta^{i}_{\At}}_{\mathcal{C}} p_i$ and $\Gamma_{\At} \Vdash^{\Theta_{\At}}_{\mathcal{D}} \bot$, we obtain $\Vdash^{\Theta_{\At}, \Delta^{1}_{\At},\dots,\Delta^{n}_{\At}}_{\mathcal{D}} \bot$ by~\ref{eq:supp-inf}. By Lemma~\ref{lemma:bottomisspecial}, we obtain $\Theta_{\At}, \Delta^{1}_{\At},\dots,\Delta^{n}_{\At} \vdash_{\mathcal{D}} \bot$, which together with $p, \Theta_{\At} \vdash_{\mathcal{D}} \bot$ gives us $\Vdash^{\Delta^{1}_{\At},\dots,\Delta^{n}_{\At}}_{\mathcal{C}} p$ by~\ref{eq:supp-at}. Thus, since $\Vdash^{\Delta^{1}_{\At},\dots,\Delta^{n}_{\At}}_{\mathcal{C}} p$ and $\mathcal{C} \supseteq \mathcal{B}$ such that $\Vdash^{\Delta^{i}_{\At}}_{\mathcal{C}} p_i$ for arbitrary multisets $\Delta^{i}_{\At}$, we obtain $\Gamma_{\At} \Vdash_{\mathcal{B}} p$ by~\ref{eq:supp-inf}.

        \item[$\phi = \alpha \otimes \beta:$] we have assumed $\Vdash^{\Gamma_{\At}}_{\mathcal{B}} \alpha \otimes \beta$. Further assume that for an arbitrary $\mathcal{D} \supseteq \mathcal{C}$ and arbitrary $\Theta_{\At}$, $\alpha, \beta \Vdash^{\Theta_{\At}}_{\mathcal{D}} \bot$. Then, by~\ref{eq:supp-tensor}, from $\Vdash^{\Gamma_{\At}}_{\mathcal{B}} \alpha \otimes \beta$ and $\alpha, \beta \Vdash^{\Theta_{\At}}_{\mathcal{D}} \bot$ we obtain $\Vdash^{\Gamma_{\At}, \Theta_{\At}}_{\mathcal{D}} \bot$, hence, by Lemma~\ref{lemma:floatingatom}, $\Gamma_{\At} \Vdash^{\Theta_{\At}}_{\mathcal{D}} \bot$. Since $\Gamma_{\At} \Vdash^{\Theta_{\At}}_{\mathcal{D}} \bot$ and $\forall p_i \in \Gamma_{\At}$, $\Vdash^{\Delta^{i}_{\At}}_{\mathcal{C}} p_i$ (thus also $\Vdash^{\Delta^{i}_{\At}}_{\mathcal{D}} p_i$), we obtain $\Vdash^{\Theta_{\At}, \Delta^{1}_{\At},\dots,\Delta^{n}_{\At}}_{\mathcal{D}} \bot$ by~\ref{eq:supp-inf}. Now, since $\alpha, \beta \Vdash^{\Theta_{\At}}_{\mathcal{D}} \bot$ for an arbitrary $\mathcal{D} \supseteq \mathcal{C}$ and $\Vdash^{\Theta_{\At}, \Delta^{1}_{\At},\dots,\Delta^{n}_{\At}}_{\mathcal{D}} \bot$ for arbitrary multisets $\Delta^{i}_{\At}, \Theta_{\At}$, we obtain $\Vdash^{\Delta^{1}_{\At},\dots,\Delta^{n}_{\At}}_{\mathcal{C}} \alpha \otimes \beta$ by~\ref{eq:supp-tensor}. Hence, since $\forall p_i \in \Gamma_{\At}$, $\Vdash^{\Delta^{i}_{\At}}_{\mathcal{C}} p_i$ for an arbitrary $\mathcal{C} \supseteq \mathcal{B}$, we obtain $\Gamma_{\At} \Vdash_{\mathcal{B}} \alpha \otimes \beta$ by~\ref{eq:supp-inf}.

        \item[$\phi = \alpha \multimap \beta:$] we have assumed $\Vdash^{\Gamma_{\At}}_{\mathcal{B}} \alpha \multimap \beta$. Further assume that for an arbitrary $\mathcal{D} \supseteq \mathcal{C}$ and arbitrary $\Theta_{\At}, \Sigma_{\At}$, $\Vdash^{\Theta_{\At}}_{\mathcal{D}} \alpha$ and $\beta \Vdash^{\Sigma_{\At}}_{\mathcal{D}} \bot$. Then, by~\ref{eq:supp-imply}, from $\Vdash^{\Gamma_{\At}}_{\mathcal{B}} \alpha \multimap \beta$, $\Vdash^{\Theta_{\At}}_{\mathcal{D}} \alpha$ and $\beta \Vdash^{\Sigma_{\At}}_{\mathcal{D}} \bot$ we obtain $\Vdash^{\Gamma_{\At}, \Theta_{\At}, \Sigma_{\At}}_{\mathcal{D}} \bot$, hence, by Lemma~\ref{lemma:floatingatom}, $\Gamma_{\At} \Vdash^{\Theta_{\At}, \Sigma_{\At}}_{\mathcal{D}} \bot$. Since $\Gamma_{\At} \Vdash^{\Theta_{\At}, \Sigma_{\At}}_{\mathcal{D}} \bot$ and $\forall p_i \in \Gamma_{\At}$, $\Vdash^{\Delta^{i}_{\At}}_{\mathcal{C}} p_i$ (thus also $\Vdash^{\Delta^{i}_{\At}}_{\mathcal{D}} p_i$), we obtain $\Vdash^{\Theta_{\At}, \Sigma_{\At}, \Delta^{1}_{\At},\dots,\Delta^{n}_{\At}}_{\mathcal{D}} \bot$ by~\ref{eq:supp-inf}. Now, since $\Vdash^{\Theta_{\At}}_{\mathcal{D}} \alpha$ and $\beta \Vdash^{\Sigma_{\At}}_{\mathcal{D}} \bot$ for an arbitrary $\mathcal{D} \supseteq \mathcal{C}$ and $\Vdash^{\Theta_{\At}, \Sigma_{\At}, \Delta^{1}_{\At},\dots,\Delta^{n}_{\At}}_{\mathcal{D}} \bot$ for arbitrary multisets $\Delta^{i}_{\At}, \Theta_{\At}, \Sigma_{\At}$, we obtain $\Vdash^{\Delta^{1}_{\At},\dots,\Delta^{n}_{\At}}_{\mathcal{C}} \alpha \multimap \beta$ by~\ref{eq:supp-imply}. Hence, since $\forall p_i \in \Gamma_{\At}$, $\Vdash^{\Delta^{i}_{\At}}_{\mathcal{C}} p_i$ for an arbitrary $\mathcal{C} \supseteq \mathcal{B}$, we obtain $\Gamma_{\At} \Vdash_{\mathcal{B}} \alpha \multimap \beta$ by~\ref{eq:supp-inf}.

        \item[$\phi = 1:$] we have assumed $\Vdash^{\Gamma_{\At}}_{\mathcal{B}} 1$. Further assume that for an arbitrary $\mathcal{D} \supseteq \mathcal{C}$ and arbitrary $\Theta_{\At}$, $\Vdash^{\Theta_{\At}}_{\mathcal{D}} \bot$. Then, by~\ref{eq:supp-1}, from $\Vdash^{\Gamma_{\At}}_{\mathcal{B}} 1$ and $\Vdash^{\Theta_{\At}}_{\mathcal{D}} \bot$ we obtain $\Vdash^{\Gamma_{\At}, \Theta_{\At}}_{\mathcal{D}} \bot$, hence, by Lemma~\ref{lemma:floatingatom}, $\Gamma_{\At} \Vdash^{\Theta_{\At}}_{\mathcal{D}} \bot$. Since $\Gamma_{\At} \Vdash^{\Theta_{\At}}_{\mathcal{D}} \bot$ and $\forall p_i \in \Gamma_{\At}$, $\Vdash^{\Delta^{i}_{\At}}_{\mathcal{C}} p_i$ (thus also $\Vdash^{\Delta^{i}_{\At}}_{\mathcal{D}} p_i$), we obtain $\Vdash^{\Theta_{\At}, \Delta^{1}_{\At},\dots,\Delta^{n}_{\At}}_{\mathcal{D}} \bot$ by~\ref{eq:supp-inf}. Now, since $\Vdash^{\Theta_{\At}}_{\mathcal{D}} \bot$ for an arbitrary $\mathcal{D} \supseteq \mathcal{C}$ and $\Vdash^{\Theta_{\At}, \Delta^{1}_{\At},\dots,\Delta^{n}_{\At}}_{\mathcal{D}} \bot$ for arbitrary multisets $\Delta^{i}_{\At}, \Theta_{\At}$, we obtain $\Vdash^{\Delta^{1}_{\At},\dots,\Delta^{n}_{\At}}_{\mathcal{C}} 1$ by~\ref{eq:supp-1}. Hence, since $\forall p_i \in \Gamma_{\At}$, $\Vdash^{\Delta^{i}_{\At}}_{\mathcal{C}} p_i$ for an arbitrary $\mathcal{C} \supseteq \mathcal{B}$, we obtain $\Gamma_{\At} \Vdash_{\mathcal{B}} 1$ by~\ref{eq:supp-inf}.

        \item[$\phi = \alpha \parr \beta:$] we have assumed $\Vdash^{\Gamma_{\At}}_{\mathcal{B}} \alpha \parr \beta$. Further assume that for an arbitrary $\mathcal{D} \supseteq \mathcal{C}$ and arbitrary $\Theta_{\At}$, $\Sigma_{\At}$,  $\alpha \Vdash^{\Theta_{\At}}_{\mathcal{D}} \bot$ and $\beta \Vdash^{\Sigma_{\At}}_{\mathcal{D}} \bot$. Then, by~\ref{eq:supp-parr}, from $\Vdash^{\Gamma_{\At}}_{\mathcal{B}} \alpha \parr \beta$ and $\alpha \Vdash^{\Theta_{\At}}_{\mathcal{D}} \bot$ and $\beta \Vdash^{\Sigma_{\At}}_{\mathcal{D}} \bot$ we obtain $\Vdash^{\Gamma_{\At}, \Theta_{\At}, \Sigma_{\At}}_{\mathcal{D}} \bot$, hence, by Lemma~\ref{lemma:floatingatom}, $\Gamma_{\At} \Vdash^{\Theta_{\At}, \Sigma_{\At}}_{\mathcal{D}} \bot$. Since $\Gamma_{\At} \Vdash^{\Theta_{\At}, \Sigma_{\At}}_{\mathcal{D}} \bot$ and $\forall p_i \in \Gamma_{\At}$, $\Vdash^{\Delta^{i}_{\At}}_{\mathcal{C}} p_i$ (thus also $\Vdash^{\Delta^{i}_{\At}}_{\mathcal{D}} p_i$), we obtain $\Vdash^{\Theta_{\At}, \Sigma_{\At}, \Delta^{1}_{\At},\dots,\Delta^{n}_{\At}}_{\mathcal{D}} \bot$ by~\ref{eq:supp-inf}. Now, since $\alpha \Vdash^{\Theta_{\At}}_{\mathcal{D}} \bot$ and $\beta \Vdash^{\Sigma_{\At}}_{\mathcal{D}} \bot$ for an arbitrary $\mathcal{D} \supseteq \mathcal{C}$ and $\Vdash^{\Theta_{\At}, \Sigma_{\At}, \Delta^{1}_{\At},\dots,\Delta^{n}_{\At}}_{\mathcal{D}} \bot$ for arbitrary multisets $\Delta^{i}_{\At}, \Theta_{\At}$. $\Sigma_{\At}$, we obtain $\Vdash^{\Delta^{1}_{\At},\dots,\Delta^{n}_{\At}}_{\mathcal{C}} \alpha \parr \beta$ by~\ref{eq:supp-parr}. Hence, since $\forall p_i \in \Gamma_{\At}$, $\Vdash^{\Delta^{i}_{\At}}_{\mathcal{C}} p_i$ for an arbitrary $\mathcal{C} \supseteq \mathcal{B}$, we obtain $\Gamma_{\At} \Vdash_{\mathcal{B}} \alpha \parr \beta$ by~\ref{eq:supp-inf}.

        \item[$\phi = \alpha \with \beta:$] we have assumed $\Vdash^{\Gamma_{\At}}_{\mathcal{B}} \alpha \with \beta$. By~\ref{eq:supp-and},  $\Vdash^{\Gamma_{\At}}_{\mathcal{B}} \alpha$ and $\Vdash^{\Gamma_{\At}}_{\mathcal{B}} \beta$. By induction hypothesis then, $\Gamma_{\At} \Vdash_{\mathcal{B}} \alpha$ and $\Gamma_{\At} \Vdash_{\mathcal{B}} \beta$. Since $\forall p_i \in \Gamma_{\At}$, $\Vdash^{\Delta^{i}_{\At}}_{\mathcal{C}} p_i$, we obtain $\Vdash^{\Delta^{1}_{\At},\dots,\Delta^{n}_{\At}}_{\mathcal{C}} \alpha$ and $\Vdash^{\Delta^{1}_{\At},\dots,\Delta^{n}_{\At}}_{\mathcal{C}} \beta$, respectively, by~\ref{eq:supp-inf}. Then, by~\ref{eq:supp-and} again, $\Vdash^{\Delta^{1}_{\At},\dots,\Delta^{n}_{\At}}_{\mathcal{C}} \alpha \with \beta$. Hence, since $\forall p_i \in \Gamma_{\At}$, $\Vdash^{\Delta^{i}_{\At}}_{\mathcal{C}} p_i$ for an arbitrary $\mathcal{C} \supseteq \mathcal{B}$, we obtain $\Gamma_{\At} \Vdash_{\mathcal{B}} \alpha \with \beta$ by~\ref{eq:supp-inf}.
        
        \item[$\phi = \alpha \oplus \beta:$] we have assumed $\Vdash^{\Gamma_{\At}}_{\mathcal{B}} \alpha \oplus \beta$. Further assume that for an arbitrary $\mathcal{D} \supseteq \mathcal{C}$ and arbitrary $\Theta_{\At}$, $\alpha \Vdash^{\Theta_{\At}}_{\mathcal{D}} \bot$ and $\beta \Vdash^{\Theta_{\At}}_{\mathcal{D}} \bot$. Then, by~\ref{eq:supp-plus}, from $\Vdash^{\Gamma_{\At}}_{\mathcal{B}} \alpha \oplus \beta$ and $\alpha \Vdash^{\Theta_{\At}}_{\mathcal{D}} \bot$ and $\beta \Vdash^{\Theta_{\At}}_{\mathcal{D}} \bot$ we obtain $\Vdash^{\Gamma_{\At}, \Theta_{\At}}_{\mathcal{D}} \bot$, hence, by Lemma~\ref{lemma:floatingatom}, $\Gamma_{\At} \Vdash^{\Theta_{\At}}_{\mathcal{D}} \bot$. Since $\Gamma_{\At} \Vdash^{\Theta_{\At}}_{\mathcal{D}} \bot$ and $\forall p_i \in \Gamma_{\At}$, $\Vdash^{\Delta^{i}_{\At}}_{\mathcal{C}} p_i$ (thus also $\Vdash^{\Delta^{i}_{\At}}_{\mathcal{D}} p_i$), we obtain $\Vdash^{\Theta_{\At}, \Delta^{1}_{\At},\dots,\Delta^{n}_{\At}}_{\mathcal{D}} \bot$ by~\ref{eq:supp-inf}. Now, since $\alpha \Vdash^{\Theta_{\At}}_{\mathcal{D}} \bot$ and $\beta \Vdash^{\Theta_{\At}}_{\mathcal{D}} \bot$ for an arbitrary $\mathcal{D} \supseteq \mathcal{C}$ and $\Vdash^{\Theta_{\At}, \Delta^{1}_{\At},\dots,\Delta^{n}_{\At}}_{\mathcal{D}} \bot$ for arbitrary multisets $\Delta^{i}_{\At}, \Theta_{\At}$, we obtain $\Vdash^{\Delta^{1}_{\At},\dots,\Delta^{n}_{\At}}_{\mathcal{C}} \alpha \oplus \beta$ by~\ref{eq:supp-plus}. Hence, since $\forall p_i \in \Gamma_{\At}$, $\Vdash^{\Delta^{i}_{\At}}_{\mathcal{C}} p_i$ for an arbitrary $\mathcal{C} \supseteq \mathcal{B}$, we obtain $\Gamma_{\At} \Vdash_{\mathcal{B}} \alpha \oplus \beta$ by~\ref{eq:supp-inf}.

        \item[$\phi = \top:$] we have assumed $\Vdash^{\Gamma_{\At}}_{\mathcal{B}} \top$. By~\ref{eq:supp-top}, $\Vdash^{\Delta^{1}_{\At},\dots,\Delta^{n}_{\At}}_{\mathcal{C}} \top$. Hence, since $\forall p_i \in \Gamma_{\At}$, $\Vdash^{\Delta^{i}_{\At}}_{\mathcal{C}} p_i$ for an arbitrary $\mathcal{C} \supseteq \mathcal{B}$, we obtain $\Gamma_{\At} \Vdash_{\mathcal{B}} \top$ by~\ref{eq:supp-inf}.

        \item[$\phi = 0:$] we have assumed $\Vdash^{\Gamma_{\At}}_{\mathcal{B}} 0$. Then, by~\ref{eq:supp-0}, we obtain $\Vdash^{\Gamma_{\At}, \Theta_{\At}}_{\mathcal{B}} \bot$ for arbitrary $\Theta_{\At}$, hence, by Lemma~\ref{lemma:floatingatom},  $\Gamma_{\At} \Vdash^{\Theta_{\At}}_{\mathcal{B}} \bot$. Since $\Gamma_{\At} \Vdash^{\Theta_{\At}}_{\mathcal{B}} \bot$ and $\forall p_i \in \Gamma_{\At}$, $\Vdash^{\Delta^{i}_{\At}}_{\mathcal{C}} p_i$, we obtain $\Vdash^{\Theta_{\At}, \Delta^{1}_{\At},\dots,\Delta^{n}_{\At}}_{\mathcal{C}} \bot$ by~\ref{eq:supp-inf}. Now, since $\Theta_{\At}$ is arbitrary, we conclude $\Vdash^{\Delta^{1}_{\At},\dots,\Delta^{n}_{\At}}_{\mathcal{C}} 0$ by~\ref{eq:supp-0}. Hence, since $\forall p_i \in \Gamma_{\At}$, $\Vdash^{\Delta^{i}_{\At}}_{\mathcal{C}} p_i$ for an arbitrary $\mathcal{C} \supseteq \mathcal{B}$, we obtain $\Gamma_{\At} \Vdash_{\mathcal{B}} 0$ by~\ref{eq:supp-inf}.
    \end{description}
\end{proof}

So far, we have used~\ref{eq:supp-inf} only to derive expressions with an empty left-hand side -- that is, effectively replacing the entire multiset supporting a formula with atomic multisets in the superscript of the support relation. We now show that this process can be applied partially or sequentially, yielding expressions where some formula remains on the left-hand side.

\begin{lemma} \label{lemma:partialinf}
    If $\Gamma, \phi \Vdash^{\Delta_{\At}}_{\mathcal{B}} \psi$ and, for $\Gamma = \{\alpha^1,\dots,\alpha^n\}$ and an arbitrary $\mathcal{C} \supseteq \mathcal{B}$, $\forall \alpha^i \in \Gamma$ $(1 \leq i \leq n)$ and arbitrary multisets $\Theta^{i}_{\At}$, $\Vdash^{\Theta^{i}_{\At}}_{\mathcal{C}} \alpha^i$, then $\phi \Vdash^{\Delta_{\At}, \Theta^{1}_{\At},\dots,\Theta^{n}_{\At}}_{\mathcal{C}} \psi$.
\end{lemma}

\begin{proof}
    Assume $\Gamma, \phi \Vdash^{\Delta_{\At}}_{\mathcal{B}} \psi$ and, for $\Gamma = \{\alpha^1,\dots,\alpha^n\}$ and an arbitrary $\mathcal{C} \supseteq \mathcal{B}$, $\forall \alpha^i \in \Gamma$ $(1 \leq i \leq n)$ and arbitrary multisets $\Theta^{i}_{\At}$, $\Vdash^{\Theta^{i}_{\At}}_{\mathcal{C}} \alpha^i$. Further assume that, for an arbitrary $\mathcal{D} \supseteq \mathcal{C}$ and arbitrary $\Sigma_{\At}$, $\Vdash^{\Sigma_{\At}}_{\mathcal{D}} \phi$. By monotonicity, also  $\Vdash^{\Theta^{i}_{\At}}_{\mathcal{D}} \alpha^i$ for all $\Theta^{i}_{\At}$ and $\alpha^i \in \Gamma$. Then, by~\ref{eq:supp-inf}, $\Vdash^{\Theta^{1}_{\At},\dots,\Theta^{n}_{\At}, \Sigma_{\At}}_{\mathcal{D}} \psi$. Finally, since $\mathcal{D} \supseteq \mathcal{C}$ such that $\Vdash^{\Sigma_{\At}}_{\mathcal{D}} \phi$ for arbitrary $\Sigma_{\At}$, we obtain $\phi \Vdash^{\Delta_{\At}, \Theta^{1}_{\At},\dots,\Theta^{n}_{\At}}_{\mathcal{C}} \psi$ by~\ref{eq:supp-inf}.
\end{proof}

Another natural property to expect of the support relation is that a formula $\phi$ supports a formula $\psi$ if and only if the inference from $\phi$ to $\psi$ is itself supported. We demonstrate this in the case where $\psi = \bot$, as this result is required for a key step in the soundness proof, and we include a remark addressing the remaining cases.

\begin{lemma} \label{lemma:negatingformula}
    $\phi \Vdash^{\Gamma_{\At}}_{\mathcal{B}} \bot$ if and only if $\Vdash^{\Gamma_{\At}}_{\mathcal{B}} \neg \phi$.
\end{lemma}

\begin{proof}
    ($\Leftarrow$): Assume $\Vdash^{\Gamma_{\At}}_{\mathcal{B}} \neg \phi$, \ie~$\Vdash^{\Gamma_{\At}}_{\mathcal{B}} \phi \multimap \bot$. Further assume that, for an arbitrary $\mathcal{C} \supseteq \mathcal{B}$ and arbitrary $\Theta_{\At}$, $\Vdash^{\Theta_{\At}}_{\mathcal{C}} \phi$. We know that $\bot \vdash_{\mathcal{C}} \bot$, hence, by Lemma~\ref{lemma:bottomisspecial}, $\Vdash^{\{\bot\}}_{\mathcal{C}} \bot$, hence, by Lemma~\ref{lemma:interchangeablesets}, $\bot \Vdash_{\mathcal{C}} \bot$. Now, by~\ref{eq:supp-imply}, from $\Vdash^{\Gamma_{\At}}_{\mathcal{B}} \phi \multimap \bot$, $\Vdash^{\Theta_{\At}}_{\mathcal{C}} \phi$ and $\bot \Vdash_{\mathcal{C}} \bot$ we obtain $\Vdash^{\Gamma_{\At}, \Theta_{\At}}_{\mathcal{C}} \bot$. Since $\mathcal{C} \supseteq \mathcal{B}$ such that $\Vdash^{\Theta_{\At}}_{\mathcal{C}} \phi$ for arbitrary $\Theta_{\At}$, by~\ref{eq:supp-inf}, $\phi \Vdash^{\Gamma_{\At}}_{\mathcal{B}} \bot$. 

    \noindent($\Rightarrow$): Assume $\phi \Vdash^{\Gamma_{\At}}_{\mathcal{B}} \bot$. Further assume that, for an arbitrary $\mathcal{C} \supseteq \mathcal{B}$ and arbitrary $\Theta_{\At}, \Sigma_{\At}$, $\Vdash^{\Theta_{\At}}_{\mathcal{C}} \phi$ and $\bot \Vdash^{\Sigma_{\At}}_{\mathcal{C}} \bot$. By~\ref{eq:supp-inf}, from $\phi \Vdash^{\Gamma_{\At}}_{\mathcal{B}} \bot$ and $\Vdash^{\Theta_{\At}}_{\mathcal{C}} \phi$ we obtain $\Vdash^{\Gamma_{\At}, \Theta_{\At}}_{\mathcal{C}} \bot$. Now, by~\ref{eq:supp-inf} again, from $\bot \Vdash^{\Sigma_{\At}}_{\mathcal{C}} \bot$ and $\Vdash^{\Gamma_{\At}, \Theta_{\At}}_{\mathcal{C}} \bot$ we obtain $\Vdash^{\Gamma_{\At}, \Theta_{\At}, \Sigma_{\At}}_{\mathcal{C}} \bot$. Now, by~\ref{eq:supp-imply}, since $\Vdash^{\Theta_{\At}}_{\mathcal{C}} \phi$ and $\bot \Vdash^{\Sigma_{\At}}_{\mathcal{C}} \bot$ and $\Vdash^{\Gamma_{\At}, \Theta_{\At}, \Sigma_{\At}}_{\mathcal{C}} \bot$ for arbitrary $\Theta_{\At}, \Sigma_{\At}$, we obtain $\Vdash^{\Gamma_{\At}}_{\mathcal{B}} \phi \multimap \bot$, \ie~$\Vdash^{\Gamma_{\At}}_{\mathcal{B}} \neg \phi$.
\end{proof}

\begin{remark}
    It is indeed the case that $\phi \Vdash^{\Gamma_{\At}}_{\mathcal{B}} \psi$ if and only if $\Vdash^{\Gamma_{\At}}_{\mathcal{B}} \phi \multimap \psi$. We only briefly touch on it here as this is not a key result. Nonetheless, to see this, choose an arbitrary $\mathcal{C} \supseteq \mathcal{B}$ such that $\Vdash^{\Delta_{\At}}_{\mathcal{C}} \phi$ and $\psi \Vdash^{\Theta_{\At}}_{\mathcal{C}} \bot$ for arbitrary $\Delta_{\At}, \Theta_{\At}$. Then, by~\ref{eq:supp-inf}, $\Vdash^{\Gamma_{\At}, \Delta_{\At}}_{\mathcal{C}} \psi$, and by~\ref{eq:supp-inf} again, $\Vdash^{\Gamma_{\At}, \Delta_{\At}, \Theta_{\At}}_{\mathcal{C}} \bot$, hence $\Vdash^{\Gamma_{\At}}_{\mathcal{B}} \phi \multimap \psi$ by~\ref{eq:supp-imply}. The other direction is a special case of the upcoming lemma (Lemma~\ref{lemma:genericimplication}): let $\mathcal{C} \supseteq \mathcal{B}$ such that $\Vdash^{\Delta_{\At}}_{\mathcal{C}} \phi$, set $\chi = \psi$ and $\Delta_{\At} = \varnothing$; hence we obtain  $\Vdash^{\Gamma_{\At}, \Delta_{\At}}_{\mathcal{C}} \psi$ and, by~\ref{eq:supp-inf}, $\phi \Vdash^{\Gamma_{\At}}_{\mathcal{B}} \psi$.
\end{remark}

Since we want the support relation to mirror the behaviour of {\MALL}, it is natural to expect that the left-to-right implications in clauses~\ref{eq:supp-tensor},\ref{eq:supp-imply},\ref{eq:supp-1},\ref{eq:supp-plus}, and\ref{eq:supp-0} from Definition~\ref{def:support} should hold for any formula $\psi$ -- not just for $\bot$ -- as these implications resemble the elimination rules for the corresponding connectives. We conclude this section by formalising this observation through a sequence of lemmas, which will also be used in the soundness proof in Section~\ref{sec:soundness}. 

\begin{lemma} \label{lemma:generictensor}
    If $\Vdash^{\Gamma_{\At}}_{\mathcal{B}} \phi \otimes \psi$ and $\phi, \psi \Vdash^{\Delta_{\At}}_{\mathcal{B}} \chi$ then $\Vdash^{\Gamma_{\At}, \Delta_{\At}}_{\mathcal{B}} \chi$.
\end{lemma}
\begin{proof}
    We shall prove the statement inductively. Note that in the case of $(\with)$, we require an induction hypothesis stating that \[\text{if } \Vdash^{\Gamma_{\At}}_{\mathcal{B}} \phi \otimes \psi \text{ and } \phi, \psi \Vdash^{\Delta_{\At}}_{\mathcal{B}} \tau \text{ then } \Vdash^{\Gamma_{\At}, \Delta_{\At}}_{\mathcal{B}} \tau\] holds true for any proper subformula $\tau$ of $\chi$, as in the proof of Lemma~\ref{lemma:monotonicity}.
    
    \begin{description}[itemsep=0.5em]
        \item[$\chi = p:$] Assume $\Vdash^{\Gamma_{\At}}_{\mathcal{B}} \phi \otimes \psi$ and $\phi, \psi \Vdash^{\Delta_{\At}}_{\mathcal{B}} p$ for arbitrary $\Delta_{\At}$. Now assume that, for an arbitrary $\mathcal{C} \supseteq \mathcal{B}$ and arbitrary $\Theta_{\At}$, $p, \Theta_{\At} \vdash_{\mathcal{C}} \bot$. Further assume that, for an arbitrary $\mathcal{D} \supseteq \mathcal{C}$ and arbitrary $\Sigma_{\At}, \Pi_{\At}$, $\Vdash^{\Sigma_{\At}}_{\mathcal{D}} \phi$ and $\Vdash^{\Pi_{\At}}_{\mathcal{D}} \psi$. By monotonicity, $\phi, \psi \Vdash^{\Delta_{\At}}_{\mathcal{D}} p$, hence with $\Vdash^{\Sigma_{\At}}_{\mathcal{D}} \phi$ and $\Vdash^{\Pi_{\At}}_{\mathcal{D}} \psi$, by~\ref{eq:supp-inf}, we obtain $\Vdash^{\Delta_{\At}, \Sigma_{\At}, \Pi_{\At}}_{\mathcal{D}} p$. From $\Vdash^{\Delta_{\At}, \Sigma_{\At}, \Pi_{\At}}_{\mathcal{D}} p$ and $p, \Theta_{\At} \vdash_{\mathcal{C}} \bot$ (thus also $p, \Theta_{\At} \vdash_{\mathcal{D}} \bot$), by~\ref{eq:supp-at}, we obtain $\Delta_{\At}, \Sigma_{\At}, \Pi_{\At}, \Theta_{\At} \vdash_{\mathcal{D}} \bot$. Hence, $\Vdash^{\Delta_{\At}, \Sigma_{\At}, \Pi_{\At}, \Theta_{\At}}_{\mathcal{D}} \bot$ by Lemma~\ref{lemma:bottomisspecial}. Since $\Vdash^{\Delta_{\At}, \Sigma_{\At}, \Pi_{\At}, \Theta_{\At}}_{\mathcal{D}} \bot$ and $\mathcal{D} \supseteq \mathcal{C}$ such that $\Vdash^{\Sigma_{\At}}_{\mathcal{D}} \phi$ and $\Vdash^{\Pi_{\At}}_{\mathcal{D}} \psi$ for arbitrary $\Sigma_{\At}, \Pi_{\At}$, by~\ref{eq:supp-inf}, $\phi, \psi \Vdash^{\Delta_{\At}, \Theta_{\At}}_{\mathcal{C}} \bot$. Since $\Vdash^{\Gamma_{\At}}_{\mathcal{B}} \phi \otimes \psi$ and $\phi, \psi \Vdash^{\Delta_{\At}, \Theta_{\At}}_{\mathcal{C}} \bot$, by~\ref{eq:supp-tensor}, we obtain $\Vdash^{\Gamma_{\At}, \Delta_{\At}, \Theta_{\At}}_{\mathcal{C}} \bot$, hence $\Gamma_{\At}, \Delta_{\At}, \Theta_{\At} \vdash_{\mathcal{C}} \bot$ by Lemma~\ref{lemma:bottomisspecial}. Since $\mathcal{C} \supseteq \mathcal{B}$ such that $p, \Theta_{\At} \vdash_{\mathcal{C}} \bot$ for arbitrary $\Theta_{\At}$, and $\Gamma_{\At}, \Delta_{\At}, \Theta_{\At} \vdash_{\mathcal{C}} \bot$, by~\ref{eq:supp-at}, $\Vdash^{\Gamma_{\At}, \Delta_{\At}}_{\mathcal{B}} p$.

        \item[$\chi = \alpha \otimes \beta:$] Assume $\Vdash^{\Gamma_{\At}}_{\mathcal{B}} \phi \otimes \psi$ and $\phi, \psi \Vdash^{\Delta_{\At}}_{\mathcal{B}} \alpha \otimes \beta$ for arbitrary $\Delta_{\At}$. Now assume that, for an arbitrary $\mathcal{C} \supseteq \mathcal{B}$ and arbitrary $\Theta_{\At}$, $\alpha, \beta \Vdash^{\Theta_{\At}}_{\mathcal{C}} \bot$. Further assume that, for an arbitrary $\mathcal{D} \supseteq \mathcal{C}$ and arbitrary $\Sigma_{\At}, \Pi_{\At}$, $\Vdash^{\Sigma_{\At}}_{\mathcal{D}} \phi$ and $\Vdash^{\Pi_{\At}}_{\mathcal{D}} \psi$. By monotonicity, $\phi, \psi \Vdash^{\Delta_{\At}}_{\mathcal{D}} \alpha \otimes \beta$, hence with $\Vdash^{\Sigma_{\At}}_{\mathcal{D}} \phi$ and $\Vdash^{\Pi_{\At}}_{\mathcal{D}} \psi$, by~\ref{eq:supp-inf}, we obtain $\Vdash^{\Delta_{\At}, \Sigma_{\At}, \Pi_{\At}}_{\mathcal{D}} \alpha \otimes \beta$. From $\Vdash^{\Delta_{\At}, \Sigma_{\At}, \Pi_{\At}}_{\mathcal{D}} \alpha \otimes \beta$ and $\alpha, \beta \Vdash^{\Theta_{\At}}_{\mathcal{C}} \bot$ (thus also $\alpha, \beta \Vdash^{\Theta_{\At}}_{\mathcal{D}} \bot$), by~\ref{eq:supp-tensor}, we obtain $\Vdash^{\Delta_{\At}, \Sigma_{\At}, \Pi_{\At}, \Theta_{\At}}_{\mathcal{D}} \bot$. Since $\Vdash^{\Delta_{\At}, \Sigma_{\At}, \Pi_{\At}, \Theta_{\At}}_{\mathcal{D}} \bot$ and $\mathcal{D} \supseteq \mathcal{C}$ such that $\Vdash^{\Sigma_{\At}}_{\mathcal{D}} \phi$ and $\Vdash^{\Pi_{\At}}_{\mathcal{D}} \psi$ for arbitrary $\Sigma_{\At}, \Pi_{\At}$, by~\ref{eq:supp-inf}, $\phi, \psi \Vdash^{\Delta_{\At}, \Theta_{\At}}_{\mathcal{C}} \bot$. Since $\Vdash^{\Gamma_{\At}}_{\mathcal{B}} \phi \otimes \psi$ and $\phi, \psi \Vdash^{\Delta_{\At}, \Theta_{\At}}_{\mathcal{C}} \bot$, by~\ref{eq:supp-tensor}, we obtain $\Vdash^{\Gamma_{\At}, \Delta_{\At}, \Theta_{\At}}_{\mathcal{C}} \bot$. Since $\mathcal{C} \supseteq \mathcal{B}$ such that $\alpha, \beta \Vdash^{\Theta_{\At}}_{\mathcal{C}} \bot$ for arbitrary $\Theta_{\At}$, and $\Vdash^{\Gamma_{\At}, \Delta_{\At}, \Theta_{\At}}_{\mathcal{C}} \bot$, by~\ref{eq:supp-tensor}, $\Vdash^{\Gamma_{\At}, \Delta_{\At}}_{\mathcal{B}} \alpha \otimes \beta$.

        \item[$\chi = \alpha \multimap \beta:$] Assume $\Vdash^{\Gamma_{\At}}_{\mathcal{B}} \phi \otimes \psi$ and $\phi, \psi \Vdash^{\Delta_{\At}}_{\mathcal{B}} \alpha \multimap \beta$ for arbitrary $\Delta_{\At}$. Now assume that, for an arbitrary $\mathcal{C} \supseteq \mathcal{B}$ and arbitrary $\Theta_{\At}, \Omega_{\At}$, $\Vdash^{\Theta_{\At}}_{\mathcal{C}} \alpha$ and $\beta \Vdash^{\Omega_{\At}}_{\mathcal{C}} \bot$. Further assume that, for an arbitrary $\mathcal{D} \supseteq \mathcal{C}$ and arbitrary $\Sigma_{\At}, \Pi_{\At}$, $\Vdash^{\Sigma_{\At}}_{\mathcal{D}} \phi$ and $\Vdash^{\Pi_{\At}}_{\mathcal{D}} \psi$. By monotonicity, $\phi, \psi \Vdash^{\Delta_{\At}}_{\mathcal{D}} \alpha \multimap \beta$, hence with $\Vdash^{\Sigma_{\At}}_{\mathcal{D}} \phi$ and $\Vdash^{\Pi_{\At}}_{\mathcal{D}} \psi$, by~\ref{eq:supp-inf}, we obtain $\Vdash^{\Delta_{\At}, \Sigma_{\At}, \Pi_{\At}}_{\mathcal{D}} \alpha \multimap \beta$. From $\Vdash^{\Delta_{\At}, \Sigma_{\At}, \Pi_{\At}}_{\mathcal{D}} \alpha \multimap \beta$, $\Vdash^{\Theta_{\At}}_{\mathcal{C}} \alpha$ and $\beta \Vdash^{\Omega_{\At}}_{\mathcal{C}} \bot$ (thus also $\Vdash^{\Theta_{\At}}_{\mathcal{D}} \alpha$ and $\beta \Vdash^{\Omega_{\At}}_{\mathcal{D}} \bot$), by~\ref{eq:supp-imply}, we obtain $\Vdash^{\Delta_{\At}, \Sigma_{\At}, \Pi_{\At}, \Theta_{\At}, \Omega_{\At}}_{\mathcal{D}} \bot$. Since $\Vdash^{\Delta_{\At}, \Sigma_{\At}, \Pi_{\At}, \Theta_{\At}, \Omega_{\At}}_{\mathcal{D}} \bot$ and $\mathcal{D} \supseteq \mathcal{C}$ such that $\Vdash^{\Sigma_{\At}}_{\mathcal{D}} \phi$ and $\Vdash^{\Pi_{\At}}_{\mathcal{D}} \psi$ for arbitrary $\Sigma_{\At}, \Pi_{\At}$, by~\ref{eq:supp-inf}, $\phi, \psi \Vdash^{\Delta_{\At}, \Theta_{\At}, \Omega_{\At}}_{\mathcal{C}} \bot$. Since $\Vdash^{\Gamma_{\At}}_{\mathcal{B}} \phi \otimes \psi$ and $\phi, \psi \Vdash^{\Delta_{\At}, \Theta_{\At}, \Omega_{\At}}_{\mathcal{C}} \bot$, by~\ref{eq:supp-tensor}, we obtain $\Vdash^{\Gamma_{\At}, \Delta_{\At}, \Theta_{\At}, \Omega_{\At}}_{\mathcal{C}} \bot$. Since $\mathcal{C} \supseteq \mathcal{B}$ such that $\Vdash^{\Theta_{\At}}_{\mathcal{C}} \alpha$ and $\beta \Vdash^{\Omega_{\At}}_{\mathcal{C}} \bot$ for arbitrary  $\Theta_{\At}, \Omega_{\At}$, and $\Vdash^{\Gamma_{\At}, \Delta_{\At}, \Theta_{\At}, \Omega_{\At}}_{\mathcal{C}} \bot$, by~\ref{eq:supp-imply}, $\Vdash^{\Gamma_{\At}, \Delta_{\At}}_{\mathcal{B}} \alpha \multimap \beta$.

        \item[$\chi = 1:$] Assume $\Vdash^{\Gamma_{\At}}_{\mathcal{B}} \phi \otimes \psi$ and $\phi, \psi \Vdash^{\Delta_{\At}}_{\mathcal{B}} 1$ for arbitrary $\Delta_{\At}$. Now assume that, for an arbitrary $\mathcal{C} \supseteq \mathcal{B}$ and arbitrary $\Theta_{\At}$, $\Vdash^{\Theta_{\At}}_{\mathcal{C}} \bot$. Further assume that, for an arbitrary $\mathcal{D} \supseteq \mathcal{C}$ and arbitrary $\Sigma_{\At}, \Pi_{\At}$, $\Vdash^{\Sigma_{\At}}_{\mathcal{D}} \phi$ and $\Vdash^{\Pi_{\At}}_{\mathcal{D}} \psi$. By monotonicity, $\phi, \psi \Vdash^{\Delta_{\At}}_{\mathcal{D}} 1$, hence with $\Vdash^{\Sigma_{\At}}_{\mathcal{D}} \phi$ and $\Vdash^{\Pi_{\At}}_{\mathcal{D}} \psi$, by~\ref{eq:supp-inf}, we obtain $\Vdash^{\Delta_{\At}, \Sigma_{\At}, \Pi_{\At}}_{\mathcal{D}} 1$. From $\Vdash^{\Delta_{\At}, \Sigma_{\At}, \Pi_{\At}}_{\mathcal{D}} 1$ and $\Vdash^{\Theta_{\At}}_{\mathcal{C}} \bot$ (thus also $\Vdash^{\Theta_{\At}}_{\mathcal{D}} \bot$), by~\ref{eq:supp-1}, we obtain $\Vdash^{\Delta_{\At}, \Sigma_{\At}, \Pi_{\At}, \Theta_{\At}}_{\mathcal{D}} \bot$. Since $\Vdash^{\Delta_{\At}, \Sigma_{\At}, \Pi_{\At}, \Theta_{\At}}_{\mathcal{D}} \bot$ and $\mathcal{D} \supseteq \mathcal{C}$ such that $\Vdash^{\Sigma_{\At}}_{\mathcal{D}} \phi$ and $\Vdash^{\Pi_{\At}}_{\mathcal{D}} \psi$ for arbitrary $\Sigma_{\At}, \Pi_{\At}$, by~\ref{eq:supp-inf}, $\phi, \psi \Vdash^{\Delta_{\At}, \Theta_{\At}}_{\mathcal{C}} \bot$. Since $\Vdash^{\Gamma_{\At}}_{\mathcal{B}} \phi \otimes \psi$ and $\phi, \psi \Vdash^{\Delta_{\At}, \Theta_{\At}}_{\mathcal{C}} \bot$, by~\ref{eq:supp-tensor}, we obtain $\Vdash^{\Gamma_{\At}, \Delta_{\At}, \Theta_{\At}}_{\mathcal{C}} \bot$. Since $\mathcal{C} \supseteq \mathcal{B}$ such that $\Vdash^{\Theta_{\At}}_{\mathcal{C}} \bot$ for arbitrary $\Theta_{\At}$, and $\Vdash^{\Gamma_{\At}, \Delta_{\At}, \Theta_{\At}}_{\mathcal{C}} \bot$, by~\ref{eq:supp-1}, $\Vdash^{\Gamma_{\At}, \Delta_{\At}}_{\mathcal{B}} 1$.

        \item[$\chi = \alpha \parr \beta:$] Assume $\Vdash^{\Gamma_{\At}}_{\mathcal{B}} \phi \otimes \psi$ and $\phi, \psi \Vdash^{\Delta_{\At}}_{\mathcal{B}} \alpha \parr \beta$ for arbitrary $\Delta_{\At}$. Now assume that, for an arbitrary $\mathcal{C} \supseteq \mathcal{B}$ and arbitrary $\Theta_{\At}$, $\Omega_{\At}$ $\alpha \Vdash^{\Theta_{\At}}_{\mathcal{C}} \bot$ and $\beta \Vdash^{\Omega_{\At}}_{\mathcal{C}} \bot$. Further assume that, for an arbitrary $\mathcal{D} \supseteq \mathcal{C}$ and arbitrary $\Sigma_{\At}, \Pi_{\At}$, $\Vdash^{\Sigma_{\At}}_{\mathcal{D}} \phi$ and $\Vdash^{\Pi_{\At}}_{\mathcal{D}} \psi$. By monotonicity, $\phi, \psi \Vdash^{\Delta_{\At}}_{\mathcal{D}} \alpha \oplus \beta$, hence with $\Vdash^{\Sigma_{\At}}_{\mathcal{D}} \phi$ and $\Vdash^{\Pi_{\At}}_{\mathcal{D}} \psi$, by~\ref{eq:supp-inf}, we obtain $\Vdash^{\Delta_{\At}, \Sigma_{\At}, \Pi_{\At}}_{\mathcal{D}} \alpha \parr \beta$. From $\Vdash^{\Delta_{\At}, \Sigma_{\At}, \Pi_{\At}}_{\mathcal{D}} \alpha \parr \beta$ and $\alpha \Vdash^{\Theta_{\At}}_{\mathcal{C}} \bot$ and $\beta \Vdash^{\Omega_{\At}}_{\mathcal{C}} \bot$ (thus also $\alpha \Vdash^{\Theta_{\At}}_{\mathcal{D}} \bot$ and $\beta \Vdash^{\Omega_{\At}}_{\mathcal{D}} \bot$), by~\ref{eq:supp-parr}, we obtain $\Vdash^{\Delta_{\At}, \Sigma_{\At}, \Pi_{\At}, \Theta_{\At}, \Omega_{\At}}_{\mathcal{D}} \bot$. Since $\Vdash^{\Delta_{\At}, \Sigma_{\At}, \Pi_{\At}, \Theta_{\At}, \Omega_{\At}}_{\mathcal{D}} \bot$ and $\mathcal{D} \supseteq \mathcal{C}$ such that $\Vdash^{\Sigma_{\At}}_{\mathcal{D}} \phi$ and $\Vdash^{\Pi_{\At}}_{\mathcal{D}} \psi$ for arbitrary $\Sigma_{\At}, \Pi_{\At}$, by~\ref{eq:supp-inf}, $\phi, \psi \Vdash^{\Delta_{\At}, \Theta_{\At}, \Omega_{\At}}_{\mathcal{C}} \bot$. Since $\Vdash^{\Gamma_{\At}}_{\mathcal{B}} \phi \otimes \psi$ and $\phi, \psi \Vdash^{\Delta_{\At}, \Theta_{\At}, \Omega_{\At}}_{\mathcal{C}} \bot$, by~\ref{eq:supp-tensor}, we obtain $\Vdash^{\Gamma_{\At}, \Delta_{\At}, \Theta_{\At}, \Omega_{\At}}_{\mathcal{C}} \bot$. Since $\mathcal{C} \supseteq \mathcal{B}$ such that $\alpha \Vdash^{\Theta_{\At}}_{\mathcal{C}} \bot$ and $\beta \Vdash^{\Omega_{\At}}_{\mathcal{C}} \bot$ for arbitrary $\Theta_{\At}$, $\Omega_{\At}$, and $\Vdash^{\Gamma_{\At}, \Delta_{\At}, \Theta_{\At}, \Omega_{\At}}_{\mathcal{C}} \bot$, by~\ref{eq:supp-parr}, $\Vdash^{\Gamma_{\At}, \Delta_{\At}}_{\mathcal{B}} \alpha \parr \beta$.

        \item[$\chi = \alpha \with \beta:$] Assume $\Vdash^{\Gamma_{\At}}_{\mathcal{B}} \phi \otimes \psi$ and $\phi, \psi \Vdash^{\Delta_{\At}}_{\mathcal{B}} \alpha \with \beta$ for arbitrary $\Delta_{\At}$. Further assume that, for an arbitrary $\mathcal{C} \supseteq \mathcal{B}$ and arbitrary $\Sigma_{\At}, \Pi_{\At}$, $\Vdash^{\Sigma_{\At}}_{\mathcal{C}} \phi$ and $\Vdash^{\Pi_{\At}}_{\mathcal{C}} \psi$. Since $\phi, \psi \Vdash^{\Delta_{\At}}_{\mathcal{B}} \alpha \with \beta$ and $\Vdash^{\Sigma_{\At}}_{\mathcal{C}} \phi$ and $\Vdash^{\Pi_{\At}}_{\mathcal{C}} \psi$, by~\ref{eq:supp-inf}, we obtain $\Vdash^{\Delta_{\At}, \Sigma_{\At}, \Pi_{\At}}_{\mathcal{C}} \alpha \with \beta$. By~\ref{eq:supp-and}, we thus obtain $\Vdash^{\Delta_{\At}, \Sigma_{\At}, \Pi_{\At}}_{\mathcal{C}} \alpha$ and $\Vdash^{\Delta_{\At}, \Sigma_{\At}, \Pi_{\At}}_{\mathcal{C}}\beta$. Since $\mathcal{C} \supseteq \mathcal{B}$ such that  $\Vdash^{\Sigma_{\At}}_{\mathcal{C}} \phi$ and $\Vdash^{\Pi_{\At}}_{\mathcal{C}} \psi$ for arbitrary $\Sigma_{\At}, \Pi_{\At}$, by~\ref{eq:supp-inf}, $\phi, \psi \Vdash^{\Delta_{\At}}_{\mathcal{B}} \alpha$ and $\phi, \psi \Vdash^{\Delta_{\At}}_{\mathcal{B}} \beta$. Since $\Vdash^{\Gamma_{\At}}_{\mathcal{B}} \phi \otimes \psi$ and $\phi, \psi \Vdash^{\Delta_{\At}}_{\mathcal{B}} \alpha$, by the induction hypothesis, we obtain $\Vdash^{\Gamma_{\At}, \Delta_{\At}}_{\mathcal{B}} \alpha$. Analogously, we obtain $\Vdash^{\Gamma_{\At}, \Delta_{\At}}_{\mathcal{B}} \beta$. Now, by~\ref{eq:supp-and}, we obtain $\Vdash^{\Gamma_{\At}, \Delta_{\At}}_{\mathcal{B}} \alpha \with \beta$ as required.

        \item[$\chi = \alpha \oplus \beta:$] Assume $\Vdash^{\Gamma_{\At}}_{\mathcal{B}} \phi \otimes \psi$ and $\phi, \psi \Vdash^{\Delta_{\At}}_{\mathcal{B}} \alpha \oplus \beta$ for arbitrary $\Delta_{\At}$. Now assume that, for an arbitrary $\mathcal{C} \supseteq \mathcal{B}$ and arbitrary $\Theta_{\At}$, $\alpha \Vdash^{\Theta_{\At}}_{\mathcal{C}} \bot$ and $\beta \Vdash^{\Theta_{\At}}_{\mathcal{C}} \bot$. Further assume that, for an arbitrary $\mathcal{D} \supseteq \mathcal{C}$ and arbitrary $\Sigma_{\At}, \Pi_{\At}$, $\Vdash^{\Sigma_{\At}}_{\mathcal{D}} \phi$ and $\Vdash^{\Pi_{\At}}_{\mathcal{D}} \psi$. By monotonicity, $\phi, \psi \Vdash^{\Delta_{\At}}_{\mathcal{D}} \alpha \oplus \beta$, hence with $\Vdash^{\Sigma_{\At}}_{\mathcal{D}} \phi$ and $\Vdash^{\Pi_{\At}}_{\mathcal{D}} \psi$, by~\ref{eq:supp-inf}, we obtain $\Vdash^{\Delta_{\At}, \Sigma_{\At}, \Pi_{\At}}_{\mathcal{D}} \alpha \oplus \beta$. From $\Vdash^{\Delta_{\At}, \Sigma_{\At}, \Pi_{\At}}_{\mathcal{D}} \alpha \oplus \beta$ and $\alpha \Vdash^{\Theta_{\At}}_{\mathcal{C}} \bot$ and $\beta \Vdash^{\Theta_{\At}}_{\mathcal{C}} \bot$ (thus also $\alpha \Vdash^{\Theta_{\At}}_{\mathcal{D}} \bot$ and $\beta \Vdash^{\Theta_{\At}}_{\mathcal{D}} \bot$), by~\ref{eq:supp-plus}, we obtain $\Vdash^{\Delta_{\At}, \Sigma_{\At}, \Pi_{\At}, \Theta_{\At}}_{\mathcal{D}} \bot$. Since $\Vdash^{\Delta_{\At}, \Sigma_{\At}, \Pi_{\At}, \Theta_{\At}}_{\mathcal{D}} \bot$ and $\mathcal{D} \supseteq \mathcal{C}$ such that $\Vdash^{\Sigma_{\At}}_{\mathcal{D}} \phi$ and $\Vdash^{\Pi_{\At}}_{\mathcal{D}} \psi$ for arbitrary $\Sigma_{\At}, \Pi_{\At}$, by~\ref{eq:supp-inf}, $\phi, \psi \Vdash^{\Delta_{\At}, \Theta_{\At}}_{\mathcal{C}} \bot$. Since $\Vdash^{\Gamma_{\At}}_{\mathcal{B}} \phi \otimes \psi$ and $\phi, \psi \Vdash^{\Delta_{\At}, \Theta_{\At}}_{\mathcal{C}} \bot$, by~\ref{eq:supp-tensor}, we obtain $\Vdash^{\Gamma_{\At}, \Delta_{\At}, \Theta_{\At}}_{\mathcal{C}} \bot$. Since $\mathcal{C} \supseteq \mathcal{B}$ such that $\alpha \Vdash^{\Theta_{\At}}_{\mathcal{C}} \bot$ and $\beta \Vdash^{\Theta_{\At}}_{\mathcal{C}} \bot$ for arbitrary $\Theta_{\At}$, and $\Vdash^{\Gamma_{\At}, \Delta_{\At}, \Theta_{\At}}_{\mathcal{C}} \bot$, by~\ref{eq:supp-plus}, $\Vdash^{\Gamma_{\At}, \Delta_{\At}}_{\mathcal{B}} \alpha \oplus \beta$.

        \item[$\chi = \top:$] Assume $\Vdash^{\Gamma_{\At}}_{\mathcal{B}} \phi \otimes \psi$ and $\phi, \psi \Vdash^{\Delta_{\At}}_{\mathcal{B}} \top$ for arbitrary $\Delta_{\At}$. By~\ref{eq:supp-top}, $\Vdash^{\Gamma_{\At}, \Delta_{\At}}_{\mathcal{B}} \top$.

        \item[$\chi = 0:$] Assume $\Vdash^{\Gamma_{\At}}_{\mathcal{B}} \phi \otimes \psi$ and $\phi, \psi \Vdash^{\Delta_{\At}}_{\mathcal{B}} 0$ for arbitrary $\Delta_{\At}$. Further assume that, for an arbitrary $\mathcal{C} \supseteq \mathcal{B}$ and arbitrary $\Sigma_{\At}, \Pi_{\At}$, $\Vdash^{\Sigma_{\At}}_{\mathcal{C}} \phi$ and $\Vdash^{\Pi_{\At}}_{\mathcal{C}} \psi$. Since $\phi, \psi \Vdash^{\Delta_{\At}}_{\mathcal{B}} 0$ and $\Vdash^{\Sigma_{\At}}_{\mathcal{C}} \phi$ and $\Vdash^{\Pi_{\At}}_{\mathcal{C}} \psi$, by~\ref{eq:supp-inf}, we obtain $\Vdash^{\Delta_{\At}, \Sigma_{\At}, \Pi_{\At}}_{\mathcal{C}} 0$. By~\ref{eq:supp-0}, we thus obtain $\Vdash^{\Delta_{\At}, \Sigma_{\At}, \Pi_{\At}, \Theta_{\At}}_{\mathcal{C}} \bot$ for arbitrary $\Theta_{\At}$. Since $\mathcal{C} \supseteq \mathcal{B}$ such that  $\Vdash^{\Sigma_{\At}}_{\mathcal{C}} \phi$ and $\Vdash^{\Pi_{\At}}_{\mathcal{C}} \psi$ for arbitrary $\Sigma_{\At}, \Pi_{\At}$, by~\ref{eq:supp-inf}, $\phi, \psi \Vdash^{\Delta_{\At}, \Theta_{\At}}_{\mathcal{B}} \bot$. Since $\Vdash^{\Gamma_{\At}}_{\mathcal{B}} \phi \otimes \psi$ and $\phi, \psi \Vdash^{\Delta_{\At}, \Theta_{\At}}_{\mathcal{B}} \bot$, by~\ref{eq:supp-tensor}, we obtain $\Vdash^{\Gamma_{\At}, \Delta_{\At}, \Theta_{\At}}_{\mathcal{B}} \bot$. Since $\Theta_{\At}$ is arbitrary, we obtain $\Vdash^{\Gamma_{\At}, \Delta_{\At}}_{\mathcal{B}} 0$ by~\ref{eq:supp-0}.
    \end{description}
\end{proof}

\begin{lemma} \label{lemma:genericimplication}
    If $\Vdash^{\Gamma_{\At}}_{\mathcal{B}} \phi \multimap \psi$ and $\Vdash^{\Delta_{\At}}_{\mathcal{B}} \phi$ and $\psi \Vdash^{\Theta_{\At}}_{\mathcal{B}} \chi$ then $\Vdash^{\Gamma_{\At}, \Delta_{\At}, \Theta_{\At}}_{\mathcal{B}} \chi$.
\end{lemma}
\begin{proof}
   We shall prove the statement inductively. Note that in the case of $(\with)$, we require an induction hypothesis stating that \[\text{if } \Vdash^{\Gamma_{\At}}_{\mathcal{B}} \phi \multimap \psi \text{ and } \Vdash^{\Delta_{\At}}_{\mathcal{B}} \phi  \text{ and } \psi \Vdash^{\Theta_{\At}}_{\mathcal{B}} \tau \text{ then } \Vdash^{\Gamma_{\At}, \Delta_{\At}, \Theta_{\At}}_{\mathcal{B}} \tau\] holds true for any proper subformula $\tau$ of $\chi$, as in the proof of Lemma~\ref{lemma:monotonicity}.
    
    \begin{description}[itemsep=0.5em]
        \item[$\chi = p:$] Assume $\Vdash^{\Gamma_{\At}}_{\mathcal{B}} \phi \multimap \psi$ and $\Vdash^{\Delta_{\At}}_{\mathcal{B}} \phi$ and $\psi \Vdash^{\Theta_{\At}}_{\mathcal{B}} p$ for arbitrary $\Delta_{\At}, \Theta_{\At}$. Now assume that, for an arbitrary $\mathcal{C} \supseteq \mathcal{B}$ and arbitrary $\Sigma_{\At}$, $p, \Sigma_{\At} \vdash_{\mathcal{C}} \bot$. Further assume that, for an arbitrary $\mathcal{D} \supseteq \mathcal{C}$ and arbitrary $\Pi_{\At}$, 
        $\Vdash^{\Pi_{\At}}_{\mathcal{D}} \psi$. By monotonicity, $\psi \Vdash^{\Theta_{\At}}_{\mathcal{D}} p$, hence with $\Vdash^{\Pi_{\At}}_{\mathcal{D}} \psi$, by~\ref{eq:supp-inf}, we obtain $\Vdash^{\Theta_{\At}, \Pi_{\At}}_{\mathcal{D}} p$. From $\Vdash^{\Theta_{\At}, \Pi_{\At}}_{\mathcal{D}} p$ and $p, \Sigma_{\At} \vdash_{\mathcal{C}} \bot$ (thus also $p, \Sigma_{\At} \vdash_{\mathcal{D}} \bot$), by~\ref{eq:supp-at}, we obtain $\Theta_{\At}, \Pi_{\At}, \Sigma_{\At} \vdash_{\mathcal{D}} \bot$. Hence, $\Vdash^{\Theta_{\At}, \Pi_{\At}, \Sigma_{\At}}_{\mathcal{D}} \bot$ by Lemma~\ref{lemma:bottomisspecial}. Since $\Vdash^{\Theta_{\At}, \Pi_{\At}, \Sigma_{\At}}_{\mathcal{D}} \bot$ and $\mathcal{D} \supseteq \mathcal{C}$ such that $\Vdash^{\Pi_{\At}}_{\mathcal{D}} \psi$ for arbitrary $\Pi_{\At}$, by~\ref{eq:supp-inf}, $\psi \Vdash^{\Theta_{\At}, \Sigma_{\At}}_{\mathcal{C}} \bot$. Since $\Vdash^{\Gamma_{\At}}_{\mathcal{B}} \phi \multimap \psi$, $\Vdash^{\Delta_{\At}}_{\mathcal{B}} \phi$ (thus also $\Vdash^{\Delta_{\At}}_{\mathcal{C}} \phi$) and $\psi \Vdash^{\Theta_{\At}, \Sigma_{\At}}_{\mathcal{C}} \bot$, by~\ref{eq:supp-imply}, we obtain $\Vdash^{\Gamma_{\At}, \Delta_{\At}, \Theta_{\At}, \Sigma_{\At}}_{\mathcal{C}} \bot$, hence $\Gamma_{\At}, \Delta_{\At}, \Theta_{\At}, \Sigma_{\At} \vdash_{\mathcal{C}} \bot$ by Lemma~\ref{lemma:bottomisspecial}. Since $\mathcal{C} \supseteq \mathcal{B}$ such that $p, \Sigma_{\At} \vdash_{\mathcal{C}} \bot$ for arbitrary $\Sigma_{\At}$, and $\Gamma_{\At}, \Delta_{\At}, \Theta_{\At}, \Sigma_{\At} \vdash_{\mathcal{C}} \bot$, by~\ref{eq:supp-at}, $\Vdash^{\Gamma_{\At}, \Delta_{\At},  \Theta_{\At}}_{\mathcal{B}} p$.

        \item[$\chi = \alpha \otimes \beta:$] Assume $\Vdash^{\Gamma_{\At}}_{\mathcal{B}} \phi \multimap \psi$ and $\Vdash^{\Delta_{\At}}_{\mathcal{B}} \phi$ and $\psi \Vdash^{\Theta_{\At}}_{\mathcal{B}} \alpha \otimes \beta$ for arbitrary $\Delta_{\At}, \Theta_{\At}$. Now assume that, for an arbitrary $\mathcal{C} \supseteq \mathcal{B}$ and arbitrary $\Sigma_{\At}$, $\alpha, \beta \Vdash^{\Sigma_{\At}}_{\mathcal{C}} \bot$. Further assume that, for an arbitrary $\mathcal{D} \supseteq \mathcal{C}$ and arbitrary $\Pi_{\At}$, $\Vdash^{\Pi_{\At}}_{\mathcal{D}} \psi$. By monotonicity, $\psi \Vdash^{\Theta_{\At}}_{\mathcal{D}} \alpha \otimes \beta$, hence with $\Vdash^{\Pi_{\At}}_{\mathcal{D}} \psi$, by~\ref{eq:supp-inf}, we obtain $\Vdash^{\Theta_{\At}, \Pi_{\At}}_{\mathcal{D}} \alpha \otimes \beta$. From $\Vdash^{\Theta_{\At}, \Pi_{\At}}_{\mathcal{D}} \alpha \otimes \beta$ and $\alpha, \beta \Vdash^{\Sigma_{\At}}_{\mathcal{C}} \bot$ (thus also $\alpha, \beta \Vdash^{\Sigma_{\At}}_{\mathcal{D}} \bot$), by~\ref{eq:supp-tensor}, we obtain $\Vdash^{\Theta_{\At}, \Pi_{\At}, \Sigma_{\At}}_{\mathcal{D}} \bot$. Since $\Vdash^{\Theta_{\At}, \Pi_{\At}, \Sigma_{\At}}_{\mathcal{D}} \bot$ and $\mathcal{D} \supseteq \mathcal{C}$ such that $\Vdash^{\Pi_{\At}}_{\mathcal{D}} \psi$ for arbitrary $\Pi_{\At}$, by~\ref{eq:supp-inf}, $\psi \Vdash^{\Theta_{\At}, \Sigma_{\At}}_{\mathcal{C}} \bot$. Since $\Vdash^{\Gamma_{\At}}_{\mathcal{B}} \phi \multimap \psi$, $\Vdash^{\Delta_{\At}}_{\mathcal{B}} \phi$ (thus also $\Vdash^{\Delta_{\At}}_{\mathcal{C}} \phi$) and $\psi \Vdash^{\Theta_{\At}, \Sigma_{\At}}_{\mathcal{C}} \bot$, by~\ref{eq:supp-imply}, we obtain $\Vdash^{\Gamma_{\At}, \Delta_{\At}, \Theta_{\At}, \Sigma_{\At}}_{\mathcal{C}} \bot$. Since $\mathcal{C} \supseteq \mathcal{B}$ such that $\alpha, \beta \Vdash^{\Sigma_{\At}}_{\mathcal{C}} \bot$ for arbitrary $\Sigma_{\At}$, and $\Vdash^{\Gamma_{\At}, \Delta_{\At}, \Theta_{\At}, \Sigma_{\At}}_{\mathcal{C}} \bot$, by~\ref{eq:supp-tensor}, $\Vdash^{\Gamma_{\At}, \Delta_{\At},  \Theta_{\At}}_{\mathcal{B}} \alpha \otimes \beta$.

        \item[$\chi = \alpha \multimap \beta:$] Assume $\Vdash^{\Gamma_{\At}}_{\mathcal{B}} \phi \multimap \psi$ and $\Vdash^{\Delta_{\At}}_{\mathcal{B}} \phi$ and $\psi \Vdash^{\Theta_{\At}}_{\mathcal{B}} \alpha \multimap \beta$ for arbitrary $\Delta_{\At}, \Theta_{\At}$. Now assume that, for an arbitrary $\mathcal{C} \supseteq \mathcal{B}$ and arbitrary $\Sigma_{\At}, \Omega_{\At}$, $\Vdash^{\Sigma_{\At}}_{\mathcal{C}} \alpha$ and $\beta \Vdash^{\Omega_{\At}}_{\mathcal{C}} \bot$. Further assume that, for an arbitrary $\mathcal{D} \supseteq \mathcal{C}$ and arbitrary $\Pi_{\At}$, $\Vdash^{\Pi_{\At}}_{\mathcal{D}} \psi$. By monotonicity, $\psi \Vdash^{\Theta_{\At}}_{\mathcal{D}} \alpha \multimap \beta$, hence with $\Vdash^{\Pi_{\At}}_{\mathcal{D}} \psi$, by~\ref{eq:supp-inf}, we obtain $\Vdash^{\Theta_{\At}, \Pi_{\At}}_{\mathcal{D}} \alpha \multimap \beta$. From $\Vdash^{\Theta_{\At}, \Pi_{\At}}_{\mathcal{D}} \alpha \multimap \beta$, $\Vdash^{\Sigma_{\At}}_{\mathcal{C}} \alpha$ and $\beta \Vdash^{\Omega_{\At}}_{\mathcal{C}} \bot$ (thus also $\Vdash^{\Sigma_{\At}}_{\mathcal{D}} \alpha$ and $\beta \Vdash^{\Omega_{\At}}_{\mathcal{D}} \bot$), by~\ref{eq:supp-imply}, we obtain $\Vdash^{\Theta_{\At}, \Pi_{\At}, \Sigma_{\At}, \Omega_{\At}}_{\mathcal{D}} \bot$. Since $\Vdash^{\Theta_{\At}, \Pi_{\At}, \Sigma_{\At}, \Omega_{\At}}_{\mathcal{D}} \bot$ and $\mathcal{D} \supseteq \mathcal{C}$ such that $\Vdash^{\Pi_{\At}}_{\mathcal{D}} \psi$ for arbitrary $\Pi_{\At}$, by~\ref{eq:supp-inf}, $\psi \Vdash^{\Theta_{\At}, \Sigma_{\At}, \Omega_{\At}}_{\mathcal{C}} \bot$. Since $\Vdash^{\Gamma_{\At}}_{\mathcal{B}} \phi \multimap \psi$, $\Vdash^{\Delta_{\At}}_{\mathcal{B}} \phi$ (thus also $\Vdash^{\Delta_{\At}}_{\mathcal{C}} \phi$) and $\psi \Vdash^{\Theta_{\At}, \Sigma_{\At}, \Omega_{\At}}_{\mathcal{C}} \bot$, by~\ref{eq:supp-imply}, we obtain $\Vdash^{\Gamma_{\At}, \Delta_{\At}, \Theta_{\At}, \Sigma_{\At}, \Omega_{\At}}_{\mathcal{C}} \bot$. Since $\mathcal{C} \supseteq \mathcal{B}$ such that $\Vdash^{\Sigma_{\At}}_{\mathcal{C}} \alpha$ and $\beta \Vdash^{\Omega_{\At}}_{\mathcal{C}} \bot$ for arbitrary  $\Sigma_{\At}, \Omega_{\At}$, and $\Vdash^{\Gamma_{\At}, \Delta_{\At}, \Theta_{\At}, \Sigma_{\At}, \Omega_{\At}}_{\mathcal{C}} \bot$, by~\ref{eq:supp-imply}, $\Vdash^{\Gamma_{\At}, \Delta_{\At}, \Theta_{\At}}_{\mathcal{B}} \alpha \multimap \beta$.

        \item[$\chi = 1:$] Assume $\Vdash^{\Gamma_{\At}}_{\mathcal{B}} \phi \multimap \psi$ and $\Vdash^{\Delta_{\At}}_{\mathcal{B}} \phi$ and $\psi \Vdash^{\Theta_{\At}}_{\mathcal{B}} 1$ for arbitrary $\Delta_{\At}, \Theta_{\At}$. Now assume that, for an arbitrary $\mathcal{C} \supseteq \mathcal{B}$ and arbitrary $\Sigma_{\At}$, $\Vdash^{\Sigma_{\At}}_{\mathcal{C}} \bot$. Further assume that, for an arbitrary $\mathcal{D} \supseteq \mathcal{C}$ and arbitrary $\Pi_{\At}$, $\Vdash^{\Pi_{\At}}_{\mathcal{D}} \psi$. By monotonicity, $\psi \Vdash^{\Theta_{\At}}_{\mathcal{D}} 1$, hence with $\Vdash^{\Pi_{\At}}_{\mathcal{D}} \psi$, by~\ref{eq:supp-inf}, we obtain $\Vdash^{\Theta_{\At}, \Pi_{\At}}_{\mathcal{D}} 1$. From $\Vdash^{\Theta_{\At}, \Pi_{\At}}_{\mathcal{D}} 1$ and $\Vdash^{\Sigma_{\At}}_{\mathcal{C}} \bot$ (thus also $\Vdash^{\Sigma_{\At}}_{\mathcal{D}} \bot$), by~\ref{eq:supp-1}, we obtain $\Vdash^{\Theta_{\At}, \Pi_{\At}, \Sigma_{\At}}_{\mathcal{D}} \bot$. Since $\Vdash^{\Theta_{\At}, \Pi_{\At}, \Sigma_{\At}}_{\mathcal{D}} \bot$ and $\mathcal{D} \supseteq \mathcal{C}$ such that $\Vdash^{\Pi_{\At}}_{\mathcal{D}} \psi$ for arbitrary $\Pi_{\At}$, by~\ref{eq:supp-inf}, $\psi \Vdash^{\Theta_{\At}, \Sigma_{\At}}_{\mathcal{C}} \bot$. Since $\Vdash^{\Gamma_{\At}}_{\mathcal{B}} \phi \multimap \psi$, $\Vdash^{\Delta_{\At}}_{\mathcal{B}} \phi$ (thus also $\Vdash^{\Delta_{\At}}_{\mathcal{C}} \phi$) and $\psi \Vdash^{\Theta_{\At}, \Sigma_{\At}}_{\mathcal{C}} \bot$, by~\ref{eq:supp-imply}, we obtain $\Vdash^{\Gamma_{\At}, \Delta_{\At}, \Theta_{\At}, \Sigma_{\At}}_{\mathcal{C}} \bot$. Since $\mathcal{C} \supseteq \mathcal{B}$ such that $\Vdash^{\Sigma_{\At}}_{\mathcal{C}} \bot$ for arbitrary $\Sigma_{\At}$, and $\Vdash^{\Gamma_{\At}, \Delta_{\At}, \Theta_{\At}, \Sigma_{\At}}_{\mathcal{C}} \bot$, by~\ref{eq:supp-1}, $\Vdash^{\Gamma_{\At}, \Delta_{\At},  \Theta_{\At}}_{\mathcal{B}} 1$.

        \item[$\chi = \alpha \parr \beta:$] Assume $\Vdash^{\Gamma_{\At}}_{\mathcal{B}} \phi \multimap \psi$ and $\Vdash^{\Delta_{\At}}_{\mathcal{B}} \phi$ and $\psi \Vdash^{\Theta_{\At}}_{\mathcal{B}} \alpha \parr \beta$ for arbitrary $\Delta_{\At}, \Theta_{\At}$. Now assume that, for an arbitrary $\mathcal{C} \supseteq \mathcal{B}$ and arbitrary $\Sigma_{\At}$, $\Omega_{\At}$, $\alpha \Vdash^{\Sigma_{\At}}_{\mathcal{C}} \bot$ and $\beta \Vdash^{\Omega_{\At}}_{\mathcal{C}} \bot$. Further assume that, for an arbitrary $\mathcal{D} \supseteq \mathcal{C}$ and arbitrary $\Pi_{\At}$, $\Vdash^{\Pi_{\At}}_{\mathcal{D}} \psi$. By monotonicity, $\psi \Vdash^{\Theta_{\At}}_{\mathcal{D}} \alpha \parr \beta$, hence with $\Vdash^{\Pi_{\At}}_{\mathcal{D}} \psi$, by~\ref{eq:supp-inf}, we obtain $\Vdash^{\Theta_{\At}, \Pi_{\At}}_{\mathcal{D}} \alpha \parr \beta$. From $\Vdash^{\Theta_{\At}, \Pi_{\At}}_{\mathcal{D}} \alpha \parr \beta$ and $\alpha \Vdash^{\Sigma_{\At}}_{\mathcal{C}} \bot$ and $\beta \Vdash^{\Omega_{\At}}_{\mathcal{C}} \bot$ (thus also $\alpha \Vdash^{\Sigma_{\At}}_{\mathcal{D}} \bot$ and $\beta \Vdash^{\Omega_{\At}}_{\mathcal{D}} \bot$), by~\ref{eq:supp-parr}, we obtain $\Vdash^{\Pi_{\At}, \Theta_{\At}, \Sigma_{\At}, \Omega_{\At}}_{\mathcal{D}} \bot$. Since $\Vdash^{\Pi_{\At}, \Theta_{\At}, \Sigma_{\At}, \Omega_{\At}}_{\mathcal{D}} \bot$ and $\mathcal{D} \supseteq \mathcal{C}$ such that $\Vdash^{\Pi_{\At}}_{\mathcal{D}} \psi$ for arbitrary $\Pi_{\At}$, by~\ref{eq:supp-inf}, $\psi \Vdash^{\Theta_{\At}, \Sigma_{\At}, \Omega_{\At}}_{\mathcal{C}} \bot$. Since $\Vdash^{\Gamma_{\At}}_{\mathcal{B}} \phi \multimap \psi$ and $\Vdash^{\Delta_{\At}}_{\mathcal{B}} \phi$ (thus also $\Vdash^{\Delta_{\At}}_{\mathcal{C}} \phi$) and $\psi \Vdash^{\Theta_{\At}, \Sigma_{\At}, \Omega_{\At}}_{\mathcal{C}} \bot$, by~\ref{eq:supp-imply}, we obtain $\Vdash^{\Gamma_{\At}, \Delta_{\At}, \Theta_{\At}, \Sigma_{\At}, \Omega_{\At}}_{\mathcal{C}} \bot$. Since $\mathcal{C} \supseteq \mathcal{B}$ such that $\alpha \Vdash^{\Sigma_{\At}}_{\mathcal{C}} \bot$ and $\beta \Vdash^{\Omega_{\At}}_{\mathcal{C}} \bot$ for arbitrary $\Sigma_{\At}$, $\Omega_{\At}$, and $\Vdash^{\Gamma_{\At}, \Delta_{\At}, \Theta_{\At}, \Sigma_{\At}, \Omega_{\At}}_{\mathcal{C}} \bot$, by~\ref{eq:supp-parr}, $\Vdash^{\Gamma_{\At}, \Delta_{\At}, \Theta_{\At}}_{\mathcal{B}} \alpha \parr \beta$.

        \item[$\chi = \alpha \with \beta:$] Assume $\Vdash^{\Gamma_{\At}}_{\mathcal{B}} \phi \multimap \psi$ and $\Vdash^{\Delta_{\At}}_{\mathcal{B}} \phi$ and $\psi \Vdash^{\Theta_{\At}}_{\mathcal{B}} \alpha \with \beta$ for arbitrary $\Delta_{\At}, \Theta_{\At}$. Further assume that, for an arbitrary $\mathcal{C} \supseteq \mathcal{B}$ and arbitrary $\Pi_{\At}$, $\Vdash^{\Pi_{\At}}_{\mathcal{C}} \psi$. Since $\psi \Vdash^{\Theta_{\At}}_{\mathcal{B}} \alpha \with \beta$ and $\Vdash^{\Pi_{\At}}_{\mathcal{C}} \psi$, by~\ref{eq:supp-inf}, we obtain $\Vdash^{\Theta_{\At}, \Pi_{\At}}_{\mathcal{C}} \alpha \with \beta$. By~\ref{eq:supp-and}, we thus obtain $\Vdash^{\Theta_{\At}, \Pi_{\At}}_{\mathcal{C}} \alpha$ and $\Vdash^{\Theta_{\At}, \Pi_{\At}}_{\mathcal{C}}\beta$. Since $\mathcal{C} \supseteq \mathcal{B}$ such that $\Vdash^{\Pi_{\At}}_{\mathcal{C}} \psi$ for arbitrary $\Pi_{\At}$, by~\ref{eq:supp-inf}, $\psi \Vdash^{\Theta_{\At}}_{\mathcal{B}} \alpha$ and $\psi \Vdash^{\Theta_{\At}}_{\mathcal{B}} \beta$. Since $\Vdash^{\Gamma_{\At}}_{\mathcal{B}} \phi \multimap \psi$ and $\Vdash^{\Delta_{\At}}_{\mathcal{B}} \phi$ and $\psi \Vdash^{\Theta_{\At}}_{\mathcal{B}} \alpha$, by the induction hypothesis, we obtain $\Vdash^{\Gamma_{\At}, \Delta_{\At}, \Theta_{\At}}_{\mathcal{B}} \alpha$. Analogously, we obtain $\Vdash^{\Gamma_{\At}, \Delta_{\At}, \Theta_{\At}}_{\mathcal{B}} \beta$. Now, by~\ref{eq:supp-and}, we obtain $\Vdash^{\Gamma_{\At}, \Delta_{\At}, \Theta_{\At}}_{\mathcal{B}} \alpha \with \beta$ as required.

        \item[$\chi = \alpha \oplus \beta:$] Assume $\Vdash^{\Gamma_{\At}}_{\mathcal{B}} \phi \multimap \psi$ and $\Vdash^{\Delta_{\At}}_{\mathcal{B}} \phi$ and $\psi \Vdash^{\Theta_{\At}}_{\mathcal{B}} \alpha \oplus \beta$ for arbitrary $\Delta_{\At}, \Theta_{\At}$. Now assume that, for an arbitrary $\mathcal{C} \supseteq \mathcal{B}$ and arbitrary $\Sigma_{\At}$, $\alpha \Vdash^{\Sigma_{\At}}_{\mathcal{C}} \bot$ and $\beta \Vdash^{\Sigma_{\At}}_{\mathcal{C}} \bot$. Further assume that, for an arbitrary $\mathcal{D} \supseteq \mathcal{C}$ and arbitrary $\Pi_{\At}$, $\Vdash^{\Pi_{\At}}_{\mathcal{D}} \psi$. By monotonicity, $\psi \Vdash^{\Theta_{\At}}_{\mathcal{D}} \alpha \oplus \beta$, hence with $\Vdash^{\Pi_{\At}}_{\mathcal{D}} \psi$, by~\ref{eq:supp-inf}, we obtain $\Vdash^{\Theta_{\At}, \Pi_{\At}}_{\mathcal{D}} \alpha \oplus \beta$. From $\Vdash^{\Theta_{\At}, \Pi_{\At}}_{\mathcal{D}} \alpha \oplus \beta$ and $\alpha \Vdash^{\Sigma_{\At}}_{\mathcal{C}} \bot$ and $\beta \Vdash^{\Sigma_{\At}}_{\mathcal{C}} \bot$ (thus also $\alpha \Vdash^{\Sigma_{\At}}_{\mathcal{D}} \bot$ and $\beta \Vdash^{\Sigma_{\At}}_{\mathcal{D}} \bot$), by~\ref{eq:supp-plus}, we obtain $\Vdash^{\Pi_{\At}, \Theta_{\At}, \Sigma_{\At}}_{\mathcal{D}} \bot$. Since $\Vdash^{\Pi_{\At}, \Theta_{\At}, \Sigma_{\At}}_{\mathcal{D}} \bot$ and $\mathcal{D} \supseteq \mathcal{C}$ such that $\Vdash^{\Pi_{\At}}_{\mathcal{D}} \psi$ for arbitrary $\Pi_{\At}$, by~\ref{eq:supp-inf}, $\psi \Vdash^{\Theta_{\At}, \Sigma_{\At}}_{\mathcal{C}} \bot$. Since $\Vdash^{\Gamma_{\At}}_{\mathcal{B}} \phi \multimap \psi$ and $\Vdash^{\Delta_{\At}}_{\mathcal{B}} \phi$ (thus also $\Vdash^{\Delta_{\At}}_{\mathcal{C}} \phi$) and $\psi \Vdash^{\Theta_{\At}, \Sigma_{\At}}_{\mathcal{C}} \bot$, by~\ref{eq:supp-imply}, we obtain $\Vdash^{\Gamma_{\At}, \Delta_{\At}, \Theta_{\At}, \Sigma_{\At}}_{\mathcal{C}} \bot$. Since $\mathcal{C} \supseteq \mathcal{B}$ such that $\alpha \Vdash^{\Sigma_{\At}}_{\mathcal{C}} \bot$ and $\beta \Vdash^{\Sigma_{\At}}_{\mathcal{C}} \bot$ for arbitrary $\Sigma_{\At}$, and $\Vdash^{\Gamma_{\At}, \Delta_{\At}, \Theta_{\At}, \Sigma_{\At}}_{\mathcal{C}} \bot$, by~\ref{eq:supp-plus}, $\Vdash^{\Gamma_{\At}, \Delta_{\At}, \Theta_{\At}}_{\mathcal{B}} \alpha \oplus \beta$.

        \item[$\chi = \top:$] Assume $\Vdash^{\Gamma_{\At}}_{\mathcal{B}} \phi \multimap \psi$ and $\Vdash^{\Delta_{\At}}_{\mathcal{B}} \phi$ and $\psi \Vdash^{\Theta_{\At}}_{\mathcal{B}} \top$ for arbitrary $\Delta_{\At}, \Theta_{\At}$. By~\ref{eq:supp-top}, $\Vdash^{\Gamma_{\At}, \Delta_{\At}, \Theta_{\At}}_{\mathcal{B}} \top$.

        \item[$\chi = 0:$] Assume $\Vdash^{\Gamma_{\At}}_{\mathcal{B}} \phi \multimap \psi$ and $\Vdash^{\Delta_{\At}}_{\mathcal{B}} \phi$ and $\psi \Vdash^{\Theta_{\At}}_{\mathcal{B}} 0$ for arbitrary $\Delta_{\At}, \Theta_{\At}$. Further assume that, for an arbitrary $\mathcal{C} \supseteq \mathcal{B}$ and arbitrary $\Pi_{\At}$, $\Vdash^{\Pi_{\At}}_{\mathcal{C}} \psi$. Since $\psi \Vdash^{\Theta_{\At}}_{\mathcal{B}} 0$ and $\Vdash^{\Pi_{\At}}_{\mathcal{C}} \psi$, by~\ref{eq:supp-inf}, we obtain $\Vdash^{\Pi_{\At}, \Theta_{\At}}_{\mathcal{C}} 0$. By~\ref{eq:supp-0}, we thus obtain $\Vdash^{\Pi_{\At}, \Theta_{\At}, \Sigma_{\At}}_{\mathcal{C}} \bot$ for arbitrary $\Sigma_{\At}$. Since $\mathcal{C} \supseteq \mathcal{B}$ such that $\Vdash^{\Pi_{\At}}_{\mathcal{C}} \psi$ for arbitrary $\Pi_{\At}$, by~\ref{eq:supp-inf}, $\psi \Vdash^{\Theta_{\At}, \Sigma_{\At}}_{\mathcal{B}} \bot$. Since $\Vdash^{\Gamma_{\At}}_{\mathcal{B}} \phi \multimap \psi$ and $\Vdash^{\Delta_{\At}}_{\mathcal{B}} \phi$ and $\psi \Vdash^{\Theta_{\At}, \Sigma_{\At}}_{\mathcal{B}} \bot$, by~\ref{eq:supp-imply}, we obtain $\Vdash^{\Gamma_{\At}, \Delta_{\At}, \Theta_{\At}, \Sigma_{\At}}_{\mathcal{B}} \bot$. Since $\Sigma_{\At}$ is arbitrary, we obtain $\Vdash^{\Gamma_{\At}, \Delta_{\At}, \Theta_{\At}}_{\mathcal{B}} 0$ by~\ref{eq:supp-0}.
    \end{description}
\end{proof}

\begin{lemma} \label{lemma:genericone}
    If $\Vdash^{\Gamma_{\At}}_{\mathcal{B}} 1$ and $\Vdash^{\Delta_{\At}}_{\mathcal{B}} \chi$ then $\Vdash^{\Gamma_{\At}, \Delta_{\At}}_{\mathcal{B}} \chi$.
\end{lemma}
\begin{proof}
   We shall prove the statement inductively. Note that in the case of $(\with)$, we require an induction hypothesis stating that \[\text{if } \Vdash^{\Gamma_{\At}}_{\mathcal{B}} 1 \text{ and } \Vdash^{\Delta_{\At}}_{\mathcal{B}} \tau \text{ then } \Vdash^{\Gamma_{\At}, \Delta_{\At}}_{\mathcal{B}} \tau\] holds true for any proper subformula $\tau$ of $\chi$, as in the proof of Lemma~\ref{lemma:monotonicity}.
    
    \begin{description}[itemsep=0.5em]
        \item[$\chi = p:$] Assume $\Vdash^{\Gamma_{\At}}_{\mathcal{B}} 1$ and $\Vdash^{\Delta_{\At}}_{\mathcal{B}} p$ for arbitrary $\Delta_{\At}$. Further assume that, for an arbitrary $\mathcal{C} \supseteq \mathcal{B}$ and arbitrary $\Theta_{\At}$, $p, \Theta_{\At} \vdash_{\mathcal{C}} \bot$. From $\Vdash^{\Delta_{\At}}_{\mathcal{B}} p$ and $p, \Theta_{\At} \vdash_{\mathcal{C}} \bot$, by~\ref{eq:supp-at}, we obtain $\Delta_{\At}, \Theta_{\At} \vdash_{\mathcal{C}} \bot$. Hence, $\Vdash^{\Delta_{\At}, \Theta_{\At}}_{\mathcal{C}} \bot$ by Lemma~\ref{lemma:bottomisspecial}. Since $\Vdash^{\Gamma_{\At}}_{\mathcal{B}} 1$ and $\Vdash^{\Delta_{\At}, \Theta_{\At}}_{\mathcal{C}} \bot$, by~\ref{eq:supp-1}, we obtain $\Vdash^{\Gamma_{\At}, \Delta_{\At}, \Theta_{\At}}_{\mathcal{C}} \bot$, hence $\Gamma_{\At}, \Delta_{\At}, \Theta_{\At} \vdash_{\mathcal{C}} \bot$ by Lemma~\ref{lemma:bottomisspecial}. Since $\mathcal{C} \supseteq \mathcal{B}$ such that $p, \Theta_{\At} \vdash_{\mathcal{C}} \bot$ for arbitrary $\Theta_{\At}$, and $\Gamma_{\At}, \Delta_{\At}, \Theta_{\At} \vdash_{\mathcal{C}} \bot$, by~\ref{eq:supp-at}, $\Vdash^{\Gamma_{\At}, \Delta_{\At}}_{\mathcal{B}} p$.

        \item[$\chi = \alpha \otimes \beta:$] Assume $\Vdash^{\Gamma_{\At}}_{\mathcal{B}} 1$ and $\Vdash^{\Delta_{\At}}_{\mathcal{B}} \alpha \otimes \beta$ for arbitrary $\Delta_{\At}$. Further assume that, for an arbitrary $\mathcal{C} \supseteq \mathcal{B}$ and arbitrary $\Theta_{\At}$, $\alpha, \beta \Vdash^{\Theta_{\At}}_{\mathcal{C}} \bot$. From $\Vdash^{\Delta_{\At}}_{\mathcal{B}} \alpha \otimes \beta$ and $\alpha, \beta \Vdash^{\Theta_{\At}}_{\mathcal{C}} \bot$, by~\ref{eq:supp-tensor}, we obtain $\Vdash^{\Delta_{\At}, \Theta_{\At}}_{\mathcal{C}} \bot$. Since $\Vdash^{\Gamma_{\At}}_{\mathcal{B}} 1$ and $\Vdash^{\Delta_{\At}, \Theta_{\At}}_{\mathcal{C}} \bot$, by~\ref{eq:supp-1}, we obtain $\Vdash^{\Gamma_{\At}, \Delta_{\At}, \Theta_{\At}}_{\mathcal{C}} \bot$. Since $\mathcal{C} \supseteq \mathcal{B}$ such that $\alpha, \beta \Vdash^{\Theta_{\At}}_{\mathcal{C}} \bot$ for arbitrary $\Theta_{\At}$, and $\Vdash^{\Gamma_{\At}, \Delta_{\At}, \Theta_{\At}}_{\mathcal{C}} \bot$, b~\ref{eq:supp-tensor}, $\Vdash^{\Gamma_{\At}, \Delta_{\At}}_{\mathcal{B}} \alpha \otimes \beta$.

        \item[$\chi = \alpha \multimap \beta:$] Assume $\Vdash^{\Gamma_{\At}}_{\mathcal{B}} 1$ and $\Vdash^{\Delta_{\At}}_{\mathcal{B}} \alpha \multimap \beta$ for arbitrary $\Delta_{\At}$. Further assume that, for an arbitrary $\mathcal{C} \supseteq \mathcal{B}$ and arbitrary $\Theta_{\At}, \Omega_{\At}$, $\Vdash^{\Theta_{\At}}_{\mathcal{C}} \alpha$ and $\beta \Vdash^{\Omega_{\At}}_{\mathcal{C}} \bot$. From $\Vdash^{\Delta_{\At}}_{\mathcal{B}} \alpha \multimap \beta$, $\Vdash^{\Theta_{\At}}_{\mathcal{C}} \alpha$ and $\beta \Vdash^{\Omega_{\At}}_{\mathcal{C}} \bot$, by~\ref{eq:supp-imply}, we obtain $\Vdash^{\Delta_{\At}, \Theta_{\At}, \Omega_{\At}}_{\mathcal{C}} \bot$. Since $\Vdash^{\Gamma_{\At}}_{\mathcal{B}} 1$ and $\Vdash^{\Delta_{\At}, \Theta_{\At}, \Omega_{\At}}_{\mathcal{C}} \bot$, by~\ref{eq:supp-1}, we obtain $\Vdash^{\Gamma_{\At}, \Delta_{\At}, \Theta_{\At}, \Omega_{\At}}_{\mathcal{C}} \bot$. Since $\mathcal{C} \supseteq \mathcal{B}$ such that $\Vdash^{\Theta_{\At}}_{\mathcal{C}} \alpha$ and $\beta \Vdash^{\Omega_{\At}}_{\mathcal{C}} \bot$ for arbitrary  $\Theta_{\At}, \Omega_{\At}$, and $\Vdash^{\Gamma_{\At}, \Delta_{\At}, \Theta_{\At}, \Omega_{\At}}_{\mathcal{C}} \bot$, by~\ref{eq:supp-imply}, $\Vdash^{\Gamma_{\At}, \Delta_{\At}}_{\mathcal{B}} \alpha \multimap \beta$.

        \item[$\chi = 1:$] Assume $\Vdash^{\Gamma_{\At}}_{\mathcal{B}} 1$ and $\Vdash^{\Delta_{\At}}_{\mathcal{B}} 1$ for arbitrary $\Delta_{\At}$. Further assume that, for an arbitrary $\mathcal{C} \supseteq \mathcal{B}$ and arbitrary $\Theta_{\At}$, $\Vdash^{\Theta_{\At}}_{\mathcal{C}} \bot$. From $\Vdash^{\Delta_{\At}}_{\mathcal{B}} 1$ and $\Vdash^{\Theta_{\At}}_{\mathcal{C}} \bot$, by~\ref{eq:supp-1}, we obtain $\Vdash^{\Delta_{\At}, \Theta_{\At}}_{\mathcal{C}} \bot$. Since $\Vdash^{\Gamma_{\At}}_{\mathcal{B}} 1$ and $\Vdash^{\Delta_{\At}, \Theta_{\At}}_{\mathcal{C}} \bot$, by~\ref{eq:supp-1}, we obtain $\Vdash^{\Gamma_{\At}, \Delta_{\At}, \Theta_{\At}}_{\mathcal{C}} \bot$. Since $\mathcal{C} \supseteq \mathcal{B}$ such that $\Vdash^{\Theta_{\At}}_{\mathcal{C}} \bot$ for arbitrary $\Theta_{\At}$, and $\Vdash^{\Gamma_{\At}, \Delta_{\At}, \Theta_{\At}}_{\mathcal{C}} \bot$, by~\ref{eq:supp-1}, $\Vdash^{\Gamma_{\At}, \Delta_{\At}}_{\mathcal{B}} 1$.

        \item[$\chi = \alpha \parr \beta:$] Assume $\Vdash^{\Gamma_{\At}}_{\mathcal{B}} 1$ and $\Vdash^{\Delta_{\At}}_{\mathcal{B}} \alpha \parr \beta$ for arbitrary $\Delta_{\At}$. Further assume that, for an arbitrary $\mathcal{C} \supseteq \mathcal{B}$ and arbitrary $\Theta_{\At}$, $\Omega_{\At}$, $\alpha \Vdash^{\Theta_{\At}}_{\mathcal{C}} \bot$ and $\beta \Vdash^{\Omega_{\At}}_{\mathcal{C}} \bot$. From $\Vdash^{\Delta_{\At}}_{\mathcal{B}} \alpha \parr \beta$ and $\alpha \Vdash^{\Theta_{\At}}_{\mathcal{C}} \bot$ and $\beta \Vdash^{\Omega_{\At}}_{\mathcal{C}} \bot$, by~\ref{eq:supp-parr}, we obtain $\Vdash^{\Delta_{\At}, \Theta_{\At}, \Omega_{\At}}_{\mathcal{C}} \bot$. Since $\Vdash^{\Gamma_{\At}}_{\mathcal{B}} 1$ and $\Vdash^{\Delta_{\At}, \Theta_{\At}, \Omega_{\At}}_{\mathcal{C}} \bot$, by~\ref{eq:supp-1}, we obtain $\Vdash^{\Gamma_{\At}, \Delta_{\At}, \Theta_{\At}, \Omega_{\At}}_{\mathcal{C}} \bot$. Since $\mathcal{C} \supseteq \mathcal{B}$ such that $\alpha \Vdash^{\Theta_{\At}}_{\mathcal{C}} \bot$ and $\beta \Vdash^{\Omega_{\At}}_{\mathcal{C}} \bot$ for arbitrary $\Theta_{\At}$, $\Omega_{\At}$, and $\Vdash^{\Gamma_{\At}, \Delta_{\At}, \Theta_{\At}, \Omega_{\At}}_{\mathcal{C}} \bot$, by~\ref{eq:supp-parr}, $\Vdash^{\Gamma_{\At}, \Delta_{\At}}_{\mathcal{B}} \alpha \parr \beta$.

        \item[$\chi = \alpha \with \beta:$] Assume $\Vdash^{\Gamma_{\At}}_{\mathcal{B}} 1$ and $\Vdash^{\Delta_{\At}}_{\mathcal{B}} \alpha \with \beta$ for arbitrary $\Delta_{\At}$. By~\ref{eq:supp-and}, we thus obtain $\Vdash^{\Delta_{\At}}_{\mathcal{B}} \alpha$ and $\Vdash^{\Delta_{\At}}_{\mathcal{B}}\beta$. Since $\Vdash^{\Gamma_{\At}}_{\mathcal{B}} 1$ and $\Vdash^{\Delta_{\At}}_{\mathcal{B}} \alpha$, by the induction hypothesis, we obtain $\Vdash^{\Gamma_{\At}, \Delta_{\At}}_{\mathcal{B}} \alpha$. Analogously, we obtain $\Vdash^{\Gamma_{\At}, \Delta_{\At}}_{\mathcal{B}} \beta$. Now, by~\ref{eq:supp-and}, we obtain $\Vdash^{\Gamma_{\At}, \Delta_{\At}}_{\mathcal{B}} \alpha \with \beta$ as required.

        \item[$\chi = \alpha \oplus \beta:$] Assume $\Vdash^{\Gamma_{\At}}_{\mathcal{B}} 1$ and $\Vdash^{\Delta_{\At}}_{\mathcal{B}} \alpha \oplus \beta$ for arbitrary $\Delta_{\At}$. Further assume that, for an arbitrary $\mathcal{C} \supseteq \mathcal{B}$ and arbitrary $\Theta_{\At}$, $\alpha \Vdash^{\Theta_{\At}}_{\mathcal{C}} \bot$ and $\beta \Vdash^{\Theta_{\At}}_{\mathcal{C}} \bot$. From $\Vdash^{\Delta_{\At}}_{\mathcal{B}} \alpha \oplus \beta$ and $\alpha \Vdash^{\Theta_{\At}}_{\mathcal{C}} \bot$ and $\beta \Vdash^{\Theta_{\At}}_{\mathcal{C}} \bot$, by~\ref{eq:supp-plus}, we obtain $\Vdash^{\Delta_{\At}, \Theta_{\At}}_{\mathcal{C}} \bot$. Since $\Vdash^{\Gamma_{\At}}_{\mathcal{B}} 1$ and $\Vdash^{\Delta_{\At}, \Theta_{\At}}_{\mathcal{C}} \bot$, by~\ref{eq:supp-1}, we obtain $\Vdash^{\Gamma_{\At}, \Delta_{\At}, \Theta_{\At}}_{\mathcal{C}} \bot$. Since $\mathcal{C} \supseteq \mathcal{B}$ such that $\alpha \Vdash^{\Theta_{\At}}_{\mathcal{C}} \bot$ and $\beta \Vdash^{\Theta_{\At}}_{\mathcal{C}} \bot$ for arbitrary $\Theta_{\At}$, and $\Vdash^{\Gamma_{\At}, \Delta_{\At}, \Theta_{\At}}_{\mathcal{C}} \bot$, by~\ref{eq:supp-plus}, $\Vdash^{\Gamma_{\At}, \Delta_{\At}}_{\mathcal{B}} \alpha \oplus \beta$.

        \item[$\chi = \top:$] Assume $\Vdash^{\Gamma_{\At}}_{\mathcal{B}} 1$ and $\Vdash^{\Delta_{\At}}_{\mathcal{B}} \top$ for arbitrary $\Delta_{\At}$. By~\ref{eq:supp-top}, $\Vdash^{\Gamma_{\At}, \Delta_{\At}}_{\mathcal{B}} \top$.

        \item[$\chi = 0:$] Assume $\Vdash^{\Gamma_{\At}}_{\mathcal{B}} 1$ and $\Vdash^{\Delta_{\At}}_{\mathcal{B}} 0$ for arbitrary $\Delta_{\At}$. By~\ref{eq:supp-0}, from $\Vdash^{\Delta_{\At}}_{\mathcal{B}}~0$ we obtain $\Vdash^{\Delta_{\At}, \Theta_{\At}}_{\mathcal{B}} \bot$ for arbitrary $\Theta_{\At}$. Since $\Vdash^{\Gamma_{\At}}_{\mathcal{B}} 1$ and $\Vdash^{\Delta_{\At}, \Theta_{\At}}_{\mathcal{B}} \bot$, by~\ref{eq:supp-1}, we obtain $\Vdash^{\Gamma_{\At}, \Delta_{\At}, \Theta_{\At}}_{\mathcal{B}} \bot$. Since $\Theta_{\At}$ is arbitrary, we obtain $\Vdash^{\Gamma_{\At}, \Delta_{\At}}_{\mathcal{B}}~0$ by~\ref{eq:supp-0}.
    \end{description}
\end{proof}

\begin{lemma} \label{lemma:genericplus}
    If $\Vdash^{\Gamma_{\At}}_{\mathcal{B}} \phi \oplus \psi$ and $\phi \Vdash^{\Delta_{\At}}_{\mathcal{B}} \chi$ and $\psi \Vdash^{\Delta_{\At}}_{\mathcal{B}} \chi$ then $\Vdash^{\Gamma_{\At}, \Delta_{\At}}_{\mathcal{B}} \chi$.
\end{lemma}
\begin{proof}
    We shall prove the statement inductively. Note that in the case of $(\with)$, we require an induction hypothesis stating that \[\text{if } \Vdash^{\Gamma_{\At}}_{\mathcal{B}} \phi \oplus \psi \text{ and } \phi \Vdash^{\Delta_{\At}}_{\mathcal{B}} \tau \text{ and } \psi \Vdash^{\Delta_{\At}}_{\mathcal{B}} \tau \text{ then } \Vdash^{\Gamma_{\At}, \Delta_{\At}}_{\mathcal{B}} \tau\] holds true for any proper subformula $\tau$ of $\chi$, as in the proof of Lemma~\ref{lemma:monotonicity}.
    
    \begin{description}[itemsep=0.5em]
        \item[$\chi = p:$] Assume $\Vdash^{\Gamma_{\At}}_{\mathcal{B}} \phi \oplus \psi$ and $\phi \Vdash^{\Delta_{\At}}_{\mathcal{B}} p$ and $\psi \Vdash^{\Delta_{\At}}_{\mathcal{B}} p$. Now assume that, for an arbitrary $\mathcal{C} \supseteq \mathcal{B}$ and arbitrary $\Theta_{\At}$, $p, \Theta_{\At} \vdash_{\mathcal{C}} \bot$. Further assume that, for an arbitrary $\mathcal{D} \supseteq \mathcal{C}$ and arbitrary $\Sigma_{\At}$, $\Vdash^{\Sigma_{\At}}_{\mathcal{D}} \phi$. By monotonicity, $\phi \Vdash^{\Delta_{\At}}_{\mathcal{D}} p$, hence, by~\ref{eq:supp-inf}, we obtain $\Vdash^{\Delta_{\At}, \Sigma_{\At}}_{\mathcal{D}} p$. From $\Vdash^{\Delta_{\At}, \Sigma_{\At}}_{\mathcal{D}} p$ and $p, \Theta_{\At} \vdash_{\mathcal{C}} \bot$ (thus also $p, \Theta_{\At} \vdash_{\mathcal{D}} \bot$), by~\ref{eq:supp-at}, we obtain $\Delta_{\At}, \Sigma_{\At}, \Theta_{\At} \vdash_{\mathcal{D}} \bot$. Hence, $\Vdash^{\Delta_{\At}, \Sigma_{\At}, \Theta_{\At}}_{\mathcal{D}} \bot$ by Lemma~\ref{lemma:bottomisspecial}. Since $\Vdash^{\Delta_{\At}, \Sigma_{\At}, \Theta_{\At}}_{\mathcal{D}} \bot$ and $\mathcal{D} \supseteq \mathcal{C}$ such that $\Vdash^{\Sigma_{\At}}_{\mathcal{D}} \phi$ for arbitrary $\Sigma_{\At}$, by~\ref{eq:supp-inf}, $\phi \Vdash^{\Delta_{\At}, \Theta_{\At}}_{\mathcal{C}} \bot$. Analogously, for an arbitrary $\mathcal{E} \supseteq \mathcal{C}$ and arbitrary $\Pi_{\At}$ such that $\Vdash^{\Pi_{\At}}_{\mathcal{E}} \psi$, obtain $\psi \Vdash^{\Delta_{\At}, \Theta_{\At}}_{\mathcal{C}} \bot$. Since $\Vdash^{\Gamma_{\At}}_{\mathcal{B}} \phi \oplus \psi$ and $\phi \Vdash^{\Delta_{\At}, \Theta_{\At}}_{\mathcal{C}} \bot$ and $\psi \Vdash^{\Delta_{\At}, \Theta_{\At}}_{\mathcal{C}} \bot$, by~\ref{eq:supp-plus}, we obtain $\Vdash^{\Gamma_{\At}, \Delta_{\At}, \Theta_{\At}}_{\mathcal{C}} \bot$, hence $\Gamma_{\At}, \Delta_{\At}, \Theta_{\At} \vdash_{\mathcal{C}} \bot$ by Lemma~\ref{lemma:bottomisspecial}. Since $\mathcal{C} \supseteq \mathcal{B}$ such that $p, \Theta_{\At} \vdash_{\mathcal{C}} \bot$ for arbitrary $\Theta_{\At}$, and $\Gamma_{\At}, \Delta_{\At}, \Theta_{\At} \vdash_{\mathcal{C}} \bot$, by~\ref{eq:supp-at}, $\Vdash^{\Gamma_{\At}, \Delta_{\At}}_{\mathcal{B}} p$.

        \item[$\chi = \alpha \otimes \beta:$] Assume $\Vdash^{\Gamma_{\At}}_{\mathcal{B}} \phi \oplus \psi$ and $\phi \Vdash^{\Delta_{\At}}_{\mathcal{B}} \alpha \otimes \beta$ and $\psi \Vdash^{\Delta_{\At}}_{\mathcal{B}} \alpha \otimes \beta$. Now assume that, for an arbitrary $\mathcal{C} \supseteq \mathcal{B}$ and arbitrary $\Theta_{\At}$, $\alpha, \beta \Vdash^{\Theta_{\At}}_{\mathcal{C}} \bot$. Further assume that, for an arbitrary $\mathcal{D} \supseteq \mathcal{C}$ and arbitrary $\Sigma_{\At}$, $\Vdash^{\Sigma_{\At}}_{\mathcal{D}} \phi$. By monotonicity, $\phi \Vdash^{\Delta_{\At}}_{\mathcal{D}} \alpha \otimes \beta$, hence with $\Vdash^{\Sigma_{\At}}_{\mathcal{D}} \phi$, by~\ref{eq:supp-inf}, we obtain $\Vdash^{\Delta_{\At}, \Sigma_{\At}}_{\mathcal{D}} \alpha \otimes \beta$. From $\Vdash^{\Delta_{\At}, \Sigma_{\At}}_{\mathcal{D}} \alpha \otimes \beta$ and $\alpha, \beta \Vdash^{\Theta_{\At}}_{\mathcal{C}} \bot$ (thus also $\alpha, \beta \Vdash^{\Theta_{\At}}_{\mathcal{D}} \bot$), by~\ref{eq:supp-tensor}, we obtain $\Vdash^{\Delta_{\At}, \Sigma_{\At}, \Theta_{\At}}_{\mathcal{D}} \bot$. Since $\Vdash^{\Delta_{\At}, \Sigma_{\At}, \Theta_{\At}}_{\mathcal{D}} \bot$ and $\mathcal{D} \supseteq \mathcal{C}$ such that $\Vdash^{\Sigma_{\At}}_{\mathcal{D}} \phi$ for arbitrary $\Sigma_{\At}$, by~\ref{eq:supp-inf}, $\phi \Vdash^{\Delta_{\At}, \Theta_{\At}}_{\mathcal{C}} \bot$. Analogously, for an arbitrary $\mathcal{E} \supseteq \mathcal{C}$ and arbitrary $\Pi_{\At}$ such that $\Vdash^{\Pi_{\At}}_{\mathcal{E}} \psi$, obtain $\psi \Vdash^{\Delta_{\At}, \Theta_{\At}}_{\mathcal{C}} \bot$. Since $\Vdash^{\Gamma_{\At}}_{\mathcal{B}} \phi \oplus \psi$ and $\phi \Vdash^{\Delta_{\At}, \Theta_{\At}}_{\mathcal{C}} \bot$ and $\psi \Vdash^{\Delta_{\At}, \Theta_{\At}}_{\mathcal{C}} \bot$, by~\ref{eq:supp-plus}, we obtain $\Vdash^{\Gamma_{\At}, \Delta_{\At}, \Theta_{\At}}_{\mathcal{C}} \bot$. Since $\mathcal{C} \supseteq \mathcal{B}$ such that $\alpha, \beta \Vdash^{\Theta_{\At}}_{\mathcal{C}} \bot$ for arbitrary $\Theta_{\At}$, and $\Vdash^{\Gamma_{\At}, \Delta_{\At}, \Theta_{\At}}_{\mathcal{C}} \bot$, by~\ref{eq:supp-tensor}, $\Vdash^{\Gamma_{\At}, \Delta_{\At}}_{\mathcal{B}} \alpha \otimes \beta$.

        \item[$\chi = \alpha \multimap \beta:$] Assume $\Vdash^{\Gamma_{\At}}_{\mathcal{B}} \phi \oplus \psi$ and $\phi \Vdash^{\Delta_{\At}}_{\mathcal{B}} \alpha \multimap \beta$ and $\psi \Vdash^{\Delta_{\At}}_{\mathcal{B}} \alpha \multimap \beta$. Further assume that, for an arbitrary $\mathcal{C} \supseteq \mathcal{B}$ and arbitrary $\Theta_{\At}, \Omega_{\At}$, $\Vdash^{\Theta_{\At}}_{\mathcal{C}} \alpha$ and $\beta \Vdash^{\Omega_{\At}}_{\mathcal{C}} \bot$. Further assume that, for an arbitrary $\mathcal{D} \supseteq \mathcal{C}$ and arbitrary $\Sigma_{\At}$, $\Vdash^{\Sigma_{\At}}_{\mathcal{D}} \phi$. By monotonicity, $\phi \Vdash^{\Delta_{\At}}_{\mathcal{D}} \alpha \multimap \beta$, hence with $\Vdash^{\Sigma_{\At}}_{\mathcal{D}} \phi$, by~\ref{eq:supp-inf}, we obtain $\Vdash^{\Delta_{\At}, \Sigma_{\At}}_{\mathcal{D}} \alpha \multimap \beta$. From $\Vdash^{\Delta_{\At}, \Sigma_{\At}}_{\mathcal{D}} \alpha \multimap \beta$, $\Vdash^{\Theta_{\At}}_{\mathcal{C}} \alpha$ and $\beta \Vdash^{\Omega_{\At}}_{\mathcal{C}} \bot$ (thus also $\Vdash^{\Theta_{\At}}_{\mathcal{D}} \alpha$ and $\beta \Vdash^{\Omega_{\At}}_{\mathcal{D}} \bot$), by~\ref{eq:supp-imply}, we obtain $\Vdash^{\Delta_{\At}, \Sigma_{\At}, \Theta_{\At}, \Omega_{\At}}_{\mathcal{D}} \bot$. Since $\Vdash^{\Delta_{\At}, \Sigma_{\At}, \Theta_{\At}, \Omega_{\At}}_{\mathcal{D}} \bot$ and $\mathcal{D} \supseteq \mathcal{C}$ such that $\Vdash^{\Sigma_{\At}}_{\mathcal{D}} \phi$ for arbitrary $\Sigma_{\At}$, by~\ref{eq:supp-inf}, $\phi \Vdash^{\Delta_{\At}, \Theta_{\At}, \Omega_{\At}}_{\mathcal{C}} \bot$. Analogously, for an arbitrary $\mathcal{E} \supseteq \mathcal{C}$ and arbitrary $\Pi_{\At}$ such that $\Vdash^{\Pi_{\At}}_{\mathcal{E}} \psi$, obtain $\psi \Vdash^{\Delta_{\At}, \Theta_{\At}, \Omega_{\At}}_{\mathcal{C}} \bot$.  Since $\Vdash^{\Gamma_{\At}}_{\mathcal{B}} \phi \oplus \psi$ and $\phi \Vdash^{\Delta_{\At}, \Theta_{\At}, \Omega_{\At}}_{\mathcal{C}} \bot$ and $\psi \Vdash^{\Delta_{\At}, \Theta_{\At}, \Omega_{\At}}_{\mathcal{C}} \bot$, by~\ref{eq:supp-plus}, we obtain $\Vdash^{\Gamma_{\At}, \Delta_{\At}, \Theta_{\At}, \Omega_{\At}}_{\mathcal{C}} \bot$. Since $\mathcal{C} \supseteq \mathcal{B}$ such that $\Vdash^{\Theta_{\At}}_{\mathcal{C}} \alpha$ and $\beta \Vdash^{\Omega_{\At}}_{\mathcal{C}} \bot$ for arbitrary  $\Theta_{\At}, \Omega_{\At}$, and $\Vdash^{\Gamma_{\At}, \Delta_{\At}, \Theta_{\At}, \Omega_{\At}}_{\mathcal{C}} \bot$, by~\ref{eq:supp-imply}, $\Vdash^{\Gamma_{\At}, \Delta_{\At}}_{\mathcal{B}} \alpha \multimap \beta$.

        \item[$\chi = 1:$] Assume $\Vdash^{\Gamma_{\At}}_{\mathcal{B}} \phi \oplus \psi$ and $\phi \Vdash^{\Delta_{\At}}_{\mathcal{B}} 1$ and $\psi \Vdash^{\Delta_{\At}}_{\mathcal{B}} 1$. Now assume that, for an arbitrary $\mathcal{C} \supseteq \mathcal{B}$ and arbitrary $\Theta_{\At}$, $\Vdash^{\Theta_{\At}}_{\mathcal{C}} \bot$. Further assume that, for an arbitrary $\mathcal{D} \supseteq \mathcal{C}$ and arbitrary $\Sigma_{\At}$, $\Vdash^{\Sigma_{\At}}_{\mathcal{D}} \phi$. By monotonicity, $\phi \Vdash^{\Delta_{\At}}_{\mathcal{D}} 1$, hence with $\Vdash^{\Sigma_{\At}}_{\mathcal{D}} \phi$, by~\ref{eq:supp-inf}, we obtain $\Vdash^{\Delta_{\At}, \Sigma_{\At}}_{\mathcal{D}} 1$. From $\Vdash^{\Delta_{\At}, \Sigma_{\At}}_{\mathcal{D}} 1$ and $\Vdash^{\Theta_{\At}}_{\mathcal{C}} \bot$ (thus also $\Vdash^{\Theta_{\At}}_{\mathcal{D}} \bot$), by~\ref{eq:supp-1}, we obtain $\Vdash^{\Delta_{\At}, \Sigma_{\At}, \Theta_{\At}}_{\mathcal{D}} \bot$. Since $\Vdash^{\Delta_{\At}, \Sigma_{\At}, \Theta_{\At}}_{\mathcal{D}} \bot$ and $\mathcal{D} \supseteq \mathcal{C}$ such that $\Vdash^{\Sigma_{\At}}_{\mathcal{D}} \phi$ for arbitrary $\Sigma_{\At}$, by~\ref{eq:supp-inf}, $\phi \Vdash^{\Delta_{\At}, \Theta_{\At}}_{\mathcal{C}} \bot$. Analogously, for an arbitrary $\mathcal{E} \supseteq \mathcal{C}$ and arbitrary $\Pi_{\At}$ such that $\Vdash^{\Pi_{\At}}_{\mathcal{E}} \psi$, obtain $\psi \Vdash^{\Delta_{\At}, \Theta_{\At}}_{\mathcal{C}} \bot$. Since $\Vdash^{\Gamma_{\At}}_{\mathcal{B}} \phi \oplus \psi$ and $\phi \Vdash^{\Delta_{\At}, \Theta_{\At}}_{\mathcal{C}} \bot$ and $\psi \Vdash^{\Delta_{\At}, \Theta_{\At}}_{\mathcal{C}} \bot$, by~\ref{eq:supp-plus}, we obtain $\Vdash^{\Gamma_{\At}, \Delta_{\At}, \Theta_{\At}}_{\mathcal{C}} \bot$. Since $\mathcal{C} \supseteq \mathcal{B}$ such that $\Vdash^{\Theta_{\At}}_{\mathcal{C}} \bot$ for arbitrary $\Theta_{\At}$, and $\Vdash^{\Gamma_{\At}, \Delta_{\At}, \Theta_{\At}}_{\mathcal{C}} \bot$, by~\ref{eq:supp-1}, $\Vdash^{\Gamma_{\At}, \Delta_{\At}}_{\mathcal{B}} 1$.

        \item[$\chi = \alpha \parr \beta:$] Assume $\Vdash^{\Gamma_{\At}}_{\mathcal{B}} \phi \oplus \psi$ and $\phi \Vdash^{\Delta_{\At}}_{\mathcal{B}} \alpha \parr \beta$ and $\psi \Vdash^{\Delta_{\At}}_{\mathcal{B}} \alpha \parr \beta$. Now assume that, for an arbitrary $\mathcal{C} \supseteq \mathcal{B}$ and arbitrary $\Theta_{\At}$, $\Omega_{\At}$, $\alpha \Vdash^{\Theta_{\At}}_{\mathcal{C}} \bot$ and $\beta \Vdash^{\Omega_{\At}}_{\mathcal{C}} \bot$. Further assume that, for an arbitrary $\mathcal{D} \supseteq \mathcal{C}$ and arbitrary $\Sigma_{\At}$, $\Vdash^{\Sigma_{\At}}_{\mathcal{D}} \phi$. By monotonicity, $\phi \Vdash^{\Delta_{\At}}_{\mathcal{D}} \alpha \parr \beta$, hence with $\Vdash^{\Sigma_{\At}}_{\mathcal{D}} \phi$, by~\ref{eq:supp-inf}, we obtain $\Vdash^{\Delta_{\At}, \Sigma_{\At}}_{\mathcal{D}} \alpha \parr \beta$. From $\Vdash^{\Delta_{\At}, \Sigma_{\At}}_{\mathcal{D}} \alpha \parr \beta$ and $\alpha \Vdash^{\Theta_{\At}}_{\mathcal{C}} \bot$ and $\beta \Vdash^{\Omega_{\At}}_{\mathcal{C}} \bot$ (thus also $\alpha \Vdash^{\Theta_{\At}}_{\mathcal{D}} \bot$ and $\beta \Vdash^{\Omega_{\At}}_{\mathcal{D}} \bot$), by~\ref{eq:supp-parr}, we obtain $\Vdash^{\Delta_{\At}, \Sigma_{\At}, \Theta_{\At}, \Omega_{\At}}_{\mathcal{D}} \bot$. Since $\Vdash^{\Delta_{\At}, \Sigma_{\At}, \Theta_{\At}, \Omega_{\At}}_{\mathcal{D}} \bot$ and $\mathcal{D} \supseteq \mathcal{C}$ such that $\Vdash^{\Sigma_{\At}}_{\mathcal{D}} \phi$ for arbitrary $\Sigma_{\At}$, by~\ref{eq:supp-inf}, $\phi \Vdash^{\Delta_{\At}, \Theta_{\At}, \Omega_{\At}}_{\mathcal{C}} \bot$. Analogously, for an arbitrary $\mathcal{E} \supseteq \mathcal{C}$ and arbitrary $\Pi_{\At}$ such that $\Vdash^{\Pi_{\At}}_{\mathcal{E}} \psi$, obtain $\psi \Vdash^{\Delta_{\At}, \Theta_{\At}, \Omega_{\At}}_{\mathcal{C}} \bot$. Since $\Vdash^{\Gamma_{\At}}_{\mathcal{B}} \phi \oplus \psi$ and $\phi \Vdash^{\Delta_{\At}, \Theta_{\At}, \Omega_{\At}}_{\mathcal{C}} \bot$ and $\psi \Vdash^{\Delta_{\At}, \Theta_{\At}, \Omega_{\At}}_{\mathcal{C}} \bot$, by~\ref{eq:supp-plus}, we obtain $\Vdash^{\Gamma_{\At}, \Delta_{\At}, \Theta_{\At}, \Omega_{\At}}_{\mathcal{C}} \bot$. Since $\mathcal{C} \supseteq \mathcal{B}$ such that $\alpha \Vdash^{\Theta_{\At}}_{\mathcal{C}} \bot$ and $\beta \Vdash^{\Omega_{\At}}_{\mathcal{C}} \bot$ for arbitrary $\Theta_{\At}$, $\Omega_{\At}$, and $\Vdash^{\Gamma_{\At}, \Delta_{\At}, \Theta_{\At}, \Omega_{\At}}_{\mathcal{C}} \bot$, by~\ref{eq:supp-parr}, $\Vdash^{\Gamma_{\At}, \Delta_{\At}}_{\mathcal{B}} \alpha \parr \beta$.

        \item[$\chi = \alpha \with \beta:$] Assume $\Vdash^{\Gamma_{\At}}_{\mathcal{B}} \phi \oplus \psi$ and $\phi \Vdash^{\Delta_{\At}}_{\mathcal{B}} \alpha \with \beta$ and $\psi \Vdash^{\Delta_{\At}}_{\mathcal{B}} \alpha \with \beta$. Further assume that, for an arbitrary $\mathcal{C} \supseteq \mathcal{B}$ and arbitrary $\Sigma_{\At}$, $\Vdash^{\Sigma_{\At}}_{\mathcal{C}} \phi$. By monotonicity, $\phi \Vdash^{\Delta_{\At}}_{\mathcal{C}} \alpha \with \beta$, hence with $\Vdash^{\Sigma_{\At}}_{\mathcal{C}} \phi$, by~\ref{eq:supp-inf}, we obtain $\Vdash^{\Delta_{\At}, \Sigma_{\At}}_{\mathcal{C}} \alpha \with \beta$. By~\ref{eq:supp-and}, we thus obtain $\Vdash^{\Delta_{\At}, \Sigma_{\At}}_{\mathcal{C}} \alpha$ and $\Vdash^{\Delta_{\At}, \Sigma_{\At}}_{\mathcal{C}}\beta$. Since $\mathcal{C} \supseteq \mathcal{B}$ such that $\Vdash^{\Sigma_{\At}}_{\mathcal{C}} \phi$ for arbitrary $\Sigma_{\At}$, by~\ref{eq:supp-inf}, $\phi \Vdash^{\Delta_{\At}}_{\mathcal{B}} \alpha$ and $\phi \Vdash^{\Delta_{\At}}_{\mathcal{B}} \beta$. Analogously, for an arbitrary $\mathcal{D} \supseteq \mathcal{B}$ and arbitrary $\Pi_{\At}$ such that $\Vdash^{\Pi_{\At}}_{\mathcal{D}} \psi$, obtain $\psi \Vdash^{\Delta_{\At}}_{\mathcal{B}} \alpha$ and $\psi \Vdash^{\Delta_{\At}}_{\mathcal{B}} \beta$. Since $\Vdash^{\Gamma_{\At}}_{\mathcal{B}} \phi \oplus \psi$ and $\phi \Vdash^{\Delta_{\At}}_{\mathcal{B}} \alpha$ and $\psi \Vdash^{\Delta_{\At}}_{\mathcal{B}} \alpha$, by the induction hypothesis, we obtain $\Vdash^{\Gamma_{\At}, \Delta_{\At}}_{\mathcal{B}} \alpha$. Analogously, we obtain $\Vdash^{\Gamma_{\At}, \Delta_{\At}}_{\mathcal{B}} \beta$. Now, by~\ref{eq:supp-and}, we obtain $\Vdash^{\Gamma_{\At}, \Delta_{\At}}_{\mathcal{B}} \alpha \with \beta$ as required.

        \item[$\chi = \alpha \oplus \beta:$] Assume $\Vdash^{\Gamma_{\At}}_{\mathcal{B}} \phi \oplus \psi$ and $\phi \Vdash^{\Delta_{\At}}_{\mathcal{B}} \alpha \oplus \beta$ and $\psi \Vdash^{\Delta_{\At}}_{\mathcal{B}} \alpha \oplus \beta$. Now assume that, for an arbitrary $\mathcal{C} \supseteq \mathcal{B}$ and arbitrary $\Theta_{\At}$, $\alpha \Vdash^{\Theta_{\At}}_{\mathcal{C}} \bot$ and $\beta \Vdash^{\Theta_{\At}}_{\mathcal{C}} \bot$. Further assume that, for an arbitrary $\mathcal{D} \supseteq \mathcal{C}$ and arbitrary $\Sigma_{\At}$, $\Vdash^{\Sigma_{\At}}_{\mathcal{D}} \phi$. By monotonicity, $\phi \Vdash^{\Delta_{\At}}_{\mathcal{D}} \alpha \oplus \beta$, hence with $\Vdash^{\Sigma_{\At}}_{\mathcal{D}} \phi$, by~\ref{eq:supp-inf}, we obtain $\Vdash^{\Delta_{\At}, \Sigma_{\At}}_{\mathcal{D}} \alpha \oplus \beta$. From $\Vdash^{\Delta_{\At}, \Sigma_{\At}}_{\mathcal{D}} \alpha \oplus \beta$ and $\alpha \Vdash^{\Theta_{\At}}_{\mathcal{C}} \bot$ and $\beta \Vdash^{\Theta_{\At}}_{\mathcal{C}} \bot$ (thus also $\alpha \Vdash^{\Theta_{\At}}_{\mathcal{D}} \bot$ and $\beta \Vdash^{\Theta_{\At}}_{\mathcal{D}} \bot$), by~\ref{eq:supp-plus}, we obtain $\Vdash^{\Delta_{\At}, \Sigma_{\At}, \Theta_{\At}}_{\mathcal{D}} \bot$. Since $\Vdash^{\Delta_{\At}, \Sigma_{\At}, \Theta_{\At}}_{\mathcal{D}} \bot$ and $\mathcal{D} \supseteq \mathcal{C}$ such that $\Vdash^{\Sigma_{\At}}_{\mathcal{D}} \phi$ for arbitrary $\Sigma_{\At}$, by~\ref{eq:supp-inf}, $\phi \Vdash^{\Delta_{\At}, \Theta_{\At}}_{\mathcal{C}} \bot$. Analogously, for an arbitrary $\mathcal{E} \supseteq \mathcal{C}$ and arbitrary $\Pi_{\At}$ such that $\Vdash^{\Pi_{\At}}_{\mathcal{E}} \psi$, obtain $\psi \Vdash^{\Delta_{\At}, \Theta_{\At}}_{\mathcal{C}} \bot$. Since $\Vdash^{\Gamma_{\At}}_{\mathcal{B}} \phi \oplus \psi$ and $\phi \Vdash^{\Delta_{\At}, \Theta_{\At}}_{\mathcal{C}} \bot$ and $\psi \Vdash^{\Delta_{\At}, \Theta_{\At}}_{\mathcal{C}} \bot$, by~\ref{eq:supp-plus}, we obtain $\Vdash^{\Gamma_{\At}, \Delta_{\At}, \Theta_{\At}}_{\mathcal{C}} \bot$. Since $\mathcal{C} \supseteq \mathcal{B}$ such that $\alpha \Vdash^{\Theta_{\At}}_{\mathcal{C}} \bot$ and $\beta \Vdash^{\Theta_{\At}}_{\mathcal{C}} \bot$ for arbitrary $\Theta_{\At}$, and $\Vdash^{\Gamma_{\At}, \Delta_{\At}, \Theta_{\At}}_{\mathcal{C}} \bot$, by~\ref{eq:supp-plus}, $\Vdash^{\Gamma_{\At}, \Delta_{\At}}_{\mathcal{B}} \alpha \oplus \beta$.

        \item[$\chi = \top:$] Assume $\Vdash^{\Gamma_{\At}}_{\mathcal{B}} \phi \oplus \psi$ and $\phi \Vdash^{\Delta_{\At}}_{\mathcal{B}} \top$ and $\psi \Vdash^{\Delta_{\At}}_{\mathcal{B}} \top$. By~\ref{eq:supp-top}, $\Vdash^{\Gamma_{\At}, \Delta_{\At}}_{\mathcal{B}} \top$.

        \item[$\chi = 0:$] Assume $\Vdash^{\Gamma_{\At}}_{\mathcal{B}} \phi \oplus \psi$ and $\phi \Vdash^{\Delta_{\At}}_{\mathcal{B}} 0$ and $\psi \Vdash^{\Delta_{\At}}_{\mathcal{B}} 0$ for arbitrary $\Delta_{\At}$. Further assume that, for an arbitrary $\mathcal{C} \supseteq \mathcal{B}$ and arbitrary $\Sigma_{\At}$, $\Vdash^{\Sigma_{\At}}_{\mathcal{C}} \phi$. Since $\phi \Vdash^{\Delta_{\At}}_{\mathcal{B}} 0$ and $\Vdash^{\Sigma_{\At}}_{\mathcal{C}} \phi$, by~\ref{eq:supp-inf}, we obtain $\Vdash^{\Delta_{\At}, \Sigma_{\At}}_{\mathcal{C}} 0$. By~\ref{eq:supp-0}, we thus obtain $\Vdash^{\Delta_{\At}, \Sigma_{\At}, \Theta_{\At}}_{\mathcal{C}} \bot$ for arbitrary $\Theta_{\At}$. Since $\mathcal{C} \supseteq \mathcal{B}$ such that  $\Vdash^{\Sigma_{\At}}_{\mathcal{C}} \phi$ for arbitrary $\Sigma_{\At}$, by~\ref{eq:supp-inf}, $\phi \Vdash^{\Delta_{\At}, \Theta_{\At}}_{\mathcal{B}} \bot$. Analogously, for an arbitrary $\mathcal{D} \supseteq \mathcal{B}$ and arbitrary $\Pi_{\At}$ such that $\Vdash^{\Pi_{\At}}_{\mathcal{D}} \psi$, obtain $\psi \Vdash^{\Delta_{\At}, \Theta_{\At}}_{\mathcal{B}} \bot$.  Since $\Vdash^{\Gamma_{\At}}_{\mathcal{B}} \phi \oplus \psi$ and $\phi \Vdash^{\Delta_{\At}, \Theta_{\At}}_{\mathcal{B}} \bot$ and $\psi \Vdash^{\Delta_{\At}, \Theta_{\At}}_{\mathcal{B}} \bot$, by~\ref{eq:supp-plus}, we obtain $\Vdash^{\Gamma_{\At}, \Delta_{\At}, \Theta_{\At}}_{\mathcal{C}} \bot$.Since $\Theta_{\At}$ is arbitrary, we obtain $\Vdash^{\Gamma_{\At}, \Delta_{\At}}_{\mathcal{B}} 0$ by~\ref{eq:supp-0}.
    \end{description}
\end{proof}

\begin{lemma} \label{lemma:genericzero}
    If $\Vdash^{\Gamma_{\At}}_{\mathcal{B}} 0$ then $\Vdash^{\Gamma_{\At}, \Delta_{\At}}_{\mathcal{B}} \chi$.
\end{lemma}
\begin{proof}
    We shall prove the statement inductively. Note that in the case of $(\with)$, we require an induction hypothesis stating that \[\text{if } \Vdash^{\Gamma_{\At}}_{\mathcal{B}} 0 \text{ then } \Vdash^{\Gamma_{\At}, \Delta_{\At}}_{\mathcal{B}} \tau\] holds true for any proper subformula $\tau$ of $\chi$, as in the proof of Lemma~\ref{lemma:monotonicity}.
    
    \begin{description}[itemsep=0.5em]
        \item[$\chi = p:$] Assume $\Vdash^{\Gamma_{\At}}_{\mathcal{B}} 0$. Further assume that, for an arbitrary $\mathcal{C} \supseteq \mathcal{B}$ and arbitrary $\Theta_{\At}$, $p, \Theta_{\At}, \Delta_{\At} \vdash_{\mathcal{C}} \bot$. Now, from $\Vdash^{\Gamma_{\At}}_{\mathcal{B}} 0$ we obtain $\Vdash^{\Gamma_{\At}, \Sigma_{\At}}_{\mathcal{B}} \bot$ for all $\Sigma_{\At}$ by~\ref{eq:supp-0}. In particular, then, let $\Sigma_{\At} = \Delta_{\At} \cup \Theta_{\At}$ for arbitrary $\Delta_{\At}$, hence $\Vdash^{\Gamma_{\At}, \Delta_{\At}, \Theta_{\At}}_{\mathcal{B}} \bot$, hence $\Vdash^{\Gamma_{\At}, \Delta_{\At}, \Theta_{\At}}_{\mathcal{C}} \bot$  by monotonicity, hence $\Gamma_{\At}, \Delta_{\At}, \Theta_{\At} \vdash_{\mathcal{C}} \bot$ by Lemma~\ref{lemma:bottomisspecial}. Since $\mathcal{C} \supseteq \mathcal{B}$ such that $p, \Theta_{\At} \vdash_{\mathcal{C}} \bot$ for arbitrary $\Theta_{\At}$, and $\Gamma_{\At}, \Delta_{\At}, \Theta_{\At} \vdash_{\mathcal{C}} \bot$, by~\ref{eq:supp-at}, $\Vdash^{\Gamma_{\At}, \Delta_{\At}}_{\mathcal{B}} p$.

        \item[$\chi = \alpha \otimes \beta:$] Assume $\Vdash^{\Gamma_{\At}}_{\mathcal{B}} 0$. Further assume that, for an arbitrary $\mathcal{C} \supseteq \mathcal{B}$ and arbitrary $\Theta_{\At}$, $\alpha, \beta \Vdash^{\Theta_{\At}}_{\mathcal{C}} \bot$. Now, from $\Vdash^{\Gamma_{\At}}_{\mathcal{B}} 0$ we obtain $\Vdash^{\Gamma_{\At}, \Sigma_{\At}}_{\mathcal{B}} \bot$ for all $\Sigma_{\At}$ by~\ref{eq:supp-0}. In particular, then, let $\Sigma_{\At} = \Delta_{\At} \cup \Theta_{\At}$ for arbitrary $\Delta_{\At}$, hence $\Vdash^{\Gamma_{\At}, \Delta_{\At} \Theta_{\At}}_{\mathcal{B}} \bot$, hence $\Vdash^{\Gamma_{\At}, \Delta_{\At}, \Theta_{\At}}_{\mathcal{C}} \bot$ by monotonicity. Since $\mathcal{C} \supseteq \mathcal{B}$ such that $\alpha, \beta \Vdash^{\Theta_{\At}}_{\mathcal{C}} \bot$ for arbitrary $\Theta_{\At}$, and $\Vdash^{\Gamma_{\At}, \Delta_{\At}, \Theta_{\At}}_{\mathcal{C}} \bot$, by~\ref{eq:supp-tensor}, $\Vdash^{\Gamma_{\At}, \Delta_{\At}}_{\mathcal{B}} \alpha \otimes \beta$.

        \item[$\chi = \alpha \multimap \beta:$] Assume $\Vdash^{\Gamma_{\At}}_{\mathcal{B}} 0$. Further assume that, for an arbitrary $\mathcal{C} \supseteq \mathcal{B}$ and arbitrary $\Theta_{\At}, \Omega_{\At}$, $\Vdash^{\Theta_{\At}}_{\mathcal{C}} \alpha$ and $\beta \Vdash^{\Omega_{\At}}_{\mathcal{C}} \bot$. Now, from $\Vdash^{\Gamma_{\At}}_{\mathcal{B}} 0$ we obtain $\Vdash^{\Gamma_{\At}, \Sigma_{\At}}_{\mathcal{B}} \bot$ for all $\Sigma_{\At}$ by~\ref{eq:supp-0}. In particular, then, let $\Sigma_{\At} = \Delta_{\At} \cup \Theta_{\At} \cup \Omega_{\At}$ for arbitrary $\Delta_{\At}$, hence $\Vdash^{\Gamma_{\At}, \Delta_{\At} \Theta_{\At}, \Omega_{\At}}_{\mathcal{B}} \bot$, hence $\Vdash^{\Gamma_{\At}, \Delta_{\At}, \Theta_{\At}, \Omega_{\At}}_{\mathcal{C}} \bot$ by monotonicity. Since $\mathcal{C} \supseteq \mathcal{B}$ such that $\Vdash^{\Theta_{\At}}_{\mathcal{C}} \alpha$ and $\beta \Vdash^{\Omega_{\At}}_{\mathcal{C}} \bot$ for arbitrary  $\Theta_{\At}, \Omega_{\At}$, and $\Vdash^{\Gamma_{\At}, \Delta_{\At}, \Theta_{\At}, \Omega_{\At}}_{\mathcal{C}} \bot$, by~\ref{eq:supp-imply}, $\Vdash^{\Gamma_{\At}, \Delta_{\At}}_{\mathcal{B}} \alpha \multimap \beta$.

        \item[$\chi = 1:$] Assume $\Vdash^{\Gamma_{\At}}_{\mathcal{B}} 0$. Further assume that, for an arbitrary $\mathcal{C} \supseteq \mathcal{B}$ and arbitrary $\Theta_{\At}$, $\Vdash^{\Theta_{\At}}_{\mathcal{C}} \bot$. Now, from $\Vdash^{\Gamma_{\At}}_{\mathcal{B}} 0$ we obtain $\Vdash^{\Gamma_{\At}, \Sigma_{\At}}_{\mathcal{B}} \bot$ for all $\Sigma_{\At}$ by~\ref{eq:supp-0}. In particular, then, let $\Sigma_{\At} = \Delta_{\At} \cup \Theta_{\At}$ for arbitrary $\Delta_{\At}$, hence $\Vdash^{\Gamma_{\At}, \Delta_{\At} \Theta_{\At}}_{\mathcal{B}} \bot$, hence $\Vdash^{\Gamma_{\At}, \Delta_{\At}, \Theta_{\At}}_{\mathcal{C}} \bot$ by monotonicity. Since $\mathcal{C} \supseteq \mathcal{B}$ such that $\Vdash^{\Theta_{\At}}_{\mathcal{C}} \bot$ for arbitrary $\Theta_{\At}$, and $\Vdash^{\Gamma_{\At}, \Delta_{\At}, \Theta_{\At}}_{\mathcal{C}} \bot$, by~\ref{eq:supp-1}, $\Vdash^{\Gamma_{\At}, \Delta_{\At}}_{\mathcal{B}} 1$.

        \item[$\chi = \alpha \parr \beta:$] Assume $\Vdash^{\Gamma_{\At}}_{\mathcal{B}} 0$. Further assume that, for an arbitrary $\mathcal{C} \supseteq \mathcal{B}$ and arbitrary $\Theta_{\At}$, $\Omega_{\At}$, $\alpha \Vdash^{\Theta_{\At}}_{\mathcal{C}} \bot$ and $\beta \Vdash^{\Omega_{\At}}_{\mathcal{C}} \bot$. Now, from $\Vdash^{\Gamma_{\At}}_{\mathcal{B}} 0$ we obtain $\Vdash^{\Gamma_{\At}, \Sigma_{\At}}_{\mathcal{B}} \bot$ for all $\Sigma_{\At}$ by~\ref{eq:supp-0}. In particular, then, let $\Sigma_{\At} = \Delta_{\At} \cup \Theta_{\At} \cup \Omega_{\At}$ for arbitrary $\Delta_{\At}$, hence $\Vdash^{\Gamma_{\At}, \Delta_{\At} \Theta_{\At}, \Omega_{\At}}_{\mathcal{B}} \bot$, hence $\Vdash^{\Gamma_{\At}, \Delta_{\At}, \Theta_{\At}, \Omega_{\At}}_{\mathcal{C}} \bot$ by monotonicity. Since $\mathcal{C} \supseteq \mathcal{B}$ such that $\alpha \Vdash^{\Theta_{\At}}_{\mathcal{C}} \bot$ and $\beta \Vdash^{\Omega_{\At}}_{\mathcal{C}} \bot$ for arbitrary $\Theta_{\At}$, $\Omega_{\At}$, and $\Vdash^{\Gamma_{\At}, \Delta_{\At}, \Theta_{\At}, \Omega_{\At}}_{\mathcal{C}} \bot$, by~\ref{eq:supp-parr}, $\Vdash^{\Gamma_{\At}, \Delta_{\At}}_{\mathcal{B}} \alpha \parr \beta$.

        \item[$\chi = \alpha \with \beta:$] Assume $\Vdash^{\Gamma_{\At}}_{\mathcal{B}} 0$. By the induction hypothesis, we obtain $\Vdash^{\Gamma_{\At}, \Delta_{\At}}_{\mathcal{B}} \alpha$ and $\Vdash^{\Gamma_{\At}, \Delta_{\At}}_{\mathcal{B}}\beta$. Now, by~\ref{eq:supp-and}, we obtain $\Vdash^{\Gamma_{\At}, \Delta_{\At}}_{\mathcal{B}} \alpha \with \beta$ as required.

        \item[$\chi = \alpha \oplus \beta:$] Assume $\Vdash^{\Gamma_{\At}}_{\mathcal{B}} 0$. Further assume that, for an arbitrary $\mathcal{C} \supseteq \mathcal{B}$ and arbitrary $\Theta_{\At}$, $\alpha \Vdash^{\Theta_{\At}}_{\mathcal{C}} \bot$ and $\beta \Vdash^{\Theta_{\At}}_{\mathcal{C}} \bot$. Now, from $\Vdash^{\Gamma_{\At}}_{\mathcal{B}} 0$ we obtain $\Vdash^{\Gamma_{\At}, \Sigma_{\At}}_{\mathcal{B}} \bot$ for all $\Sigma_{\At}$ by~\ref{eq:supp-0}. In particular, then, let $\Sigma_{\At} = \Delta_{\At} \cup \Theta_{\At}$ for arbitrary $\Delta_{\At}$, hence $\Vdash^{\Gamma_{\At}, \Delta_{\At} \Theta_{\At}}_{\mathcal{B}} \bot$, hence $\Vdash^{\Gamma_{\At}, \Delta_{\At}, \Theta_{\At}}_{\mathcal{C}} \bot$ by monotonicity. Since $\mathcal{C} \supseteq \mathcal{B}$ such that $\alpha \Vdash^{\Theta_{\At}}_{\mathcal{C}} \bot$ and $\beta \Vdash^{\Theta_{\At}}_{\mathcal{C}} \bot$ for arbitrary $\Theta_{\At}$, and $\Vdash^{\Gamma_{\At}, \Delta_{\At}, \Theta_{\At}}_{\mathcal{C}} \bot$, by~\ref{eq:supp-plus}, $\Vdash^{\Gamma_{\At}, \Delta_{\At}}_{\mathcal{B}} \alpha \oplus \beta$.

        \item[$\chi = \top:$] Assume $\Vdash^{\Gamma_{\At}}_{\mathcal{B}} 0$. By~\ref{eq:supp-top}, $\Vdash^{\Gamma_{\At}, \Delta_{\At}}_{\mathcal{B}} \top$.

        \item[$\chi = 0:$] Assume $\Vdash^{\Gamma_{\At}}_{\mathcal{B}} 0$. Now, by~\ref{eq:supp-0}, from $\Vdash^{\Gamma_{\At}}_{\mathcal{B}} 0$ we obtain $\Vdash^{\Gamma_{\At}, \Sigma_{\At}}_{\mathcal{B}} \bot$ for all $\Sigma_{\At}$. In particular, then, let $\Sigma_{\At} = \Delta_{\At} \cup \Theta_{\At}$ for arbitrary $\Delta_{\At}, \Theta_{\At}$, hence $\Vdash^{\Gamma_{\At}, \Delta_{\At} \Theta_{\At}}_{\mathcal{B}} \bot$. Since $\Theta_{\At}$ is arbitrary, by~\ref{eq:supp-0} we obtain $\Vdash^{\Gamma_{\At}, \Delta_{\At}}_{\mathcal{B}} 0$.
    \end{description}
\end{proof}

\section{Soundness}\label{sec:soundness}
In this section, we show that {\MALL} is sound with respect to our semantics -- meaning that every provable formula is genuinely valid.  In other words, we will prove that if $\Gamma\vdash_ {\MALL}\phi$ then  $\Gamma\Vdash\phi$. This follows from the semantic \emph{reductio ad absurdum}, shown next, along with the fact that $\Vdash$ respects $\MALL$ inference rules.

\begin{lemma} \label{lemma:genericraa}
    If $\Gamma, \neg \phi \Vdash_{\mathcal{B}} \bot$ then $\Gamma \Vdash_{\mathcal{B}} \phi$.
\end{lemma}

\begin{proof}
Let $\Gamma = \{\psi^1,\dots,\psi^n\}$ and assume that, for an arbitrary $\mathcal{C} \supseteq \mathcal{B}$, for all $ \psi^i \in \Gamma$ $(1 \leq i \leq n)$ and arbitrary multisets $\Theta^{i}_{\At}$, $\Vdash^{\Theta^{i}_{\At}}_{\mathcal{C}} \psi^i$. 
The statement is proved inductively. 
We will illustrate the proof for $\parr,\&$ and $0$. The remaining cases are similar. 
    \begin{description}[itemsep=0.4em]
                \item[$\phi = \alpha \parr \beta:$] Assume $\Gamma, \neg (\alpha \parr \beta) \Vdash_{\mathcal{B}} \bot$. Since $\Vdash^{\Theta^{i}_{\At}}_{\mathcal{C}} \psi^i$, $\forall \psi^i \in \Gamma$, we obtain $\neg (\alpha \parr \beta) \Vdash^{\Theta^{1}_{\At},\dots,\Theta^{n}_{\At}}_{\mathcal{C}} \bot$ by Lemma~\ref{lemma:partialinf}. 
 Now assume that, for an arbitrary $\mathcal{D} \supseteq \mathcal{C}$ and arbitrary $\Sigma_{\At}, \Omega_{\At}$, $\alpha\Vdash^{\Sigma_{\At}}_{\mathcal{D}} \bot$ and $\beta \Vdash^{\Omega_{\At}}_{\mathcal{D}} \bot$. Further assume that, for an arbitrary $\mathcal{E} \supseteq \mathcal{D}$ and arbitrary $\Pi_{\At}$, $\Vdash^{\Pi_{\At}}_{\mathcal{E}} \alpha \parr \beta$. 
Since $\Vdash^{\Pi_{\At}}_{\mathcal{E}} \alpha \parr \beta$ and $\alpha\Vdash^{\Sigma_{\At}}_{\mathcal{D}} \bot$ and $\beta \Vdash^{\Omega_{\At}}_{\mathcal{D}} \bot$ (thus also $\alpha\Vdash^{\Sigma_{\At}}_{\mathcal{E}} \bot$ and $\beta \Vdash^{\Omega_{\At}}_{\mathcal{E}} \bot$ by monotonicity), by~\ref{eq:supp-parr}, we obtain $\Vdash^{\Pi_{\At}, \Sigma_{\At}, \Omega_{\At}}_{\mathcal{E}} \bot$. 
Since also $\mathcal{E} \supseteq \mathcal{D}$ such that $\Vdash^{\Pi_{\At}}_{\mathcal{E}} \alpha \parr \beta$ for arbitrary $\Pi_{\At}$, we obtain $\alpha \parr \beta \Vdash^{\Sigma_{\At}, \Omega_{\At}}_{\mathcal{D}}\bot$ by~\ref{eq:supp-inf}. 
Hence, $\Vdash^{\Sigma_{\At}, \Omega_{\At}}_{\mathcal{D}} (\alpha \parr \beta) \multimap \bot$, by Lemma~\ref{lemma:negatingformula}, \ie~$\Vdash^{\Sigma_{\At}, \Omega_{\At}}_{\mathcal{D}} \neg (\alpha \parr \beta)$. 
Now, since $\neg (\alpha \parr \beta) , \Vdash^{\Theta^{1}_{\At},\dots,\Theta^{n}_{\At}}_{\mathcal{C}} \bot$ and $\Vdash^{\Sigma_{\At}, \Omega_{\At}}_{\mathcal{D}} \neg (\alpha \parr \beta)$, we obtain $\Vdash^{\Sigma_{\At}, \Omega_{\At}, \Theta^{1}_{\At},\dots,\Theta^{n}_{\At}}_{\mathcal{D}} \bot$ by~\ref{eq:supp-inf}. 
Since also $\mathcal{D} \supseteq \mathcal{C}$ such that $\alpha\Vdash^{\Sigma_{\At}}_{\mathcal{D}} \bot$ and $\beta \Vdash^{\Omega_{\At}}_{\mathcal{D}} \bot$ for arbitrary $\Sigma_{\At}, \Omega_{\At}$, by~\ref{eq:supp-parr}, we obtain $\Vdash^{\Theta^{1}_{\At},\dots,\Theta^{n}_{\At}}_{\mathcal{C}} \alpha \parr \beta$. 
Finally, since $\mathcal{C} \supseteq \mathcal{B}$ such that $\Vdash^{\Theta^{i}_{\At}}_{\mathcal{C}} \psi^i$ for arbitrary multisets $\Theta^{i}_{\At}$, we obtain $\Gamma \Vdash_{\mathcal{B}} \alpha \parr \beta$ by~\ref{eq:supp-inf}.   

        \item[$\phi = \alpha \with \beta:$] Assume $\Gamma, \neg (\alpha \with \beta) \Vdash_{\mathcal{B}} \bot$. Since $\Vdash^{\Theta^{i}_{\At}}_{\mathcal{C}} \psi^i$, $\forall \psi^i \in \Gamma$, we obtain $\neg (\alpha \with \beta) \Vdash^{\Theta^{1}_{\At},\dots,\Theta^{n}_{\At}}_{\mathcal{C}} \bot$ by Lemma~\ref{lemma:partialinf}. Now assume that, for an arbitrary $\mathcal{D} \supseteq \mathcal{C}$ and arbitrary $\Sigma_{\At}$, $\Vdash^{\Sigma_{\At}}_{\mathcal{D}} \neg \alpha$ (\ie~$\alpha \Vdash^{\Sigma_{\At}}_{\mathcal{D}} \bot $ by Lemma~\ref{lemma:negatingformula}). Further assume that, for an arbitrary $\mathcal{E} \supseteq \mathcal{D}$ and arbitrary $\Pi_{\At}$, $\Vdash^{\Pi_{\At}}_{\mathcal{E}} \alpha \with \beta$. Then $\Vdash^{\Pi_{\At}}_{\mathcal{E}} \alpha$ by~\ref{eq:supp-and}, and hence $\Vdash^{\Sigma_{\At}, \Pi_{\At}}_{\mathcal{E}} \bot$ by~\ref{eq:supp-inf}. Since $\Vdash^{\Sigma_{\At}, \Pi_{\At}}_{\mathcal{E}} \bot$ and $\mathcal{E} \supseteq \mathcal{D}$ such that $\Vdash^{\Pi_{\At}}_{\mathcal{E}} \alpha \with \beta$, we obtain $\alpha \with \beta \Vdash^{\Sigma_{\At}}_{\mathcal{D}} \bot$ by~\ref{eq:supp-inf}. Thus, $\Vdash^{\Sigma_{\At}}_{\mathcal{D}} \alpha \with \beta \multimap \bot$ by Lemma~\ref{lemma:negatingformula}, \ie~$\Vdash^{\Sigma_{\At}}_{\mathcal{D}} \neg (\alpha \with \beta)$. Now, from $\neg (\alpha \with \beta) \Vdash^{\Theta^{1}_{\At},\dots,\Theta^{n}_{\At}}_{\mathcal{C}} \bot$ and $\Vdash^{\Sigma_{\At}}_{\mathcal{D}} \neg (\alpha \with \beta)$ we obtain $\Vdash^{\Sigma_{\At}, \Theta^{1}_{\At},\dots,\Theta^{n}_{\At}}_{\mathcal{D}} \bot$ by~\ref{eq:supp-inf}. Since also $\mathcal{D} \supseteq \mathcal{C}$ such that $\Vdash^{\Sigma_{\At}}_{\mathcal{D}} \neg \alpha $ for arbitrary $\Sigma_{\At}$, by~\ref{eq:supp-inf}, we obtain $\neg \alpha \Vdash^{\Theta^{1}_{\At},\dots,\Theta^{n}_{\At}}_{\mathcal{C}} \bot$. Analogously, for an arbitrary $\mathcal{F} \supseteq \mathcal{C}$ and arbitrary $\Pi_{\At}$ such that $\Vdash^{\Pi_{\At}}_{\mathcal{F}} \neg \beta$, obtain $\neg \beta \Vdash^{\Theta^{1}_{\At},\dots,\Theta^{n}_{\At}}_{\mathcal{C}} \bot$. Since $\mathcal{C} \supseteq \mathcal{B}$ such that $\Vdash^{\Theta^{i}_{\At}}_{\mathcal{C}} \psi^i$, $\forall \psi^i \in \Gamma$, for arbitrary multisets $\Theta^{i}_{\At}$, by ~\ref{eq:supp-inf} we conclude $\Gamma, \neg \alpha \Vdash_{\mathcal{B}} \bot$ and $\Gamma, \neg \beta \Vdash_{\mathcal{B}} \bot$; the induction hypothesis yields $\Gamma \Vdash_{\mathcal{B}} \alpha$ and $\Gamma \Vdash_{\mathcal{B}} \beta$. Now, since $\Vdash^{\Theta^{i}_{\At}}_{\mathcal{C}} \psi^i$, $\forall \psi^i \in \Gamma$, by~\ref{eq:supp-inf} again we obtain $\Vdash^{\Theta^{1}_{\At},\dots,\Theta^{n}_{\At}}_{\mathcal{C}} \alpha$ and $\Vdash^{\Theta^{1}_{\At},\dots,\Theta^{n}_{\At}}_{\mathcal{C}} \beta$, respectively. Then, by~\ref{eq:supp-and} we obtain $\Vdash^{\Theta^{1}_{\At},\dots,\Theta^{n}_{\At}}_{\mathcal{C}} \alpha \with \beta$. Finally, since $\mathcal{C} \supseteq \mathcal{B}$ such that $\Vdash^{\Theta^{i}_{\At}}_{\mathcal{C}} \psi^i$  for arbitrary multisets $\Theta^{i}_{\At}$, we obtain $\Gamma \Vdash_{\mathcal{B}} \alpha \with \beta$ by~\ref{eq:supp-inf}.

        \item[$\phi = 0:$] Assume $\Gamma, \neg 0 \Vdash_{\mathcal{B}} \bot$. Since $\Vdash^{\Theta^{i}_{\At}}_{\mathcal{C}} \psi^i$, $\forall \psi^i \in \Gamma_{\At}$, we obtain $\neg 0 \Vdash^{\Theta^{1}_{\At},\dots,\Theta^{n}_{\At}}_{\mathcal{C}} \bot$ by Lemma~\ref{lemma:partialinf}. Further assume that, for an arbitrary $\mathcal{D} \supseteq \mathcal{C}$ and arbitrary $\Pi_{\At}$, $\Vdash^{\Pi_{\At}}_{\mathcal{D}} 0$. Then, by~\ref{eq:supp-0}, we have that  $\Vdash^{\Pi_{\At}, \Delta_{\At}}_{\mathcal{D}} \bot$ for arbitrary $\Delta_{\At}$. Since also $\mathcal{D} \supseteq \mathcal{C}$ such that $\Vdash^{\Pi_{\At}}_{\mathcal{D}} 0$ for arbitrary $\Pi_{\At}$, we obtain $0 \Vdash^{\Delta_{\At}}_{\mathcal{C}}\bot$ by~\ref{eq:supp-inf}. Hence, $\Vdash^{\Delta_{\At}}_{\mathcal{C}} 0 \multimap \bot$, by Lemma~\ref{lemma:negatingformula}, \ie~$\Vdash^{\Delta_{\At}}_{\mathcal{C}} \neg 0$. Now, since $\neg 0 \Vdash^{\Theta^{1}_{\At},\dots,\Theta^{n}_{\At}}_{\mathcal{C}} \bot$ and $\Vdash^{\Delta_{\At}}_{\mathcal{C}} \neg 0$, we obtain $\Vdash^{\Delta_{\At}, \Theta^{1}_{\At},\dots,\Theta^{n}_{\At}}_{\mathcal{C}} \bot$ by~\ref{eq:supp-inf}. Now, since $\Delta_{\At}$ is arbitrary, we obtain $\Vdash^{\Theta^{1}_{\At},\dots,\Theta^{n}_{\At}}_{\mathcal{C}} 0$ by~\ref{eq:supp-0}. Finally, since $\mathcal{C} \supseteq \mathcal{B}$ such that $\Vdash^{\Theta^{i}_{\At}}_{\mathcal{C}} \psi^i$ for arbitrary multisets $\Theta^{i}_{\At}$, we obtain $\Gamma \Vdash_{\mathcal{B}} 0$ by~\ref{eq:supp-inf}.
    \end{description}
\end{proof}

\begin{theorem}[Soundness] \label{thm:soundness}
    If $\Gamma \vdash_{\MALL} \phi$ then $\Gamma \Vdash \phi$.
\end{theorem}

\begin{proof}
    Given that $\vdash$ is defined inductively, it suffices to prove the following:
\begin{itemize}[itemsep=0.1em, labelindent=1em, labelsep=0.8em, leftmargin=*]
    \item[(Ax)'] $\phi \Vdash \phi$.
    \item[($\otimes$I)'] If $\Gamma \Vdash \phi$ and $\Delta \Vdash \psi$ then $\Gamma, \Delta \Vdash \phi \otimes \psi$.
    \item[($\otimes$E)'] If $\Gamma \Vdash \phi \otimes \psi$ and $\Delta, \phi, \psi \Vdash \chi$ then $\Gamma, \Delta \Vdash \chi$.
    \item[($\multimap$I)'] If $\Gamma, \phi \Vdash \psi$ then $\Gamma \Vdash \phi \multimap \psi$.
    \item[($\multimap$E)'] If $\Gamma \Vdash \phi \multimap \psi$ and $\Delta \Vdash \phi$ then $\Gamma, \Delta \Vdash \psi$.
    \item[($1$I)'] $\Vdash 1$.
    \item[($1$E)'] If $\Gamma \Vdash \phi$ and $\Delta \Vdash 1$ then $\Gamma, \Delta \Vdash \phi$.
    \item[($\parr$I)'] If $\Gamma, \neg \phi, \neg \psi \Vdash \bot$ then $\Gamma \Vdash \phi \parr \psi$
    \item[($\parr$E)'] If $\Gamma \Vdash \phi \parr \psi$ and $\Delta, \phi \Vdash \bot$ and $\Theta, \psi \Vdash \bot$ then $\Gamma, \Delta, \Theta \Vdash \bot$.
    \item[(Raa)'] If $\Gamma, \neg \phi \Vdash \bot$ then $\Gamma \Vdash \phi$.
    \item[($\with$I)'] If $\Gamma \Vdash \phi$ and $\Gamma \Vdash \psi$ then $\Gamma\Vdash \phi \with \psi$.
    \item[($\with$E)'] If $\Gamma \Vdash \phi \with \psi$ then $\Gamma \Vdash \phi$ and $\Gamma \Vdash \psi$.
    \item[($\oplus$I)'] If $\Gamma \Vdash \phi$ or $\Gamma \Vdash \psi$ then $\Gamma\Vdash \phi \oplus \psi$.
    \item[($\oplus$E)'] If $\Gamma \Vdash \phi \oplus \psi$ and $\Delta, \phi \Vdash \chi$ and $\Delta, \psi \Vdash \chi$ then $\Gamma, \Delta \Vdash \chi$.
    \item[($\top$I)'] $\Gamma \Vdash \top$.
    \item[($0$E)'] If $\Gamma \Vdash 0$ then $\Gamma, \Delta \Vdash \phi$.
\end{itemize}

In this proof, we set $\Gamma = \{\alpha^1,\dots,\alpha^n\}$ and $\Delta = \{\beta^1,\dots,\beta^m\}$. We will make use of Lemma~\ref{lemma:validity}, which states that $\Gamma \Vdash \phi$ is equivalent to $\Gamma \Vdash_{\mathcal{S}} \phi$ for any $\Gamma$ and  $\phi$.

\begin{description}[itemsep=0.8em, font=\normalfont]
\item[(Ax)'.] By~\ref{eq:supp-inf}, it suffices to show that, for arbitrary $\mathcal{B}$ and arbitary $\Delta_{\At}$, $\Vdash_{\mathcal{B}}^{\Delta_{\At}} \phi$ implies $\Vdash_{\mathcal{B}}^{\Delta_{\At}} \phi$, which trivially holds.

\item[($\otimes$I)'.] Assume $\Gamma \Vdash \phi$ and $\Delta \Vdash \psi$.  
Now assume that, for an arbitrary $\mathcal{B}$, $\forall \alpha^i \in \Gamma$ $(1 \leq i \leq n)$, $\forall \beta^j \in \Delta$ $(1 \leq j \leq m)$ and arbitrary multisets $\Theta^{i}_{\At}, \Sigma^{j}_{\At}$, $\Vdash^{\Theta^{i}_{\At}}_{\mathcal{B}} \alpha^i$ and $\Vdash^{\Sigma^{j}_{\At}}_{\mathcal{B}} \beta^j$. By~\ref{eq:supp-inf}, we obtain $\Vdash^{\Theta^{1}_{\At},\dots,\Theta^{n}_{\At}}_{\mathcal{B}} \phi$ and $\Vdash^{\Sigma^{1}_{\At},\dots,\Sigma^{m}_{\At}}_{\mathcal{B}} \psi$, respectively. Further assume, for an arbitrary $\mathcal{C} \supseteq \mathcal{B}$ and arbitrary $\Pi_{\At}$, that $\phi, \psi \Vdash^{\Pi_{\At}}_{\mathcal{C}} \bot$. Since moreover $\Vdash^{\Theta^{1}_{\At},\dots,\Theta^{n}_{\At}}_{\mathcal{B}} \phi$ and $\Vdash^{\Sigma^{1}_{\At},\dots,\Sigma^{m}_{\At}}_{\mathcal{B}} \psi$ (thus also $\Vdash^{\Theta^{1}_{\At},\dots,\Theta^{n}_{\At}}_{\mathcal{C}} \phi$ and $\Vdash^{\Sigma^{1}_{\At},\dots,\Sigma^{m}_{\At}}_{\mathcal{C}} \psi$ by monotonicity), we obtain $\Vdash^{\Pi_{\At}, \Theta^{1}_{\At},\dots,\Theta^{n}_{\At}, \Sigma^{1}_{\At},\dots,\Sigma^{m}_{\At}}_{\mathcal{C}} \bot$ by~\ref{eq:supp-inf}. Since also $\mathcal{C} \supseteq \mathcal{B}$ such that $\phi, \psi \Vdash^{\Pi_{\At}}_{\mathcal{C}} \bot$ for arbitrary $\Pi_{\At}$, by~\ref{eq:supp-tensor}, $\Vdash^{\Theta^{1}_{\At},\dots,\Theta^{n}_{\At}, \Sigma^{1}_{\At},\dots,\Sigma^{m}_{\At}}_{\mathcal{B}} \phi \otimes \psi$. Finally, since $\mathcal{B}$ was chosen such that $\Vdash^{\Theta^{i}_{\At}}_{\mathcal{B}} \alpha^i$ and $\Vdash^{\Sigma^{j}_{\At}}_{\mathcal{B}} \beta^j$ for arbitrary multisets $\Theta^{i}_{\At}, \Sigma^{j}_{\At}$, by~\ref{eq:supp-inf}, $\Gamma, \Delta \Vdash \phi \otimes \psi$. 

\item[($\otimes$E)'.] Assume $\Gamma \Vdash \phi \otimes \psi$ and $\Delta, \phi, \psi \Vdash \chi$. 
Further assume that, for an arbitrary $\mathcal{B}$, $\forall \alpha^i \in \Gamma$ $(1 \leq i \leq n)$, $\forall \beta^j \in \Delta$ $(1 \leq j \leq m)$ and arbitrary multisets $\Theta^{i}_{\At}, \Sigma^{j}_{\At}$, $\Vdash^{\Theta^{i}_{\At}}_{\mathcal{B}} \alpha^i$ and $\Vdash^{\Sigma^{j}_{\At}}_{\mathcal{B}} \beta^j$. Then, by~\ref{eq:supp-inf}, we obtain $\Vdash^{\Theta^{1}_{\At},\dots,\Theta^{n}_{\At}}_{\mathcal{B}} \phi \otimes \psi$ and, by Lemma~\ref{lemma:partialinf}, $\phi, \psi \Vdash^{\Sigma^{1}_{\At},\dots,\Sigma^{m}_{\At}}_{\mathcal{B}} \chi$. By Lemma~\ref{lemma:generictensor}, we thus obtain $\Vdash^{\Theta^{1}_{\At},\dots,\Theta^{n}_{\At}, \Sigma^{1}_{\At},\dots,\Sigma^{m}_{\At}}_{\mathcal{B}} \chi$. Finally, since $\mathcal{B}$ was chosen such that $\Vdash^{\Theta^{i}_{\At}}_{\mathcal{B}} \alpha^i$ and $\Vdash^{\Sigma^{j}_{\At}}_{\mathcal{B}} \beta^j$ for arbitrary multisets $\Theta^{i}_{\At}, \Sigma^{j}_{\At}$, we obtain $\Gamma, \Delta \Vdash \chi$ by~\ref{eq:supp-inf}. 

\item[($\multimap$I)'.] Assume $\Gamma, \phi \Vdash \psi$. 
Now assume that, for an arbitrary $\mathcal{B}$, $\forall \alpha^i \in \Gamma$ $(1 \leq i \leq n)$ and arbitrary multisets $\Theta^{i}_{\At}$, $\Vdash^{\Theta^{i}_{\At}}_{\mathcal{B}} \alpha^i$. Further assume that, for an arbitrary $\mathcal{C} \supseteq \mathcal{B}$ and arbitrary $\Sigma_{\At}$, $\Pi_{\At}$, $\Vdash^{\Sigma_{\At}}_{\mathcal{C}} \phi$ and $\psi \Vdash^{\Pi_{\At}}_{\mathcal{C}} \bot$. Now, from $\Gamma, \phi \Vdash \psi$ and $\Vdash^{\Theta^{i}_{\At}}_{\mathcal{B}} \alpha^i$ (thus also $\Vdash^{\Theta^{i}_{\At}}_{\mathcal{C}} \alpha^i$ by monotonicity) and $\Vdash^{\Sigma_{\At}}_{\mathcal{C}} \phi$, by~\ref{eq:supp-inf}, $\Vdash^{\Theta^{1}_{\At},\dots,\Theta^{n}_{\At}, \Sigma_{\At}}_{\mathcal{C}} \psi$. Since moreover $\psi \Vdash^{\Pi_{\At}}_{\mathcal{C}} \bot$, we obtain $\Vdash^{\Pi_{\At}, \Theta^{1}_{\At},\dots,\Theta^{n}_{\At}, \Sigma_{\At}}_{\mathcal{C}} \bot$ by~\ref{eq:supp-inf}. Since also $\mathcal{C} \supseteq \mathcal{B}$ such that $\Vdash^{\Sigma_{\At}}_{\mathcal{C}} \phi$ and $\psi \Vdash^{\Pi_{\At}}_{\mathcal{C}} \bot$ for arbitrary $\Sigma_{\At}, \Pi_{\At}$, by~\ref{eq:supp-imply}, $\Vdash^{\Theta^{1}_{\At},\dots,\Theta^{n}_{\At}}_{\mathcal{B}} \phi \multimap \psi$. Finally, since $\mathcal{B}$ was chosen such that $\Vdash^{\Theta^{i}_{\At}}_{\mathcal{B}} \alpha^i$ for arbitrary multisets $\Theta^{i}_{\At}$, by~\ref{eq:supp-inf}, $\Gamma \Vdash \phi \multimap \psi$.

\item[($\multimap$E)'.] Assume $\Gamma \Vdash \phi \multimap \psi$ and $\Delta \Vdash \phi$. 
Further assume that, for an arbitrary $\mathcal{B}$, $\forall \alpha^i \in \Gamma$ $(1 \leq i \leq n)$, $\forall \beta^j \in \Delta$ $(1 \leq j \leq m)$ and arbitrary multisets $\Theta^{i}_{\At}, \Sigma^{j}_{\At}$, $\Vdash^{\Theta^{i}_{\At}}_{\mathcal{B}} \alpha^i$ and $\Vdash^{\Sigma^{j}_{\At}}_{\mathcal{B}} \beta^j$. By~\ref{eq:supp-inf}, we obtain $\Vdash^{\Theta^{1}_{\At},\dots,\Theta^{n}_{\At}}_{\mathcal{B}} \phi \multimap \psi$ and $\Vdash^{\Sigma^{1}_{\At},\dots,\Sigma^{m}_{\At}}_{\mathcal{B}} \phi$, respectively. Notice that $\psi \Vdash \psi$ by (Ax)', thus $\psi \Vdash_{\mathcal{B}} \psi$ by monotonicity. We then obtain $\Vdash^{\Theta^{1}_{\At},\dots,\Theta^{n}_{\At}, \Sigma^{1}_{\At},\dots,\Sigma^{m}_{\At}}_{\mathcal{B}} \psi$ by Lemma~\ref{lemma:genericimplication}. Finally, since $\mathcal{B}$ was chosen such that $\Vdash^{\Theta^{i}_{\At}}_{\mathcal{B}} \alpha^i$ and $\Vdash^{\Sigma^{j}_{\At}}_{\mathcal{B}} \beta^j$ for arbitrary multisets $\Theta^{i}_{\At}, \Sigma^{j}_{\At}$, by~\ref{eq:supp-inf}, $\Gamma, \Delta \Vdash \psi$.

\item[($1$I)'.] Assume that, for an arbitrary $\mathcal{B}$ and arbitrary $\Theta_{\At}$, $\Vdash^{\Theta_{\At}}_{\mathcal{B}} \bot$. Now, $\Vdash^{\Theta_{\At}}_{\mathcal{B}} \bot$ trivially implies $\Vdash^{\Theta_{\At}}_{\mathcal{B}} \bot$. Thus, since $\mathcal{B}$ and $\Theta_{\At}$ are arbitrary, by~\ref{eq:supp-1} we obtain $\Vdash 1$.

\item[($1$E)'.] Assume $\Gamma \Vdash \phi$ and $\Delta \Vdash 1$. 
Further assume that, for an arbitrary $\mathcal{B}$, $\forall \alpha^i \in \Gamma$ $(1 \leq i \leq n)$, $\forall \beta^j \in \Delta$ $(1 \leq j \leq m)$ and arbitrary multisets $\Theta^{i}_{\At}, \Sigma^{j}_{\At}$, $\Vdash^{\Theta^{i}_{\At}}_{\mathcal{B}} \alpha^i$ and $\Vdash^{\Sigma^{j}_{\At}}_{\mathcal{B}} \beta^j$. Then, by~\ref{eq:supp-inf} we obtain $\Vdash^{\Theta^{1}_{\At},\dots,\Theta^{n}_{\At}}_{\mathcal{B}} \phi$ and $\Vdash^{\Sigma^{1}_{\At},\dots,\Sigma^{m}_{\At}}_{\mathcal{B}} 1$, respectively. By Lemma~\ref{lemma:genericone}, we thus obtain $\Vdash^{\Theta^{1}_{\At},\dots,\Theta^{n}_{\At}, \Sigma^{1}_{\At},\dots,\Sigma^{m}_{\At}}_{\mathcal{B}} \phi$. Finally, since $\mathcal{B}$ was chosen such that $\Vdash^{\Theta^{i}_{\At}}_{\mathcal{B}} \alpha^i$ and $\Vdash^{\Sigma^{j}_{\At}}_{\mathcal{B}} \beta^j$ for arbitrary multisets $\Theta^{i}_{\At}, \Sigma^{j}_{\At}$, we obtain $\Gamma, \Delta \Vdash \phi$ by~\ref{eq:supp-inf}.

\item[($\parr$I)'.] Assume $\Gamma, \neg\phi,\neg\psi \Vdash \bot$ and that, for an arbitrary $\mathcal{B}$, $\forall \alpha^i \in \Gamma$ $(1 \leq i \leq n)$ and arbitrary multisets $\Theta^{i}_{\At}$, $\Vdash^{\Theta^{i}_{\At}}_{\mathcal{B}} \alpha^i$. 
Further assume that, for an arbitrary $\mathcal{C} \supseteq \mathcal{B}$ and arbitrary $\Sigma_{\At}$, $\Pi_{\At}$, $\phi\Vdash^{\Sigma_{\At}}_{\mathcal{C}} \bot$ and $\psi \Vdash^{\Pi_{\At}}_{\mathcal{C}} \bot$. 
From Lemma~\ref{lemma:negatingformula} we have that $\Vdash^{\Sigma_{\At}}_{\mathcal{C}} \neg\phi$ and $ \Vdash^{\Pi_{\At}}_{\mathcal{C}} \neg\psi$ and, by monotonicity, $\Vdash^{\Theta^{i}_{\At}}_{\mathcal{C}} \alpha^i$ . 
Hence by applying~\ref{eq:supp-inf} twice, $\Vdash^{\Pi_{\At}, \Theta^{1}_{\At},\dots,\Theta^{n}_{\At}, \Sigma_{\At}}_{\mathcal{C}} \bot$. 
Thus by~\ref{eq:supp-parr}, $\Vdash^{\Theta^{1}_{\At},\dots,\Theta^{n}_{\At}}_{\mathcal{B}} \phi \parr \psi$. Finally, since $\mathcal{B}$ was chosen such that $\Vdash^{\Theta^{i}_{\At}}_{\mathcal{B}} \alpha^i$ for arbitrary multisets $\Theta^{i}_{\At}$, by~\ref{eq:supp-inf}, $\Gamma \Vdash \phi \parr \psi$.

\item[($\parr$E)'.] Assume $\Gamma \Vdash \phi \parr \psi$ and  $\Delta, \phi\Vdash \bot$ and $\Omega, \psi\Vdash \bot$. 
Further assume that, for an arbitrary $\mathcal{B}$, $\forall \alpha^i \in \Gamma$ $(1 \leq i \leq n)$, $\forall \beta^j \in \Delta$ $(1 \leq j \leq m)$, $\forall \gamma^k \in \Omega$ $(1 \leq k \leq s)$ and arbitrary multisets $\Theta^{i}_{\At}, \Sigma^{j}_{\At},  \Pi^{k}_{\At}$, $\Vdash^{\Theta^{i}_{\At}}_{\mathcal{B}} \alpha^i$, $\Vdash^{\Sigma^{j}_{\At}}_{\mathcal{B}} \beta^j$ and $\Vdash^{\Pi^{k}_{\At}}_{\mathcal{B}} \gamma^k$. Then, by~\ref{eq:supp-inf}, we obtain $\Vdash^{\Theta^{1}_{\At},\dots,\Theta^{n}_{\At}}_{\mathcal{B}} \phi \parr \psi$ and, by Lemma~\ref{lemma:partialinf}, $\phi\Vdash^{\Sigma^{j}_{\At},\dots,\Sigma^{m}_{\At}}_{\mathcal{B}} \bot$ and $\psi\Vdash^{\Pi^{1}_{\At},\dots,\Pi^{s}_{\At}}_{\mathcal{B}} \bot$. By Lemma~\ref{lemma:generictensor}, we thus obtain $\Vdash^{\Theta^{1}_{\At}\ldots\Theta^{n}_{\At}, \Sigma^{1}_{\At}\ldots,\Sigma^{m}_{\At},\Pi^1_{At}\ldots \Pi^s_{At}}_{\mathcal{B}} \bot$. Hence $\Gamma, \Delta,\Omega \Vdash \bot$ by~\ref{eq:supp-inf}. 

\item[(Raa)'.] Immediate by Lemma~\ref{lemma:genericraa} (set $\mathcal{B} = \mathcal{S}$).

\item[($\with$I)'.] Assume $\Gamma \Vdash \phi$ and $\Gamma \Vdash \psi$. 
Further assume that, for an arbitrary $\mathcal{B}$, $\forall \alpha^i \in \Gamma$ $(1 \leq i \leq n)$ and arbitrary multisets $\Theta^{i}_{\At}$, $\Vdash^{\Theta^{i}_{\At}}_{\mathcal{B}} \alpha^i$. Then, by~\ref{eq:supp-inf}, we obtain $\Vdash^{\Theta^{1}_{\At},\dots,\Theta^{n}_{\At}}_{\mathcal{B}} \phi$ and $\Vdash^{\Theta^{1}_{\At},\dots,\Theta^{n}_{\At}}_{\mathcal{B}} \psi$. By~\ref{eq:supp-and}, we thus obtain $\Vdash^{\Theta^{1}_{\At},\dots,\Theta^{n}_{\At}}_{\mathcal{B}} \phi \with \psi$. Finally, since $\mathcal{B}$ was chosen such that $\Vdash^{\Theta^{i}_{\At}}_{\mathcal{B}} \alpha^i$ for arbitrary multisets $\Theta^{i}_{\At}$, we obtain $\Gamma \Vdash \phi \with \psi$ by~\ref{eq:supp-inf}.

\item[($\with$E)'.] Assume $\Gamma \Vdash \phi \with \psi$. 
Further assume that, for an arbitrary $\mathcal{B}$, $\forall \alpha^i \in \Gamma$ $(1 \leq i \leq n)$ and arbitrary multisets $\Theta^{i}_{\At}$, $\Vdash^{\Theta^{i}_{\At}}_{\mathcal{B}} \alpha^i$. Then, by~\ref{eq:supp-inf}, we obtain $\Vdash^{\Theta^{1}_{\At},\dots,\Theta^{n}_{\At}}_{\mathcal{B}} \phi \with \psi$. By~\ref{eq:supp-and}, we thus obtain $\Vdash^{\Theta^{1}_{\At},\dots,\Theta^{n}_{\At}}_{\mathcal{B}} \phi$ and $\Vdash^{\Theta^{1}_{\At},\dots,\Theta^{n}_{\At}}_{\mathcal{B}} \psi$. Finally, since $\mathcal{B}$ was chosen such that $\Vdash^{\Theta^{i}_{\At}}_{\mathcal{B}} \alpha^i$ for arbitrary multisets $\Theta^{i}_{\At}$, we obtain $\Gamma \Vdash \phi$ and $\Gamma \Vdash \psi$ by~\ref{eq:supp-inf}.

\item[($\oplus$I)'.] Assume that $\Gamma \Vdash \phi$ or $\Gamma \Vdash \psi$. 
Now assume that, for an arbitrary $\mathcal{B}$, $\forall \alpha^i \in \Gamma$ $(1 \leq i \leq n)$ and arbitrary multisets $\Theta^{i}_{\At}$, $\Vdash^{\Theta^{i}_{\At}}_{\mathcal{B}} \alpha^i$. Then, by~\ref{eq:supp-inf}, we obtain $\Vdash^{\Theta^{1}_{\At},\dots,\Theta^{n}_{\At}}_{\mathcal{B}} \phi$ or $\Vdash^{\Theta^{1}_{\At},\dots,\Theta^{n}_{\At}}_{\mathcal{B}} \psi$. Further assume that, for an arbitrary $\mathcal{C} \supseteq \mathcal{B}$ and arbitrary $\Delta_{\At}$, $\phi \Vdash^{\Delta_{\At}}_{\mathcal{C}} \bot$ and $\psi \Vdash^{\Delta_{\At}}_{\mathcal{C}} \bot$. Now, from either $\phi \Vdash^{\Delta_{\At}}_{\mathcal{C}} \bot$ and $\Vdash^{\Theta^{1}_{\At},\dots,\Theta^{n}_{\At}}_{\mathcal{B}} \phi$ (thus also $\Vdash^{\Theta^{1}_{\At},\dots,\Theta^{n}_{\At}}_{\mathcal{C}} \phi$) or from $\psi \Vdash^{\Delta_{\At}}_{\mathcal{C}} \bot$ and $\Vdash^{\Theta^{1}_{\At},\dots,\Theta^{n}_{\At}}_{\mathcal{B}} \psi$ (thus also $\Vdash^{\Theta^{1}_{\At},\dots,\Theta^{n}_{\At}}_{\mathcal{C}} \psi$), we obtain $\Vdash^{\Theta^{1}_{\At},\dots,\Theta^{n}_{\At}, \Delta_{\At}}_{\mathcal{C}} \bot$ by~\ref{eq:supp-inf}. By~\ref{eq:supp-plus}, we thus obtain $\Vdash^{\Theta^{1}_{\At},\dots,\Theta^{n}_{\At}}_{\mathcal{B}} \phi \oplus \psi$. Finally, since $\mathcal{B}$ was chosen such that $\Vdash^{\Theta^{i}_{\At}}_{\mathcal{B}} \alpha^i$ for arbitrary multisets $\Theta^{i}_{\At}$, we obtain $\Gamma \Vdash \phi \oplus \psi$ by~\ref{eq:supp-inf}.

\item[($\oplus$E)'.] Assume $\Gamma \Vdash \phi \oplus \psi$ and $\Delta, \phi \Vdash \chi$ and $\Delta, \psi \Vdash \chi$. 
Further assume that, for an arbitrary $\mathcal{B}$, $\forall \alpha^i \in \Gamma$ $(1 \leq i \leq n)$, $\forall \beta^j \in \Delta$ $(1 \leq j \leq m)$ and arbitrary multisets $\Theta^{i}_{\At}, \Sigma^{j}_{\At}$, $\Vdash^{\Theta^{i}_{\At}}_{\mathcal{B}} \alpha^i$ and $\Vdash^{\Sigma^{j}_{\At}}_{\mathcal{B}} \beta^j$. Then, by~\ref{eq:supp-inf} and Lemma~\ref{lemma:partialinf}, we obtain $\Vdash^{\Theta^{1}_{\At},\dots,\Theta^{n}_{\At}}_{\mathcal{B}} \phi \oplus \psi$ and $\phi \Vdash^{\Sigma^{1}_{\At},\dots,\Sigma^{m}_{\At}}_{\mathcal{B}} \chi$ and $\psi \Vdash^{\Sigma^{1}_{\At},\dots,\Sigma^{m}_{\At}}_{\mathcal{B}} \chi$. By Lemma~\ref{lemma:genericplus}, we thus obtain $\Vdash^{\Theta^{1}_{\At},\dots,\Theta^{n}_{\At}, \Sigma^{1}_{\At},\dots,\Sigma^{m}_{\At}}_{\mathcal{B}} \chi$. Finally, since $\mathcal{B}$ was chosen such that $\Vdash^{\Theta^{i}_{\At}}_{\mathcal{B}} \alpha^i$ and $\Vdash^{\Sigma^{j}_{\At}}_{\mathcal{B}} \beta^j$ for arbitrary multisets $\Theta^{i}_{\At}, \Sigma^{j}_{\At}$, by~\ref{eq:supp-inf}, $\Gamma, \Delta \Vdash \chi$.

\item[($\top$I)'.]
Further assume that, for an arbitrary $\mathcal{B}$, $\forall \alpha^i \in \Gamma$ $(1 \leq i \leq n)$ and arbitrary multisets $\Theta^{i}_{\At}$, $\Vdash^{\Theta^{i}_{\At}}_{\mathcal{B}} \alpha^i$. By~\ref{eq:supp-top}, $\Vdash^{\Theta^{1}_{\At},\dots,\Theta^{n}_{\At}}_{\mathcal{B}} \top$, hence, by~\ref{eq:supp-inf}, $\Gamma \Vdash \top$ as desired.

\item[($0$E)'.] Assume $\Gamma \Vdash 0$. 
Further assume that, for an arbitrary $\mathcal{B}$, $\forall \alpha^i \in \Gamma$ $(1 \leq i \leq n)$, $\forall \beta^j \in \Delta$ $(1 \leq j \leq m)$ and arbitrary multisets $\Theta^{i}_{\At}, \Sigma^{j}_{\At}$, $\Vdash^{\Theta^{i}_{\At}}_{\mathcal{B}} \alpha^i$ and $\Vdash^{\Sigma^{j}_{\At}}_{\mathcal{B}} \beta^j$. Then, by~\ref{eq:supp-inf}, we obtain $\Vdash^{\Theta^{1}_{\At},\dots,\Theta^{n}_{\At}}_{\mathcal{B}} 0$. By Lemma~\ref{lemma:genericzero}, we thus obtain $\Vdash^{\Theta^{1}_{\At},\dots,\Theta^{n}_{\At}, \Delta_{\At}}_{\mathcal{B}} \phi$ for all $\Delta_{\At}$ and any $\phi$. In particular, let $\Delta_{\At} = \Sigma^{1}_{\At} \cup \dots \cup \Sigma^{m}_{\At}$, hence $\Vdash^{\Theta^{1}_{\At},\dots,\Theta^{n}_{\At}, \Sigma^{1}_{\At},\dots,\Sigma^{m}_{\At}}_{\mathcal{B}} \phi$. By~\ref{eq:supp-inf}, then, $\Gamma, \Delta \Vdash \phi$.
\end{description}
\end{proof} 

\section{Completeness}\label{sec:comp}
In this section, we prove that {\MALL} is complete with respect to the proposed semantics; that is, if $\Gamma \Vdash \phi$, then there exists a {\MALL}-proof of $\Gamma \vdash \phi$. To establish this, 
we associate to each subformula $\psi$ of $\Gamma \cup \{\phi\}$ a unique atom $p^\psi$, and then, exploiting the fact that $\Gamma \Vdash \phi$ is valid with respect to every base, we construct a \emph{simulation base} $\mathcal{U}$ for $\Gamma \cup \{\phi\}$ such that $p^\psi$ behaves in $\mathcal{U}$ as $\psi$ behaves in {\MALL}.

\begin{definition}[Atomic mapping] Let $\Gamma$ be a set of formulas. Let $\Gamma_{S}$ be the set of all subformulas and negations thereof of formulas in $\Gamma$. We say that a function $\sigma: \Gamma_{S} \cup \{\bot\} \rightarrow \At$ is an \emph{atomic mapping} for $\Gamma$ if (1) $\sigma$ is injective, (2) $\sigma(\phi) = \phi$ if $\phi \in \At$. For convenience, we denote $\sigma(\phi) \eqcolon p^\phi$.
\end{definition}

We note that such functions do exist as $\At$ is countably infinite. 

\begin{definition}[Simulation base]
Let $\Gamma$ be a set of formulas and $\sigma$ an atomic mapping for $\Gamma$. Then a \textit{simulation base} $\mathcal{U}$ for $\Gamma$ and $\sigma$ is the base containing exactly the following rules for all $\phi, \psi \in \Gamma$, all multisets $\Gamma_{\At}, \Delta_{\At}, \Theta_{\At}$:\footnote{Remember that bases are closed under Ax and Subs (see Definition~\ref{def:support}).}

\end{definition}

   \begin{center}
        \noindent\begin{minipage}{0.21\textwidth}\scriptsize
        \begin{prooftree}
        \def\ScoreOverhang{0.5pt}
            \AxiomC{$\Gamma_{\At}, p^{\neg \phi} \vdash \bot$}
            \RightLabel{Raa}
            \UnaryInfC{$\Gamma_{\At} \vdash p^{\phi}$}
        \end{prooftree}
        \end{minipage}
        \vspace{2pt}
        \noindent\begin{minipage}{0.21\textwidth}\scriptsize
        \begin{prooftree}
        \def\ScoreOverhang{0.5pt}
            \AxiomC{}
            \RightLabel{$p^{\top}$I}
            \UnaryInfC{$\Gamma_{\At} \vdash p^{\top}$}
        \end{prooftree}
        \end{minipage}
        \vspace{2pt}
        \noindent\begin{minipage}{0.21\textwidth}\scriptsize
        \begin{prooftree}
        \def\ScoreOverhang{0.5pt}
            \AxiomC{$\Gamma_{\At} \vdash p^{0}$}
            \RightLabel{$0$E}
            \UnaryInfC{$\Gamma_{\At},\Delta_{\At} \vdash \bot$}
        \end{prooftree}
        \end{minipage}
        \vspace{2pt}
        \hspace{-1.5em}
        \noindent\begin{minipage}{0.23\textwidth}\scriptsize
        \begin{prooftree}
        \def\ScoreOverhang{0.5pt}
            \AxiomC{$\Gamma_{\At} \vdash p^{\phi}$}
            \AxiomC{$\Delta_{\At} \vdash p^{\psi}$}
            \RightLabel{$\otimes$I}
            \BinaryInfC{$\Gamma_{\At},\Delta_{\At} \vdash p^{\phi\otimes\psi}$}
        \end{prooftree}
        \begin{prooftree}
        \def\ScoreOverhang{0.5pt}
            \AxiomC{}
            \RightLabel{$1$I}
            \UnaryInfC{$\vdash p^{1}$}
        \end{prooftree}
        \begin{prooftree}
        \def\ScoreOverhang{0.5pt}
            \AxiomC{$\Gamma_{\At} \vdash p^{\phi}$}
            \AxiomC{$\Gamma_{\At} \vdash p^{\psi}$}
            \RightLabel{$\with$I}
            \BinaryInfC{$\Gamma_{\At} \vdash p^{\phi \with \psi}$}
        \end{prooftree}
        \end{minipage}%
        \noindent\begin{minipage}{0.23\textwidth}\scriptsize
        \begin{prooftree}
        \def\ScoreOverhang{0.5pt}
            \AxiomC{$\Gamma_{\At} \vdash p^{\phi\otimes\psi}$}
            \AxiomC{$\Delta_{\At}, p^{\phi}, p^{\psi} \vdash \bot$}
            \RightLabel{$\otimes$E}
            \BinaryInfC{$\Gamma_{\At}, \Delta_{\At} \vdash \bot$}
        \end{prooftree}
        \begin{prooftree}
        \def\ScoreOverhang{0.5pt}
            \AxiomC{$\Gamma_{\At} \vdash p^{\phi}$}
            \AxiomC{$\Delta_{\At} \vdash p^{1}$}
            \RightLabel{$1$E}
            \BinaryInfC{$\Gamma_{\At},\Delta_{\At} \vdash p^{\phi}$}
        \end{prooftree}
        \begin{prooftree}
        \def\ScoreOverhang{0.5pt}
            \AxiomC{$\Gamma_{\At} \vdash p^{\phi_1 \with \phi_2}$}
            \RightLabel{$\with\text{E}_i$}
            \UnaryInfC{$\Gamma_{\At} \vdash p^{\phi_i}$}
        \end{prooftree}
        \end{minipage}%
        \noindent\begin{minipage}{0.2\textwidth}\scriptsize
        \begin{prooftree}
          \def\ScoreOverhang{0.5pt}
              \AxiomC{$\Gamma_{\At},p^{\phi} \vdash p^{\psi}$}
              \RightLabel{$\multimap$I}
              \UnaryInfC{$\Gamma_{\At} \vdash p^{\phi \multimap \psi}$}
        \end{prooftree}
        \begin{prooftree}
        \def\ScoreOverhang{0.5pt}
            \AxiomC{$\Gamma_{\At}, p^{\neg \phi}, \neg p^{\psi} \vdash \bot$}
            \RightLabel{$\parr\text{I}$}
            \UnaryInfC{$\Gamma_{\At} \vdash p^{\phi \parr \psi}$}
        \end{prooftree}
        \begin{prooftree}
        \def\ScoreOverhang{0.5pt}
            \AxiomC{$\Gamma_{\At} \vdash p^{\phi_i}$}
            \RightLabel{$\oplus\text{I}_i$}
            \UnaryInfC{$\Gamma_{\At} \vdash p^{\phi_1 \oplus \phi_2}$}
        \end{prooftree}
        \end{minipage}%
        \noindent\begin{minipage}{0.35\textwidth}\scriptsize
        \begin{prooftree}
          \def\ScoreOverhang{0.5pt}
              \AxiomC{$\Gamma_{\At} \vdash p^{\phi \multimap \psi}$}
              \AxiomC{$\Delta_{\At} \vdash p^{\phi}$}
              \RightLabel{$\multimap$E}
              \BinaryInfC{$\Gamma_{\At},\Delta_{\At} \vdash p^{\psi}$}
        \end{prooftree}
        \begin{prooftree}
        \def\ScoreOverhang{0.2pt}
            \AxiomC{$\Gamma_{\At} \vdash p^{\phi \parr \psi}$}
            \AxiomC{$\Delta_{\At} , p^{\phi} \vdash \bot$}
            \AxiomC{$\Theta_{\At} ,p^{\psi}\vdash \bot$} 
            \RightLabel{$\parr$E}
            \TrinaryInfC{$\Gamma_{\At}, \Delta_{\At}, \Theta_{\At} \vdash \bot$}
        \end{prooftree}
        \begin{prooftree}
        \def\ScoreOverhang{0.5pt}
            \AxiomC{$\Gamma_{\At} \vdash p^{\phi \oplus \psi}$}
            \AxiomC{$\Delta_{\At},p^{\phi} \vdash \bot$}
            \AxiomC{$\Delta_{\At},p^{\psi} \vdash \bot$}
            \RightLabel{$\oplus$E}
            \TrinaryInfC{$\Gamma_{\At}, \Delta_{\At} \vdash \bot$}
        \end{prooftree}
        \end{minipage}%
    \end{center}

Notice that, unlike usual proofs via simulation bases \cite{DBLP:journals/corr/abs-2402-01982,DBLP:conf/tableaux/GheorghiuGP23,Sandqvist2015IL}, ours does not require inclusion of all atomic instances of $\otimes E$, $\oplus E$, $0 E$, and $\parr E$; we only require instances with minor premises of shape $\bot$.

\begin{lemma} \label{lemma:completeness}
    Let $\Pi$ be a set of formulas, $\sigma$ an atomic mapping and $\mathcal{U}$ a simulation base for $\Pi$ and $\sigma$. Then, for all $\phi \in \Pi$, all $\mathcal{B} \supseteq \mathcal{U}$ and all $\Gamma_{\At}$, $\Vdash^{\Gamma_{\At}}_{\mathcal{B}} \phi$ if and only if  $\Gamma_{\At} \vdash_{\mathcal{B}} p^{\phi}$.
\end{lemma}

\begin{proof}
%
The proof is by induction on the complexity of $\phi$. The induction hypothesis is such that $$\Vdash^{\Delta_{\At}}_{\mathcal{B}} \chi \text{ if and only if } \Delta_{\At} \vdash_{\mathcal{B}} p^{\chi}$$ holds true for any $\Delta_{\At}$, any subformula $\chi$ of $\phi$ and any base in place of $\mathcal{B}$. We illustrate the base case and the case for $\multimap$, the other cases are similar and simpler.

\begin{description}[itemsep=0.5em, font=\normalfont]
\item[(Base case).] $\phi = p$, hence $p^{p} = p$; denote $p^{\neg p}$ as $\neg p$.

\noindent($\Rightarrow$): Assume $\Vdash^{\Gamma_{\At}}_{\mathcal{B}} p$. Notice that $p \vdash_{\mathcal{B}} p$ and $\neg p \vdash_{\mathcal{B}} \neg p$. Then the following is a deduction in $\mathcal{B}$:

\begin{prooftree}
\def\ScoreOverhang{0.5pt}
    \AxiomC{}
     \RightLabel{\small Ax}
    \UnaryInfC{$p \vdash p$}
    \AxiomC{}
    \RightLabel{\small Ax}
    \UnaryInfC{$\neg p \vdash \neg p$}
    \RightLabel{\small $\multimap$E}
    \BinaryInfC{$p, \neg p \vdash \bot$}
\end{prooftree}
This deduction shows $p, \neg p \vdash_{\mathcal{B}} \bot$, so, together with $\Vdash^{\Gamma_{\At}}_{\mathcal{B}} p$, we conclude $\Gamma_{\At}, \neg p \vdash_{\mathcal{B}} \bot$ by~\ref{eq:supp-at}. Hence, by applying \textit{reductio ad absurdum}, we obtain $\Gamma_{\At}\vdash_{\mathcal{B}} p$, as desired. \\

\noindent($\Leftarrow$): Assume $\Gamma_{\At}\vdash_{\mathcal{B}} p$. Further assume, for an arbitrary $\mathcal{C} \supseteq \mathcal{B}$ and arbitrary $\Delta_{\At}$, that $\Delta_{\At}, p \vdash_{\mathcal{C}} \bot$. Since deductions are preserved under base extensions, it is also the case that $\Gamma_{\At}\vdash_{\mathcal{C}} p$. We can thus compose $\Delta_{\At}, p \vdash_{\mathcal{C}} \bot$ and $\Gamma_{\At}\vdash_{\mathcal{C}} p$ to obtain $\Delta_{\At}, \Gamma_{\At} \vdash_{\mathcal{C}} \bot$. Since $\mathcal{C} \supseteq \mathcal{B}$ such that $\Delta_{\At}, p \vdash_{\mathcal{C}} \bot$ for arbitrary $\Delta_{\At}$, and $\Delta_{\At}, \Gamma_{\At} \vdash_{\mathcal{C}} \bot$,  by~\ref{eq:supp-at}, $\Vdash^{\Gamma_{\At}}_{\mathcal{B}} p$, as expected.

\item[($\multimap$).] $\phi = \alpha \multimap \beta$.

\noindent$(\Rightarrow)$: Assume $\Vdash^{\Gamma_{\At}}_{\mathcal{B}} \alpha \multimap \beta$. Since $p^{\alpha} \vdash_{\mathcal{B}} p^{\alpha}$, the induction hypothesis yields $\Vdash^{p^{\alpha}}_{\mathcal{B}} \alpha$. Further assume, for an arbitrary $\mathcal{C} \supseteq \mathcal{B}$ and arbitrary $\Delta_{\At}$, that $\Vdash_{\mathcal{C}}^{\Delta_{\At}} \beta$. The induction hypothesis thus yields $\Delta_{\At} \vdash_\mathcal{C} p^{\beta}$. Notice further that $p^{\neg \beta} \vdash_{\mathcal{C}} p^{\neg \beta}$. Then the following is a deduction in $\mathcal{C}$:

\begin{prooftree}
  \def\ScoreOverhang{0.5pt}
      \AxiomC{$\Delta_{\At} \vdash p^\beta$}
      \AxiomC{}
      \RightLabel{\small Ax}
      \UnaryInfC{$p^{\neg \beta} \vdash p^{\neg \beta}$}
      \RightLabel{\small $\multimap$E}
      \BinaryInfC{$\Delta_{\At}, p^{\neg \beta} \vdash\bot$}
\end{prooftree}
This deduction shows $\Delta_{\At}, p^{\neg \beta} \vdash_{\mathcal{C}} \bot$, so by Lemma \ref{lemma:bottomisspecial}, $\Vdash^{\Delta_{\At}, p^{\neg \beta}}_{\mathcal{C}} \bot$. Since $\mathcal{C} \supseteq \mathcal{B}$ such that $\Vdash_{\mathcal{C}}^{\Delta_{\At}} \beta$ for arbitrary $\Delta_{\At}$, and $\Vdash^{\Delta_{\At}, p^{\neg \beta}}_{\mathcal{C}} \bot$, by~\ref{eq:supp-inf} we obtain $\beta \Vdash^{p^{\neg \beta}}_{\mathcal{B}} \bot$. Now, from $\Vdash^{\Gamma_{\At}}_{\mathcal{B}} \alpha \multimap \beta$ and $\Vdash^{p^{\alpha}}_{\mathcal{B}} \alpha$ and $\beta \Vdash^{p^{\neg \beta}}_{\mathcal{B}} \bot$ we obtain $\Vdash^{\Gamma_{\At}, p^{\alpha}, p^{\neg \beta}}_{\mathcal{B}} \bot$ by~\ref{eq:supp-imply}. Hence, $\Gamma_{\At},p^{\alpha}, p^{\neg \beta} \vdash_{\mathcal{B}} \bot$ by Lemma \ref{lemma:bottomisspecial}. Then the following is a deduction in $\mathcal{B}$:
\begin{prooftree}
\def\ScoreOverhang{0.5pt}
    \AxiomC{$\Gamma_{\At}, p^{\alpha}, p^{\neg \beta} \vdash_ \bot$}
    \RightLabel{\small Raa}
    \UnaryInfC{$\Gamma_{\At}, p^{\alpha} \vdash p^{\beta}$}
    \RightLabel{\small $\multimap$I}
    \UnaryInfC{$\Gamma_{\At} \vdash p^{\alpha \multimap \beta}$}
\end{prooftree}
This deduction shows $\Gamma_{\At} \vdash_{\mathcal{B}} p^{\alpha\multimap \beta}$, as expected. \\

\noindent$(\Leftarrow)$: Assume $\Gamma_{\At} \vdash_{\mathcal{B}} p^{\alpha \multimap \beta}$. Further assume that, for an arbitrary $\mathcal{C} \supseteq \mathcal{B}$ and arbitrary $\Delta_{\At}, \Theta_{\At}$, $\Vdash^{\Delta_{\At}}_{\mathcal{C}} \alpha$ and $\beta \Vdash^{\Theta_{\At}}_{\mathcal{C}} \bot$. Induction hypothesis yields $\Delta_{\At} \vdash_{\mathcal{C}} p^{\alpha}$. Notice further that $p^{\beta} \vdash_{\mathcal{C}} p^{\beta}$, hence, by the induction hypothesis, $\Vdash^{p^{\beta}}_{\mathcal{C}} \beta$. Since $\beta \Vdash^{\Theta_{\At}}_{\mathcal{C}} \bot$ and $\Vdash^{p^{\beta}}_{\mathcal{C}} \beta$,  by~\ref{eq:supp-inf}, we obtain $\Vdash^{p^{\beta}, \Theta_{\At}}_{\mathcal{C}} \bot$. Hence, $p^{\beta}, \Theta_{\At} \vdash_{\mathcal{C}} \bot$ by Lemma~\ref{lemma:bottomisspecial}. Since $\mathcal{C} \supseteq \mathcal{B}$, it is also the case that $\Gamma_{\At} \vdash_{\mathcal{C}} p^{\alpha \multimap \beta}$, so the following is a deduction in $\mathcal{C}$:

\begin{prooftree}
  \def\ScoreOverhang{0.5pt}
      \AxiomC{$\Gamma_{\At} \vdash p^{\alpha \multimap \beta}$}
      \AxiomC{$\Delta_{\At} \vdash p^{\alpha}$}
      \RightLabel{\small $\multimap$E}
      \BinaryInfC{$\Gamma_{\At}, \Delta_{\At} \vdash p^\beta$}
      \AxiomC{$\Theta_{\At}, p^{\beta} \vdash \bot$}
    \RightLabel{\small Subs}
      \BinaryInfC{$\Gamma_{\At}, \Delta_{\At}, \Theta_{\At} \vdash \bot$}
\end{prooftree}
This deduction shows $\Gamma_{\At}, \Delta_{\At}, \Theta_{\At} \vdash_{\mathcal{C}} \bot$, so, by Lemma~\ref{lemma:bottomisspecial}, we conclude $\Vdash^{\Gamma_{\At}, \Delta_{\At}, \Theta_{\At}}_{\mathcal{C}} \bot$. Since $\mathcal{C} \supseteq \mathcal{B}$ such that $\Vdash^{\Delta_{\At}}_{\mathcal{C}} \alpha$ and $\beta \Vdash^{\Theta_{\At}}_{\mathcal{C}} \bot$ for arbitrary $\Delta_{\At}, \Theta_{\At}$, and $\Vdash^{\Gamma_{\At}, \Delta_{\At}, \Theta_{\At}}_{\mathcal{C}} \bot$, by~\ref{eq:supp-imply}, $\Vdash^{\Gamma_{\At}}_{\mathcal{B}} \alpha \multimap \beta$, as expected.

\end{description}
\end{proof}

\begin{theorem}[Completeness] \label{thm:completeness}
    If $\Gamma \Vdash \phi$ then $\Gamma \vdash_{\MALL} \phi$.
\end{theorem}

\begin{proof}
    Assume $\Gamma \Vdash \phi$. Let $\Gamma = \{\alpha^1,\dots,\alpha^n\}$ and let $\mathcal{U}$ be a simulation base for $\Gamma \cup \phi$ and some atomic mapping $\sigma$. By Definition~\ref{def:validity}, $\Gamma \Vdash_{\mathcal{U}} \phi$. Let $p^{\alpha^i}$ be an atom representing $\alpha^i$ $(1 \leq i \leq n)$. Since $p^{\alpha^i} \vdash_{\mathcal{U}} p^{\alpha^i}$, by Lemma~\ref{lemma:completeness} we have that $\Vdash^{p^{\alpha^i}}_{\mathcal{U}} \alpha^i$ for all $\alpha^i \in \Gamma$. Then, by~\ref{eq:supp-inf}, we obtain $\Vdash^{p^{\alpha^1},\dots,p^{\alpha^n}}_{\mathcal{U}} \phi$. By Lemma~\ref{lemma:completeness} again, $p^{\alpha^1},\dots,p^{\alpha^n} \vdash_{\mathcal{U}} p^{\phi}$. Denote $\Gamma_{\At} \coloneq \{p^{\alpha^1},\dots,p^{\alpha^n}\}$; then $\Gamma_{\At} \vdash_{\mathcal{U}} p^{\phi}$. Since the rules in $\mathcal{U}$ precisely correspond to the natural deduction rules of $\MALL$ (plus the admissible substitution rules), we can rewrite each atom $p^{\psi}$ in the derivation of $\Gamma_{\At} \vdash_{\mathcal{U}} p^{\phi}$ as $\psi$, hence obtaining the derivation $\Gamma \vdash \phi$ as desired.
\end{proof}

As previously remarked, our completeness proof requires only a notion of simulation base in which applications of $\otimes E$, $\oplus E$, $0 E$ and $\parr E$ have minor premises with shape $\bot$. This also yields a purely semantic proof of the following proof-theoretic property:

\begin{corollary}
    If $\Gamma \vdash_{\MALL} \phi$, then there is a $\MALL$-derivation of $\Gamma \vdash \phi$ in which all applications of $\otimes E$, $\oplus E$, $0 E$ and $\parr E$ have minor premises with shape $\bot$.
\end{corollary}

\begin{proof}
    Assume $\Gamma \vdash \phi$. By soundness we conclude $\Gamma \Vdash \phi$. Let $\mathcal{U}$ be a simulation base for $\Gamma \cup \phi$ and some atomic mapping $\sigma$. Theorem \ref{thm:completeness} yields a derivation $\pi$ showing $\Gamma \vdash \phi$. Since $\pi$ was obtained by rewriting atoms $p^{\phi}$ as $\phi$ in a derivation of $\mathcal{U}$ and all instances of $\otimes E$, $\oplus E$, $0 E$ and $\parr E$ in $\mathcal{U}$ have minor premises with shape $\bot$, a straightforward induction on the length of $\pi$ shows that it has the desired property.
\end{proof}

The same holds for every logic sound and complete with respect to similar classical semantics. This property, which is sometimes used in classical normalisation proofs (see, for instance, Definition 3.8 and Lemma 3.16 of \cite{Medeiros_2006}, as well as the proof corrections in \cite{LopesEnglanderLoboCruz}), highlights important features of the interaction between classical negation and classical disjunctions. Perhaps more importantly, the proof of the corollary is purely semantic and does not require any reduction procedures or similar techniques, showing once again that in $\Bes$ and $\Pts$ it is possible both to prove semantic results through syntactic means and syntactic results through semantic means.

\section{Concluding Remarks}

Switching from the truth-centered model-theoretic paradigm to the demonstrability-centered proof-theoretic paradigm yields a semantic framework perfectly suited for intuitionistic logics. This raises the question of what would a general proof-theoretic account of classical semantics look like. This paper provides an illuminating answer in terms of a simple characterisation of classical linear logic through $\Bes$. Furthermore, even though our results are proven only for (a fragment of) classical linear logic, the structure of proofs suggests that our methods are fully general and may also be applied to the intuitionistic version of other logics, resulting in a similar semantics for their classical version.


The inner workings of our characterisation also bring to light some important conceptual insights. Classical proof semantics can be derived by applying mild restrictions to the semantic clauses of intuitionistic proof semantics--which themselves are simply explicit descriptions of what qualifies as an intuitionistic proof for each logical connective. This suggests a natural conclusion: the very notion of a classical proof might be seen as a {\em restriction} of the constructive concept of proof, or conversely, that the constructive proof concept is a {\em generalisation} of the classical concept of proof. In this sense, a classical proof still carries constructive content but requires significantly less information to be established.


This provides an intuitive justification for results showing that weakened algorithmic content can consistently be extracted from classical proofs~\cite{Constable_Murthy_1991,Murthy1990ExtractingCC}. It also sheds light on how the classical version of linear logic can be seen as constructive~\cite{DBLP:journals/mscs/Girard91,DBLP:journals/apal/Girard93}, even in the presence of rules incorporating the \textit{reductio ad absurdum} principle, and despite the existence of intuitionistic linear logic. The absence of structural rules in the calculus increases the informational content required to establish proofs to such an extent that their algorithmic interpretations become robust enough to be considered constructive -- regardless of the use of \textit{reductio ad absurdum}. 

This, of course, does not prevent intuitionistic linear logic from being \textit{even more constructive} than its classical counterpart, since it demands \textit{even more information} to establish a proof. All this suggests that the difference between classical and constructive proofs is best understood as {\em quantitative}, rather than {\em qualitative}. It also supports the idea that constructivity is not a binary property but rather a \textit{spectrum} of informational requirements for proof construction.



Our characterisation of $\parr$ and $\with$ also sheds further light on the relation between classical and intuitionistic interpretations of connectives. All other logical operators are obtained after restricting their natural proof conditions for intuitionistic logic, but the classical proof conditions for $\parr$ and $\with$ can be read directly from their standard introduction and elimination rules. The reason for this seems to be wholly different in the two cases. The proof conditions for $\parr$ are extracted directly from the rules because the rules themselves seem to be \textit{essentially classical}, in the sense that they already express the restricted proof conditions we expect to see in classical logic. The fact that the inductive step for $\parr$ in the completeness result does not require applications of \textit{reductio ad absurdum} evidentiates this. This also explains why $\parr$ is often viewed as an essentially classical connective  usually absent in formulations of intuitionistic linear logic \cite{Girard10.1007/BFb0014972,HYLAND1993273}. Interestingly, it is straightforward to present an intuitionistic version of $\parr$ by writing down a unrestricted version of the classical semantic clause:

\begin{itemize}[itemsep=0.5em, font=\normalfont] \label{parrAltInt}
    \item[\namedlabel{eq:supp-parrAlternativeInt}{($\parr$ Int)}]  $\Vdash^{\Gamma_{\At}}_{\mathcal{B}} \phi \parr \psi$ iff, for all $\mathcal{C} \supseteq \mathcal{B}$, all $p \in \At$ and all $\Delta_{\At}, \Theta_{\At}$, if $\phi \Vdash^{\Delta_{\At}}_{\mathcal{C}} p$ and $\psi \Vdash^{\Theta_{\At}}_{\mathcal{C}} p$ then $\Vdash^{\Gamma_{\At}, \Delta_{\At}, \Theta_{\At}}_{\mathcal{C}} p$;
\end{itemize}
\noindent
which is very sensible for a definition of multiplicative disjunction. It is not immediately clear, however, how those proof conditions would translate into an actual proof system. In fact, even though there are proof systems for intuitionistic linear logic with $\parr$~\cite{val05}, it is generally not easy to define intuitionistic version of $\parr$ with desirable properties such as cut elimination \cite{Schellinx1991SomeSO}. Since an investigation of the properties of such a definition of $\parr$ is entirely outside the scope of this paper, this is left for future work.

On the other hand, the classical clause for $\with$ is perfectly acceptable from an intuitionistic viewpoint, which means that intuitionistic and classical logic actually \textit{share} the proof conditions for $\with$. The claim that different logics might share proof conditions for connectives figures prominently in the literature on logical ecumenism \cite{Pimentel2023ATO,DBLP:journals/Prawitz15} and was predated by a result of G\"odel showing that classical and intuitionist logic coincide w.r.t. derivability in the fragment containing only conjunction and negation \cite{RodolfoAlbareli,Godel10.1093/oso/9780195147209.003.0063}.  Since such claims are usually formulated in syntactic frameworks, our results add to the arguments to that effect by showing that this is also reflected on the semantic level.

We conclude this paper by discussing the extension of $\Bes$ to full linear logic. It is well known that incorporating exponential modalities significantly increases the complexity of the semantic analysis of $\LL$-- for example, the categorical interpretation of exponentials has been a longstanding subject of debate (see~\cite{ll-mellies}).

In the case of $\Bes$, the following semantic clause for the bang modality in intuitionistic linear logic has been proposed in~\cite{DBLP:journals/corr/abs-2402-01982}:

\begin{itemize}[itemsep=0.5em, font=\normalfont] \label{bang} \item[\namedlabel{eq:bang}{($!$ Int)}]
$\Vdash^{\Gamma_{\At}}_{\mathcal{B}} !\phi$ iff, for all $\mathcal{C} \supseteq \mathcal{B}$, all $p \in \At$ and all $\Delta_{\At}$,
if for all $\mathcal{D} \supseteq \mathcal{C}$,
$\Vdash^{\varnothing}_{\mathcal{D}} \phi$ implies $\Vdash^{\Delta_{\At}}_{\mathcal{D}} p$,
then $\Vdash^{\Gamma_{\At}, \Delta_{\At}}_{\mathcal{C}} p$; \end{itemize}
Intuitively, this clause asserts that $!\phi$ is valid relative to the multiset $\Gamma_{\At}$ if and only if anything derivable from $\phi$ without consuming {\em any} resources is also valid relative to $\Gamma_{\At}$. This goes well along with the intended meaning of $!$, specially when read from its introduction rule/promotion in natural deduction/sequent calculus.

Given the structural constraints of our semantics, it is natural to expect that, in the classical setting, the corresponding clause for the bang modality should take the following form: 
\begin{itemize}[itemsep=0.5em, font=\normalfont] \label{cbang} \item[\namedlabel{eq:cbang}{($!$)}]
$\Vdash^{\Gamma_{\At}}_{\mathcal{B}} !\phi$ iff, for all $\mathcal{C} \supseteq \mathcal{B}$ and all $\Delta_{\At}$,
if for all $\mathcal{D} \supseteq \mathcal{C}$,
$\Vdash^{\varnothing}_{\mathcal{D}} \phi$ implies $\Vdash^{\Delta_{\At}}_{\mathcal{D}} \red{\bot}$,
then $\Vdash^{\Gamma_{\At}, \Delta_{\At}}_{\mathcal{C}} \red{\bot}$. \end{itemize}

As for the dual exponential $?$, we conjecture that the following clause provides a sound interpretation:\begin{itemize}[itemsep=0.5em, font=\normalfont] \label{cquest} \item[\namedlabel{eq:cquest}{($?$)}]
$\Vdash^{\Gamma_{\At}}_{\mathcal{B}} ?\phi$ iff, for all $\mathcal{C} \supseteq \mathcal{B}$ and all $\Delta_{\At}$,
if for all $\mathcal{D} \supseteq \mathcal{C}$,
$\Vdash^{\Delta_{\At}}_{\mathcal{D}} \phi$ implies $\Vdash^{\Delta_{\At}}_{\mathcal{D}} \red{\bot}$,
then $\Vdash^{\Gamma_{\At}, \Delta_{\At}}_{\mathcal{C}} \red{\bot}$. \end{itemize}
The idea is to match the elimination clause for $?$, in which if anything derivable from $\phi$ without consuming {\em any  extra} resources is also valid relative to $\Gamma_{\At}$.

Exploring whether these clauses indeed yield sound interpretations of $!$ and $?$ in full classical linear logic $\LL$ is an interesting direction for future work, which we intend to pursue next.

\newpage

\bibliographystyle{entics}


%

\end{document}